\documentclass{lmcs}
\pdfoutput=1

\usepackage[utf8]{inputenc}

\usepackage{booktabs}

\usepackage{lastpage}
\lmcsdoi{18}{3}{3}
\lmcsheading{}{\pageref{LastPage}}{}{}%
{Dec.~14,~2020}{Aug.~02,~2022}{}

\usepackage{xspace,mathpartir,xparse,hyperref,multirow}
\usepackage{cmll,ebmath,ebutf8}

\newlist{enumerati}{enumerate}{10}
\setlist[enumerati]{label=\emph{(\roman*)}, ref=\emph{(\roman*)}}
\newlist{enumerata}{enumerate}{10}
\setlist[enumerata]{label=\emph{(\alph*)}, ref=\emph{(\alph*)}}

\usepackage{tikz-cats}
\usetikzlibrary{matrix,cd}
\tikzcdset{arrow style=tikz}
\tikzcdset{every arrow/.append style = -{Computer Modern Rightarrow[]}}
\tikzset{every arrow/.append style = -{Computer Modern Rightarrow[]}}

\newtheorem{terminology}[thm]{Terminology}
\theoremstyle{defC}
\newtheorem{exaC}[thm]{Example}
\newtheorem{remC}[thm]{Remark}

\newcommand{\xto}[1]{\xrightarrow{#1}}
\newcommand{\xot}[1]{\xleftarrow{#1}}
\newcommand{\nxto}[1]{\xarrow[negate,->]{#1}}

\newcommand{\ens}[1]{\{ #1 \}}
\newcommand{\op}[1]{{#1}^{\mathit{op}}}

\newcommand{\src}{\mathit{src}}
\newcommand{\tgt}{\mathit{tgt}}

\renewcommand{\pmod}{^{\mathit{stmod}}}
\newcommand{\pmodmon}{^{\mathit{monmod}}}

\DeclareMathOperator{\alg}{-alg}
\DeclareMathOperator{\Mod}{-\mathbf{Mod}}

\DeclareMathOperator{\Mon}{-\mathbf{Mon}}

\DeclareMathOperator{\id}{id}
\DeclareMathOperator{\Id}{Id}
\DeclareMathOperator{\colim}{colim}

\DeclareMathOperator{\app}{app}
\DeclareMathOperator{\ope}{op}
\DeclareMathOperator{\val}{val}
\DeclareMathOperator{\paral}{par}
\DeclareMathOperator{\arr}{arr}

\DeclareMathOperator{\ev}{ev}

\newcommand{\Gph}{\text{-}𝐆𝐩𝐡}

\newcommand{\specfun}{\mathit{spec}}
\newcommand{\spec}[1]{\specfun(#1)}

\newcommand\Tstrut{\rule{0pt}{3ex}}
\newcommand\Bstrut{\rule[-0.9ex]{0pt}{0pt}}
\newcommand\Bstrutbas{\rule[-1.2ex]{0pt}{0pt}}
\newcommand\Tstruthaut{\rule{0pt}{2em}}
\newcommand\Tstrutpetit{\rule{0pt}{1em}}
\renewcommand{\epsilon}{\varepsilon}
\renewcommand{\phi}{\varphi}
\newcommand{\abar}{\bar{a}}

\DeclareTextMath{\lam}[lambda]{λ}
\DeclareTextMath{\lbmu}[lambda-bar-mu]{\bar{λ}μ}
\DeclareTextMath{\pii}[pi]{π}

\newcommand{\transition}{transition\xspace}
\newcommand{\Transition}{Transition\xspace}
\newcommand{\atransition}{a transition\xspace}

\newcommand{\Red}{\mathit{Trans}}

\newcommand{\semsig}{abstract signature\xspace}

\newcommand{\semsigs}{abstract signatures\xspace}

\newcommand{\asemsig}{an abstract signature\xspace}
\newcommand{\Asemsig}{An abstract signature\xspace}

\newcommand{\SemSig}{𝐒𝐞𝐦𝐒𝐢𝐠}

\newcommand{\labredrule}[3]{\ensuremath{{#2} \xto{#1} {#3}}}
\newcommand{\alert}[1]{{\bf #1}}

\newcommand{\mypar}[1]{\medskip \noindent\textbf{#1}\ }


\begin{document}

\title{Modules over monads and operational semantics (expanded
  version)}

\titlecomment{An extended abstract of this paper appeared in FSCD
  '20~\cite{DBLP:conf/fscd/HirschowitzHL20}. The present version
  incorporates several improvements and additions, listed at the end
  of~§\ref{s:intro}.}

\thanks{We thank the anonymous reviewers for their careful reading and
  constructive comments.}

\author[A.\ Hirschowitz]{André Hirschowitz\lmcsorcid{0000-0003-2523-1481}}[a]
\author[T.\ Hirschowitz]{Tom Hirschowitz\lmcsorcid{0000-0002-7220-4067}}[b]
\author[A.\ Lafont]{Ambroise Lafont\lmcsorcid{0000-0002-9299-641X}}[c]

\address{Univ. Côte d’Azur, CNRS, LJAD, 06103, France}
\address{Univ. Grenoble Alpes, Univ. Savoie Mont Blanc, CNRS, LAMA, 73000, Chambéry, France}
\address{UNSW, Sydney, Australia}

\begin{abstract}
  This paper is a contribution to the search for efficient and
  high-level mathematical tools to specify and reason about (abstract)
  programming languages or calculi.  Generalising the reduction monads
  of Ahrens et al., we introduce \transition monads, thus covering new
  applications such as $\overline{λ}μ$-calculus, $π$-calculus,
  Positive GSOS specifications, differential $λ$-calculus, and the
  simply-typed, call-by-value $λ$-calculus.  Moreover, we design a
  suitable notion of signature for \transition monads.
\end{abstract}

\maketitle

\section{Introduction}\label{s:intro}
The search for a mathematical notion of programming language goes back
at least to Turi and Plotkin~\cite{plotkin:turi:bialgebraic}, who
coined the name “Mathematical Operational Semantics”, and explained
how known classes of well-behaved rules for structural operational
semantics, such as GSOS~\cite{GSOS}, can be categorically understood
and specified via bialgebras and distributive laws. Their initial
framework did not cover variable binding, and several authors have
proposed variants which
do~\cite{DBLPconf/lics/FioreT01,DBLPconf/lics/FioreS06,DBLPconf/lics/Staton08},
treating examples like the $π$-calculus.  However, none of these
approaches covers higher-order languages like the $λ$-calculus.

In recent work, following previous work on modules over monads for
syntax with binding~\cite{hirscho:lam,AHLM:2sig} (see
also~\cite{Ahrens16}), Ahrens et al.~\cite{AHLM19} introduce
\alert{reduction monads}, and show how they cover several standard
variants of the $λ$-calculus. Furthermore, as expected in similar
contexts, they propose a mechanism for specifying reduction monads by
suitable signatures.

Our starting point is the fact that already the call-by-value
$λ$-calculus does not form a reduction monad. Indeed, in this
calculus, variables are placeholders for values but not for general
$λ$-terms; in other words, reduction, although it involves general
terms, is stable under substitution by values only.  In the present
work, we generalise reduction monads to what we call
\alert{\transition monads}.  The main new ingredients of our
generalisation are as follows.
\begin{itemize}
\item We now have two kinds of terms, called \alert{placetakers} and
  \alert{states}: variables are placeholders for our placetakers, while
  transitions relate states. Typically, in the call-by-value, small-step
  $λ$-calculus, placetakers are values, while states are general
  terms.
\item We also have a set of types for placetakers, and a possibly
  different set of types for transitions and states: transitions of a
  given type relate states of this type.  Typically, in the
  simply-typed, call-by-value $λ$-calculus, both sets of types
  coincide and are given by simple types; while in pure \lbmu-calculus, we
  have two placetaker types, one for programs and one for stacks, and
  three transition types, respectively for programs, stacks, and
  commands.
\item We in fact have two possibly different kinds of states for each 
transition type,
  \alert{source} states and \alert{target} states, so that a transition
  of a given type 
  now relates a source state to a target state of this type. Typically, in the
  untyped, call-by-value, big-step $λ$-calculus, source states (of the unique
  transition type)  are general
  terms, while target states are values.
\item 
Our variables form  a (variable!) family of sets indexed by the placetaker types.
To such a variables family $X$, a \transition monad assigns 
\begin{itemize}
\item 
an object $ T(X)$ (`of placetakers with free
  variables in $X$'), which is again a 
  family of sets indexed by the placetaker types, and
  \item 
  two state objects $S₁(T(X))$ and $S₂(T(X))$ (`of source  (resp.\ target) states with free
  variables in $X$'), which are  
  families of sets indexed by transition types;
  here $S₁$ and $S₂$ are two functors producing state objects out of placetaker objects. 
\end{itemize}
\end{itemize}

Roughly speaking (see~§\ref{sec:operational-monads}),  
 a \transition monad  consists of three components:
\begin{itemize}
\item a \alert{placetaker} monad $T$,
\item two \alert{state} functors $S₁,S₂$, 
\item a \alert{\transition structure} consisting of a 
  $T$-module $R$ and two $T$-module morphisms
  \begin{center}
    $\src∶ R → S₁T$ \qquad and \qquad $\tgt∶ R → S₂T$,
  \end{center}
\end{itemize}
where $T$-modules~\cite{hirscho:lam} are objects equipped with
substitution by elements of $T$.  We view the transition structure
$(R,\src,\tgt)$ as an object of the slice category $T\Mod_f/S₁T×S₂T$ of
(finitary) $T$-modules over $S₁ T×S₂ T$.

Reduction monads~\cite{AHLM19} correspond to the untyped case with 
$S₁ = S₂ = \Id_{𝐒𝐞𝐭}$.
There, reduction monads are identified with suitable 
 \alert{relative
  monads}~\cite{DBLP:journals/corr/AltenkirchCU14}, and 
  we provide a similar identification for \transition monads 
(see Proposition~\ref{prop:synthana}).

We present our series of examples of \transition monads
in~§\ref{s:examples}: $\overline{λ}μ$-calculus, simply-typed
$λ$-calculus (in its call-by-value, big-step variant), $π$-calculus
(as an unlabelled
transition system), positive GSOS systems, and differential $λ$-calculus.\\

After defining \transition monads, we embark on a second part of the
paper, devoted to offering the operational semanticist some hopefully
convenient tools for defining programming languages (as \transition
monads).

Specifically, we propose an approach to the specification of
\transition monads via signatures, which follows the spirit of Initial Algebra
Semantics~\cite{adj}.  This approach is thus categorical in nature,
hence we have to upgrade our sets of \transition monads into
categories, say $𝐓𝐫𝐚𝐧𝐬𝐌𝐧𝐝(ℙ,𝕊)$, one for each pair $(ℙ,𝕊)$ of sets of
placetaker types and transition types (see
Definition~§\ref{def:omndof}). For these categories, we propose what
we call a \alert{register}: for a category $𝐂$, a register consists of
\begin {enumerate}
\item a set $𝐒𝐢𝐠$ of \alert{signatures},
\item a \alert{semantics} map $⟦-⟧$ assigning to each signature $S$ in
  $𝐒𝐢𝐠$ a category $S\alg$ of \alert{algebras}, or \alert{models},
  together with a \alert{forgetful} \footnote{Here ``algebra'' and
    ``forgetful'' have no technical meaning and are chosen by
    analogy.}  functor $U∶ S\alg → 𝐂$, and
\item a \alert{validity} proof of the fact that for each $S$,
the category $S\alg$ 
has an initial object.
\end{enumerate}
\begin{rem}
  This definition is written in type-theoretic style, in the sense
  that the validity proof is treated as a proper mathematical object.
  The reader working in a standard, set-theoretical logical setting
  should of course understand this as an additional condition that
  registers must satisfy.  
\end{rem}
A register for $𝐂$ yields a  ``decoding'' map $\specfun$ sending $S$
to $\spec{S} ≔ U(0)$, where $0$ here denotes the initial
object in $S\alg$. The efficiency of a register for a category lies in the fact 
that it allows the expert to easily formalise the informal specification
they have in mind for the relevant object.
We illustrate in~§\ref{s:applications} the expressiveness of
our register $𝐑𝐞𝐠𝐓𝐫𝐚𝐧𝐬𝐌𝐧𝐝_{ℙ,𝕊}$ for transition monads by designing signatures
for our examples of~§\ref{s:examples}.  For Positive GSOS systems, 
we even go further by
defining a specific register, in which each individual system is more
easily specified. 

Our register for \transition monads is built out from three
intermediate registers, corresponding to the three components listed
above:
\begin{itemize}
\item a register $𝐑𝐞𝐠𝐌𝐧𝐝_f(𝐒𝐞𝐭^ℙ)$  for monads in the category $𝐒𝐞𝐭^ℙ$,
\item a register $𝐑𝐞𝐠[𝐒𝐞𝐭^ℙ,𝐒𝐞𝐭^𝕊]_f$ for functors $𝐒𝐞𝐭^ℙ → 𝐒𝐞𝐭^𝕊$, and
\item a register $𝐑𝐞𝐠𝐓𝐫𝐚𝐧𝐬𝐒𝐭𝐫𝐮𝐜𝐭_{ℙ,𝕊}(T,S₁,S₂)$ for transition
  structures over $(T,S₁,S₂)$.
\end{itemize}
And as could be expected, a signature for a \transition monad in our
register $𝐑𝐞𝐠𝐓𝐫𝐚𝐧𝐬𝐌𝐧𝐝_{ℙ,𝕊}$ is a record consisting of
\begin{itemize}
\item a signature of $𝐑𝐞𝐠𝐌𝐧𝐝_f(𝐒𝐞𝐭^ℙ)$ specifying the placetaker monad $T$,
\item two signatures of $𝐑𝐞𝐠[𝐒𝐞𝐭^ℙ,𝐒𝐞𝐭^𝕊]_f$ specifying the state
functors  $S₁$ and $S₂$, and
\item a signature of $𝐑𝐞𝐠𝐓𝐫𝐚𝐧𝐬𝐒𝐭𝐫𝐮𝐜𝐭_{ℙ,𝕊}(T,S₁,S₂)$ specifying the
  transition structure.
\end{itemize}

Here, a crucial feature is that the involved register in the last
field of this record depends on the objects specified by previous
fields. This contrasts with the approach taken in~\cite{AHLM19}, where
the corresponding field cannot be expressed in terms of the specified
placetaker monad, and must instead be parametric in the models of
previous fields.  Our choice allows more signatures, and we take
advantage of this subtle fact in our treatment
(see~§\ref{ss:difftrans}) of Ehrhard and Regnier's differential
\lam-calculus~\cite{LambdaDiff}.  Let us mention that the counterpart
for this advantage is that the recursion principle induced by our
initial semantics is significantly weaker than could be expected (see
Remark~\ref{rk:induction} and~§\ref{s:conclu}).

Our signatures for functors and monads incorporate equations similar
to those of equational systems: a signature without equations
specifies a kind of free object, and adding equations specifies a kind
of quotient of this free object, obtained by somehow forcing the added
equations to hold.  Such a quotienting procedure has already been
achieved in a fairly general (and elaborate) way by Fiore and
Hur~\cite{FioreHurEquational}. Under mild additional hypotheses on the
ambient category and on equations, we obtain a compact
``free+quotient'' description of our initial models
(Theorems~\ref{thm:modmnd} and~\ref{prop:statefuns}).

This description roughly goes as follows: we consider a
signature augmented with one formal operation for each equation, which
yields a new set of terms, say \alert{augmented} terms.  By
interpreting the new operations as prescribed by each side of the
equations, we obtain two translations, say $L$ and $R$, from augmented
terms to plain terms. The initial model is then obtained by
identifying any two plain terms of the form $L(a)$ and $R(a)$, for
some augmented term $a$.

This description is a crucial ingredient of our treatment of the
differential \lam-calculus: as already mentioned, the type of the
third field of its signature depends in particular on the monad
specified by the first field, and the construction of the signature
for this third field relies on our explicit description for that monad
(see the proof of Lemma~\ref{lem:deriv}).

Along the elaboration of our registers, we strive to offer the
operational semanticist intuitive notation for defining signatures.
They may thus specify programming languages as \transition monads by
the following procedure:
\begin{itemize}
\item Fix sets $ℙ$ and $𝕊$ of placetaker and transition types.
\item Pick a signature of $𝐑𝐞𝐠𝐌𝐧𝐝_f(𝐒𝐞𝐭^ℙ)$ for specifying the
  placetaker monad $T$, by operations and equations (e.g., using
  Notation~\ref{not:monadformat}).
\item Choose two signatures of $𝐑𝐞𝐠[𝐒𝐞𝐭^ℙ,𝐒𝐞𝐭^𝕊]_f$ for specifying the
  state functors $S₁$ and $S₂$, again by operations and equations
  (e.g., using Notation~\ref{not:endoformat}).
\item Choose a signature of $𝐑𝐞𝐠𝐓𝐫𝐚𝐧𝐬𝐒𝐭𝐫𝐮𝐜𝐭_{ℙ,𝕊}(T,S₁,S₂)$ for
  specifying the transition structure, e.g., by giving some rules
  using the notation of~§\ref{sss:formatrules}.
\end{itemize}

\subsection*{Plan}
In~§\ref{s:preliminaries}, we present our notations and give some
categorical preliminaries.  In~§\ref{sec:operational-monads} we define
\transition monads and in~§\ref{s:examples}, we present our selected
examples (in the traditional way).  Then, in~§\ref{s:sigtrans}, we
introduce registers and define our register for \transition monads,
deferring the precise definition of its components to the next
sections.  We then introduce our registers for transition structures
(§\ref{s:slices}), monads (§\ref{ss:regmonads}), and functors
(§\ref{s:regfun}), stating the announced explicit descriptions of
initial algebras along the way.  Proofs are dealt with
in~§\ref{sec:gen-signatures} and~§\ref{s:friendliness}.
In~§\ref{s:applications}, we then revisit all examples
from~§\ref{s:examples}, specifying them through signatures of our
registers $𝐑𝐞𝐠𝐓𝐫𝐚𝐧𝐬𝐌𝐧𝐝_{ℙ,𝕊}$ for transition monads.
Finally, we conclude in~§\ref{s:conclu} by summing up our
contributions, assessing the scope of our registers $𝐑𝐞𝐠𝐓𝐫𝐚𝐧𝐬𝐌𝐧𝐝_{ℙ,𝕊}$, and
giving some perspectives.

\begin{rem}
  The reader only interested in using our framework may safely
  skip~§\ref{sec:gen-signatures} and~§\ref{s:friendliness}.  
\end{rem}

\subsection*{Related work}
As mentioned above, our work refines a recent work~\cite{AHLM19},
allowing to cover many new applications with a very similar approach.
This approach differs from the bialgebraic one
introduced long ago by Plotkin and
Turi~\cite{plotkin:turi:bialgebraic}: the positive difference is that
it covers higher-order languages like the $λ$-calculus, while a 
negative difference is that it does not recover congruence of
bisimilarity.

Regarding syntax with variable binding, we model it using monads,
following a standard approach going back to Bellegarde and
Hook~\cite{DBLP:journals/scp/BellegardeH94} (but see
also~\cite{DBLP:conf/csl/AltenkirchR99,DBLP:journals/iandc/HirschowitzM10,AHLM:2sig}).
Because the monad-based approach is essentially equivalent to the
presheaf-based one~\cite{fiore:presheaf,DBLP:conf/lics/Fiore08}, we
anticipate that our whole framework could be straightforwardly adapted
to presheaf models.  We are more cautious about nominal
sets~\cite{PittsAM:newaas}, mainly because we would need a better
understanding of the status of substitution in the latter approach
(see, e.g., Power~\cite{DBLP:journals/entcs/Power07}).

Furthermore, our main register for monads (§\ref{ss:regmonads}) is a
simply-typed refinement of a known one~\cite{AHLM:2sig}, close in
spirit to~\cite{Ahrens16}.  Its validity proof relies on another, new
register (§\ref{sec:gen-signatures}), which combines two standard
registers from the literature, respectively based on equational
systems~\cite{FioreHurEquational} and pointed strong
endofunctors~\cite{fiore:presheaf}.  We could in fact define a
register based on equational systems, and use it directly for
specifying monads. Moreover, our explicit description of initial
models applies in particular to equational systems
(Theorem~\ref{thm:ES}).

Concerning our registers for slice categories, we suspect that they
could be understood in terms of polynomial functors~\cite{KockPoly},
which were recently used in a similar context by Arkor and
Fiore~\cite{DBLP:conf/lics/ArkorF20}.

We are aware of a few notions of signatures for languages with
variable binding equipped with some notion of evaluation (or
transition)~\cite{Hamana03,DBLP:conf/aplas/Hamana04,HIRSCHOWITZ:2010:HAL-00540205:2,Ahrens16}. They
essentially model rewriting, which implies that transitions are closed
under arbitrary contexts. Such approaches cannot cover languages like
the $π$-calculus, in which transitions may not occur under inputs or
outputs.

As is well-known,  evaluation is a directed variant of equality, and 
 transition proofs are a directed form of identity proofs.
For this reason, notions of signatures for dependently-typed languages like type
theory may provide an alternative approach to the specification of
operational semantics systems.  Examples of such notions of signatures
include Fiore's simple dependently-sorted
signatures~\cite{DBLP:conf/lics/Fiore08}, Altenkirch et al.'s
categorical semantics of inductive-inductive
definitions~\cite{DBLP:conf/calco/AltenkirchMFS11}, and
Garner's combinatorial approach~\cite{DBLP:journals/corr/Garner14}.

Finally, let us mention that a preliminary version of the present work
appears in the third author's PhD thesis~\cite[Chapter 6]{AL19}.

\subsection*{Differences with conference version}
Here is a list of significant differences w.r.t.\ the conference
version, beyond more detailed proofs.
\begin{itemize}
\item Most importantly, we provide an explicit description of the initial
  algebras of signatures involving equations, in~§\ref{s:friendliness}.
\item We also correct a few errors, notably:
  \begin{itemize}
  \item We had omitted the congruence rules for reduction in
    \lbmu-calculus and differential $λ$-calculus. 
  \item We had omitted the syntactic equation
    $D (D e · f) · g ≡ D (D e · g) · f$ from our definition of
    differential $λ$-calculus.
  \item Again about differential $λ$-calculus, we had also erroneously
    claimed that the method we used for defining unary multiterm
    substitution applies to partial differentiation.  We now rely on
    the explicit descriptions of~§\ref{s:friendliness} for both
    operations.
\end{itemize}
\item Finally, we include a few minor improvements, e.g.:
  \begin{itemize}
  \item We design an abstract version of the original register for
    slice module categories. In passing, the new version is slightly
    more expressive.
\item We design a new register combining the features of equational
  systems and pointed strong endofunctors into \alert{monoidal
    equational systems} (§\ref{ss:MES}).
\item We provide an alternative way in which to organise the
  $π$-calculus as \atransition monad.
  \end{itemize}
\end{itemize}

\section{Notations and categorical preliminaries}
\label{s:preliminaries}
In this section, we fix some notation, and recall a few categorical
notions. We advise the reader to skip it, except perhaps
for~§\ref{ss:notation}, and get back to it when needed.
In~§\ref{ss:notation}, we fix some basic notation.
In~§\ref{ss:finitely-pres-cat}, we recall some well-known results
about locally finitely presentable categories.  Then,
in~§\ref{ss:convenient-monoidal}, we introduce a notion of
\emph{convenient} monoidal category. This is important for us because
we will use finitary monads a lot, and these are monoids for the
composition monoidal structure on the category of finitary
endofunctors. We show in particular that this composition monoidal
structure is convenient.  We go on in~§\ref{ss:modules} by recalling
the standard notion of module over a monoid in a monoidal category,
and introduce a notion of parametric module, roughly a module
definable for all monoids.  We then recall a slightly more general
definition of (parametric) modules in the particular case of modules
over monads (i.e., when the base category is a functor
category). In~§\ref{ss:creation}, after briefly recalling Beck's
monadicity theorem, we state one of its useful (folklore)
consequences, a ``cancellation'' property for monadic functors.
In~§\ref{ss:street}, we recall the well-known correspondence between
finitary monadic functors and finitary monads.  Finally,
in~§\ref{ss:monadic}, we recall the important fact that equalisers of
(finitary) monadic functors are computed as in $𝐂𝐀𝐓$.

\subsection{Basic notation}\label{ss:notation}
In the following, $𝐒𝐞𝐭$ denotes the category of sets,
and $𝐂𝐀𝐓$ denotes the (very large) category of locally small
categories.  We often implicitly view natural numbers $n$ as intervals
$\ens{1,…,n}$ in $ℕ$, so that, e.g., $2 ∈ 3$, $2 ⊆ 3$, and so on.
Furthermore, in any category with binary products, we denote by
$⟨f,g⟩∶ C → A×B$ the pairing induced by any morphisms $f∶ C → A$ and
$g∶ C → B$.  Similarly, when it makes sense, the copairing of
$f∶ A → C$ and $g∶ B → C$ is denoted by $[f,g]∶ A + B → C$.
Initial objects are denoted by $0$.
Given an endofunctor $F$, the category of algebras $F x → x$ is denoted by
$F\alg$. When $F$ is a monad, the notation $F\alg$ rather refers to
its category of algebras in the sense of monads, that is, morphisms $F x → x$ satisfying
the three standard axioms.
Finally, given an
object $c$ of a category $C$, we denote the slice (resp.\ coslice)
category over (resp.\ under) $c$ by $C/c$ (resp.\ $c/C$).

\subsection{Locally finitely presentable categories and finitary functors}
\label{ss:finitely-pres-cat}
We heavily rely on the theory of locally finitely presentable
categories~\cite{Adamek}.  Very briefly, recall that filtered
categories are a categorical generalisation of directed posets.
\begin{defi}
  A category is \alert{filtered} when
\begin{itemize}
\item it is not empty,
\item for any two objects $C$ and $D$, there is an object $E$
  and arrows $C → E$ and $D → E$, and furthermore,
\item any two parallel arrows $C ⇉ D$ are coequalised by some morphism
  $D → E$.
\end{itemize}
A \alert{filtered colimit} is a colimit of some functor from some
small filtered category.
\end{defi}
\begin{defi}
  An object of a category is \alert{finitely presentable} iff
its covariant hom-functor preserves filtered colimits.
\end{defi}
\begin{defiC}[{\cite[Definition 1.9]{Adamek}}]
  A locally small category is \alert{locally finitely presentable} iff it is
  cocomplete and every object is a filtered colimit of objects from a
  fixed set of finitely presentable generators.
\end{defiC}
\begin{exaC}[{\cite[Example 1.12]{Adamek}}]
 Any presheaf category is locally finitely presentable. 
\end{exaC}
\begin{propC}[{\cite[Corollary 1.28]{Adamek}}]
 Any locally finitely presentable category is complete. 
\end{propC}
In this context, functors that preserve filtered colimits are
important.  They are called \alert{finitary}.
\begin{defi}
  Let $[𝐂,𝐃]_f$ denote the category of finitary functors $𝐂 → 𝐃$ for
  any categories $𝐂$ and $𝐃$.
\end{defi}
\begin{prop}
  \label{prop:functor-finitely-pres}
  If $𝐂$ and $𝐃$ are locally finitely presentable, then so
  is $[𝐂, 𝐃]_f$.
\end{prop}
\begin{proof}
  See~\cite[Section 4]{KellyPower} for the more general enriched case.
\end{proof}
\subsection{Convenient monoidal categories}
\label{ss:convenient-monoidal}
We sometimes work in a monoidal category satisfying some additional
properties. We call such monoidal categories convenient.
\begin{defi}
  A monoidal category is \alert{convenient} when 
  \begin{itemize}
  \item it is locally finitely presentable;
  \item the tensor preserves filtered colimits on the right and all
    colimits on the left.
  \end{itemize}
\end{defi}
\begin{prop}
\label{prop:finitary-endo-convenient}
Any category $[𝐂,𝐂]_f$ of finitary endofunctors on a locally finitely
presentable category $𝐂$ is convenient for the composition monoidal
structure.
\end{prop}
\begin{proof}
  By Proposition~\ref{prop:functor-finitely-pres}, any category of
  finitary endofunctors on a locally finitely presentable category is
  locally finitely presentable.
  Furthermore, since colimits are computed pointwise in functor
  categories whenever the codomain category is
  cocomplete~\cite[§V.4]{MacLane:cwm}, we have
  $(\colimᵢ Gᵢ) ∘ F ≅ \colimᵢ (Gᵢ ∘ F)$ for any diagram $G$ and object
  $F$, thus the composition tensor product preserves all colimits on
  the left. Finally, by definition of finitariness, the considered
  functors preserve filtered colimits, hence for any such diagram $G$
  and object $F$ we have $F ∘ \colimᵢ Gᵢ ≅ \colimᵢ (F ∘ Gᵢ)$ as
  desired.
\end{proof}
\begin{exa}
  In particular, all categories of the form $[𝐒𝐞𝐭^ℙ,𝐒𝐞𝐭^ℙ]_f$ that we
  will consider below (where $ℙ$ is a set) are convenient
  for the composition monoidal structure.
\end{exa}

\subsection{(Parametric) modules over monoids and monads}\label{ss:modules}

\subsubsection{(Parametric) modules over monoids}
Let us fix a monoidal category $𝐂$.  We first recall the standard
notions of monoid and (right) module over a monoid in a monoidal
category, and then introduce the notion of parametric module, inspired
by~\cite{hirscho:lam}.

\begin{defi}
  Let $𝐌𝐨𝐧(𝐂)$ denote the category of monoids, in any monoidal category $𝐂$~\cite[VII.3]{MacLane:cwm}.
\end{defi}
\begin{defi}\label{def:Xmodules}
  Given a monoid $(X,e_X,m_X)$ in $𝐂$, a (right) \alert{$X$-module} is an object
  $M$ equipped with a morphism $a∶ M ⊗ X → M$
  making the following diagrams commute
  \begin{mathpar}
\begin{tikzcd}[ampersand replacement=\&]
	{M \otimes I} \&\& {M \otimes X} \\
	\& M
	\arrow["{M \otimes e_X}", from=1-1, to=1-3]
	\arrow["a", from=1-3, to=2-2]
	\arrow["{\rho_M}"', from=1-1, to=2-2]
\end{tikzcd}
\and
\begin{tikzcd}[ampersand replacement=\&]
	{(M \otimes X) \otimes X} \&\& {M \otimes (X \otimes X)} \\
	{M \otimes X} \&\& {M \otimes X} \\
	\& M
	\arrow["{\alpha_{M,X,X}}", from=1-1, to=1-3]
	\arrow["{M \otimes m_X}", from=1-3, to=2-3]
	\arrow["a", from=2-3, to=3-2]
	\arrow["{a \otimes X}"', from=1-1, to=2-1]
	\arrow["a"', from=2-1, to=3-2]
\end{tikzcd}
  \end{mathpar}
 (the axioms for the dual case of left modules
  are given in~\cite[§3.2]{SealTensors}).  We denote by $X\Mod$ the
  category of $X$-modules, with action-preserving morphisms between
  them.
\end{defi}
\begin{rem}
  The coherence conditions amount to equipping $M$ with algebra
  structure for the monad ${-} ⊗ X$ induced by $X$.
\end{rem}

A parametric module will assign a module to any monoid.
In order to formally define this, let us introduce
the following category, which collects all categories of the form $X\Mod$.
\begin{defi}
Let $𝐌𝐨𝐝(𝐂)$ denote the category
\begin{itemize}
\item whose objects are pairs $(X,M)$ where $X$ is a monoid and $M$ is an
  $X$-module, and
\item whose morphisms $(X,M) → (Y,N)$ are pairs $(f,g)$ of a monoid
  morphism $f∶ X → Y$ and a morphism $g∶ M → N$ in $𝐂$ such that the
  following diagram commutes.
  \begin{center}
    \diag{%
      M⊗X \& N⊗Y \\
      M \& N
    }{%
      (m-1-1) edge[labela={g⊗f}] (m-1-2) %
      edge[labell={}] (m-2-1) %
      (m-2-1) edge[labelb={g}] (m-2-2) %
      (m-1-2) edge[labelr={}] (m-2-2) %
    }
  \end{center}
\end{itemize}
Let $𝐔^{𝐌𝐨𝐝}∶ 𝐌𝐨𝐝(𝐂) → 𝐌𝐨𝐧(𝐂)$ denote the forgetful functor.
\end{defi}
\begin{defi}\label{def:paramX}
  A \alert{parametric module} over $𝐂$ is a section of the forgetful
  functor $𝐔^{𝐌𝐨𝐝}∶ 𝐌𝐨𝐝(𝐂) → 𝐌𝐨𝐧(𝐂)$, i.e., a functor
  $S∶ 𝐌𝐨𝐧(𝐂) → 𝐌𝐨𝐝(𝐂)$ such that $𝐔^{𝐌𝐨𝐝} ∘ S = \id_{𝐌𝐨𝐧(𝐂)}$.
\end{defi}
In other words, a parametric module functorially assigns to each monoid
a module over it.
\begin{exa}
  If $𝐂$ has products, then we can define the parametric module mapping
  any monoid $X$ to the $X$-module $X×X$. 
  It will become clear in~§\ref{ss:regmonads} how this
  parametric module can be viewed as the arity of a binary operation,
  with $𝐂 = [𝐒𝐞𝐭, 𝐒𝐞𝐭]_f$.
\end{exa}
\begin{exa}\label{ex:Tmod}
  Any endofunctor $T$ on $𝐌𝐨𝐧(𝐂)$ equipped with a natural
  transformation $η∶ \Id → T$ induces a parametric module $T\pmodmon$
  mapping any monoid $M$ to $T(M)$, with action given by
  \[T(M) ⊗ M \xto{T(M)⊗η_M} T(M)⊗T(M) → T(M)\rlap{,}\]
  where the second morphism is the multiplication of $T(M) ∈ 𝐌𝐨𝐧(𝐂)$.
  This construction applies in particular for any monad on $𝐌𝐨𝐧(𝐂)$.
\end{exa}

\begin{defi}
  A \alert{parametric module morphism} $M → N$ over $𝐂$ is a natural
  transformation $α∶ M → N$ between underlying functors
  $𝐌𝐨𝐧(𝐂) → 𝐌𝐨𝐝(𝐂)$, such that $𝐔^{𝐌𝐨𝐝} ∘ α = \id$.
\end{defi}
\begin{rem}
  Concretely, the components of a natural transformation $α∶ M → N$ at
  any monoid $X ∈ 𝐌𝐨𝐧(𝐂)$ are pairs of a monoid morphism $f∶ X → X$
  and a suitable natural transformation $g∶ M(X) → N(X)$.  The
  condition $𝐔^{𝐌𝐨𝐝} ∘ α = \id$ unfolds to $f = \id_X$ for all $X$.
\end{rem}

\subsubsection{(Parametric) modules over monads}\label{sss:param}
The previous definitions of (parametric) modules specialise to the case where $𝐂 = [𝐄,𝐄]_f$
(for the composition monoidal structure). Then, monoids are finitary monads.
However, because in this case $𝐂$ is an
endofunctor category, there is a slightly more general, ``relative'',
or ``heterogeneous'' notion of module~\cite{hirscho:lam} which will be
important for us.  Let us recall it now.

\begin{defi}
  Given a monad $T$ on $𝐂$ and a category $𝐃$, a \alert{$𝐃$-valued $T$-module} is a
  finitary functor $M∶ 𝐂 → 𝐃$ equipped with a right $T$-action
  $M∘T → M$ satisfying coherence conditions analogous to those of
  Definition~\ref{def:Xmodules}.  A morphism of $T$-modules is
  similarly a natural transformation commuting with action.  We denote
  by $T\Mod_f(𝐃)$ the category of $T$-modules and morphisms between
  them.
\end{defi}
\begin{rem}\label{rk:samesame}
  $𝐃$-valued $T$-modules are algebras for the monad $-∘T $ on
  $ [𝐂,𝐃]_f $.  When $𝐃 = 𝐂$, these are the same as $T$-modules in
  $[𝐂,𝐂]_f$ in the sense of Definition~\ref{def:Xmodules}.
\end{rem}
\begin{nota}\label{not:freemodule}
  Given any monad $T$ on $𝐂$, and functor $F∶ 𝐂 → 𝐃$, we sometimes
  implicitly equip the functor $FT$ with its canonical $T$-module
  structure. This module is free on $F$, in the sense that
  for any $T$-module $M$ and natural transformation
  $α∶ F → M$, there is a unique $T$-module morphism
  $\tilde{α}∶ FT → M$ making the following diagram commute.
  \begin{center}
    \diag{%
      F \& FT \\
        \& M
      }{%
        (m-1-1) edge[labela={Fηᵀ}] (m-1-2) %
        edge[labelbl={α}] (m-2-2) %
        (m-1-2) edge[labelr={\tilde{α}}] (m-2-2) %
    }
  \end{center}
  Of course, $\tilde{α}$ is merely the composite
  $FT \xto{α_T} MT → M$.
\end{nota}
Let us briefly show how to exploit the variability of $𝐃$.
\begin{defi}
  For any $p$ in a set $ℙ$, and any monad $T$ on $𝐒𝐞𝐭^ℙ$, the functor
  $Tₚ∶ 𝐒𝐞𝐭^ℙ → 𝐒𝐞𝐭$ mapping any $X ∈ 𝐒𝐞𝐭^ℙ$ to $T(X)(p)$ is a
  $T$-module, with action given by
  \[μ_{X,p}∶ T(T(X))(p) → T(X)(p)\rlap{,}\] at any $X ∈ 𝐒𝐞𝐭^ℙ$.
\end{defi}
Here is another example construction of $T$-module, which is useful
for specifying syntax with variable binding.
\begin{defi}
For any sequence $p₁,…,pₙ$ in a set $ℙ$, for any monad $T$ on $𝐒𝐞𝐭^ℙ$
and $𝐃$-valued $T$-module $M$, we denote by $M^{(p₁,…,pₙ)}$ the
$𝐃$-valued $T$-module defined by
  \[M^{(p₁,…,pₙ)} (X) = M(X + 𝐲_{p₁} + … + 𝐲_{pₙ})\rlap{,}\] where $𝐲∶ ℙ → 𝐒𝐞𝐭^ℙ$
is the embedding defined by $𝐲ₚ(q) = 1$ if $p=q$ and $∅$ otherwise.
If $ℙ$ is a singleton, we
abbreviate this to $M^{(n)}$.
\end{defi}

Let us now recall parametric modules over monads (called \alert{signatures}
in~\cite{DBLP:conf/csl/AhrensHLM18}).
\begin{defi}
  Given categories $𝐂$ and $𝐃$, let $𝐌𝐨𝐝(𝐂,𝐃)$, or $𝐌𝐨𝐝(𝐃)$ when $𝐂$
  is clear from context, denote the category
\begin{itemize}
\item whose objects are pairs $(T,M)$ of a finitary monad $T$ on $𝐂$
  and a finitary $T$-module $M∶ 𝐂 → 𝐃$,
\item and whose morphisms $(T,M) → (U,N)$ are pairs $(α,β)$ of a monad
  morphism $α∶ T → U$ and a natural transformation $β∶ M → N$
  commuting with action, in the sense that the following square commutes.
  \begin{center}
    \diag{%
      M∘T \& N∘U \\
      M \& N
    }{%
      (m-1-1) edge[labela={β∘α}] (m-1-2) %
      edge[labell={}] (m-2-1) %
      (m-2-1) edge[labelb={β}] (m-2-2) %
      (m-1-2) edge[labelr={}] (m-2-2) %
    }
  \end{center}
\end{itemize}
The first projection yields a forgetful functor $𝐩∶ 𝐌𝐨𝐝(𝐃) → 𝐌𝐧𝐝_f$.
\end{defi}
\begin{defi}\label{def:paramodule}
  A ($𝐃$-valued) \alert{parametric} module is a section of $𝐩$, or in
  other words a functor $s∶ 𝐌𝐧𝐝_f → 𝐌𝐨𝐝(𝐃)$ such that
  $𝐩∘s = \id_{𝐌𝐧𝐝_f}$.

  A \alert{parametric module morphism} $M → N$ is a natural
  transformation $α∶ M → N$ between underlying functors such that
  $𝐩 ∘ α = \id$.
\end{defi}

\begin{terminology}\label{term:setvalued}
  In the following, $𝐂 = 𝐒𝐞𝐭^ℙ$, and parametric modules are implicitly
  $𝐒𝐞𝐭$-valued by default.
\end{terminology}

\begin{exa}
\label{ex:basic-param-mod}
  Let us start by a few basic constructions of parametric modules:
  \begin{itemize}
  \item we denote by $Θ$ the $𝐒𝐞𝐭^ℙ$-valued parametric module mapping
    a monad $T$ to itself, as a module over itself;
  \item for any $p₁,…,pₙ ∈ ℙ$ and $𝐃$-valued parametric module $M$, let
    $M^{(p₁,…,pₙ)}$ associate to each monad $T$ the $T$-module
    $M(T)^{(p₁,…,pₙ)}$ as in~§\ref{ss:modules}, i.e.,
    $M^{(p₁,…,pₙ)}(T)(X) = M(T)(X + 𝐲_{p₁} + … + 𝐲_{pₙ})$; when
    $ℙ = 1$, we merely count the $pᵢ$'s and write $M^{(n)}$;
  \item for any finitary functor $F ∶ 𝐃 → 𝐄$ and $𝐃$-valued parametric module $M$,
    the $𝐄$-parametric module $F∘M$ maps any monad $T$ to the $T$-module $F∘M(T)$;
    as particular cases:
    \begin{itemize}
    \item when $𝐃$ has a terminal object, the terminal $𝐃$-valued
      parametric module $1 = 1 ∘ Θ$ maps any monad $T$ to the constant
      $T$-module $1$;
    \item for any $p∈ℙ$ and $𝐒𝐞𝐭^ℙ$-valued parametric module $M$, we
      denote by $Mₚ$ the $𝐒𝐞𝐭$-valued parametric module mapping any
      monad $T$ to the $T$-module $X ↦ M(X)ₚ$;
    \item in particular, for any $p∈ℙ$, the $𝐒𝐞𝐭$-valued parametric module $Θₚ$ maps
      a monad $T$ on $𝐒𝐞𝐭^ℙ$ to the module $X↦T(X)ₚ$;
    \item given a finite family $(Mᵢ)_{i∈ I}$ of $𝐒𝐞𝐭$-valued
      parametric modules, let $∏ᵢMᵢ$ associate to any monad $T$ the
      $T$-module $∏ᵢMᵢ(T)$.
    \end{itemize}
  \end{itemize}
\end{exa}

  In the paper, we will at times use two distinct viewpoints on
  our signatures for monads.  One, more user-friendly, is based on
  $𝐒𝐞𝐭$-valued, or \alert{heterogeneous}, parametric modules.  The other,
  more efficient for stating the explicit description of initial algebras,
  is based on $𝐒𝐞𝐭^ℙ$-valued, or \alert{homogeneous}, modules.
  The two viewpoints are related by a $ℙ$-indexed family of
  adjunctions, which we now recall.  For all $r ∈ ℙ$, there is an
  adjunction
  \begin{center}
    \adj{[𝐒𝐞𝐭^ℙ,𝐒𝐞𝐭]_f}{[𝐒𝐞𝐭^ℙ,𝐒𝐞𝐭^ℙ]_f\rlap{,}}{(-)·𝐲ᵣ}{(-)ᵣ}
  \end{center}
  where
  \begin{itemize}
  \item $𝐲ᵣ∶ 𝐒𝐞𝐭^ℙ$ denotes the family with a single element, sitting
    over $r ∈ ℙ$, and
  \item the left adjoint maps any $F∶ 𝐒𝐞𝐭^ℙ → 𝐒𝐞𝐭$ to the endofunctor
    $(X ∈ 𝐒𝐞𝐭^ℙ) ↦ F(X)·𝐲ᵣ$, i.e., $(F(X)·𝐲ᵣ)(r) = F(X)$ and
    $(F(X)·𝐲ᵣ)(p) = ∅$ when $p≠r$.
  \end{itemize}
  Indeed, we have a natural isomorphism
  \begin{mathpar}
    {\mprset{fraction={===}}\inferrule{F(X)·𝐲ᵣ → G(X)}{F(X) → G(X)(r)}}~·
  \end{mathpar}
  \begin{prop}\label{prop:adjyev}
    These adjunctions lift to a $ℙ$-indexed family of adjunctions
    \begin{equation}
      \adj{𝐌𝐨𝐝(𝐒𝐞𝐭^ℙ,𝐒𝐞𝐭)}{𝐌𝐨𝐝(𝐒𝐞𝐭^ℙ,𝐒𝐞𝐭^ℙ)\rlap{.}}{(-)·𝐲ᵣ}{(-)ᵣ}
      \label{eq:adjyev}
    \end{equation}
  \end{prop}
  \begin{proof}
    Straightforward.
  \end{proof}

\subsection{Creation of (co)limits and monadic functors}\label{ss:creation}
In this section, we recall Beck's monadicity theorem~\cite[Theorem
VI.7.1]{MacLane:cwm}, and state a ``cancellation'' property for
monadic functors.  First, we need to recall creation of (co)limits,
monadic functors, and absolute (co)limits.  We will see
in~§\ref{s:friendliness} that the main technical notion for our
explicit descriptions of initial algebras is creation of
(co)limits, so let us briefly recall the basics.
\begin{defiC}[{\cite[§V.1]{MacLane:cwm}}]
  Given a small category $𝐄$, a functor $F∶ 𝐂 → 𝐃$ \alert{creates} (co)limits of shape $𝐄$ if for
  any functor $J∶ 𝐄 → 𝐂$, if $FJ$ has a (co)limit, then the
  (co)limiting (co)cone uniquely lifts to $𝐂$, and the lifting is 
  again (co)limiting.
\end{defiC}
A typical example is:
\begin{propC}[{\cite[Exercise VI.2.2]{MacLane:cwm}}]
  For any monad $T$ on a category $𝐂$, the forgetful functor $U∶ T\alg → 𝐂$
  creates limits.
\end{propC}

We also have the following well-known result:
\begin{prop}\label{prop:algprescolim}
  For any monad $T$ on a category $𝐂$, the forgetful functor
  $U∶ T\alg → 𝐂$ creates colimits of a given shape whenever $T$
  preserves them.  More concretely, if $T$ preserves all colimits of
  functors with some domain $𝐃$, then $U$ creates them.
\end{prop}
\begin{proof}
  This is a straightforward consequence of~\cite[Proposition 4.3.2]{BorceuxII}.
\end{proof}

\begin{lem}\label{lem:creates:preserves}
  For any category $𝐄$ and functor $F∶ 𝐂 → 𝐃$ with $𝐃$ cocomplete, if
  $F$ creates colimits of shape $𝐄$, then $F$ preserves them.
\end{lem}
\begin{proof}
  Consider any colimiting cocone, say $λ∶ J → C$ of a functor
  $J∶ 𝐄 → 𝐂$.  We want to show that $Fλ$ is colimiting. By
  cocompleteness of $𝐃$, $FJ$ has a colimiting cocone, say
  $γ∶ FJ → D$. By creation of colimits, $γ$ has a unique lifting
  $γ'∶ F → C'$ to $𝐂$, which is colimiting.  Because colimiting
  cocones are uniquely isomorphic, there is a unique isomorphism
  $λ ≅ γ'$. Finally, functors preserve isomorphisms, so we have
  $Fλ ≅ Fγ' = γ$. But $γ$ is colimiting, hence so is $Fλ$, as desired.
\end{proof}

\begin{cor}\label{cor:Ucreates:Tpreserves}
  For any monad $T$ on a cocomplete category $𝐂$, and for any category
  $𝐃$, the following are equivalent
  \begin{enumerati}
  \item \label{item:Upreserves} the forgetful functor $U∶ T\alg → 𝐂$
    preserves all colimits of shape $𝐃$, 
  \item \label{item:Tpreserves} the monad $T$ preserves all colimits
    of shape $𝐃$, and
  \item \label{item:Ucreates} the forgetful functor $U∶ T\alg → 𝐂$
    creates all colimits of shape $𝐃$.
\end{enumerati}
\end{cor}
\begin{proof}
  For \ref{item:Upreserves} $⇒$ \ref{item:Tpreserves}, the left
  adjoint $L$ to $U$ is cocontinuous, so $T = UL$ preserves all such
  colimits by composition. Furthermore, \ref{item:Tpreserves} $⇒$
  \ref{item:Ucreates} follows readily from
  Proposition~\ref{prop:algprescolim}.  Finally,
  Lemma~\ref{lem:creates:preserves} proves \ref{item:Ucreates} $⇒$
  \ref{item:Upreserves}.
\end{proof}

\begin{defi}
  A functor $U∶ 𝐄 → 𝐁$ is \alert{monadic} if $𝐄$ is isomorphic to a category of
  algebras $T\alg$ for some monad $T$ on $𝐁$, and furthermore, the following diagram commutes.
  
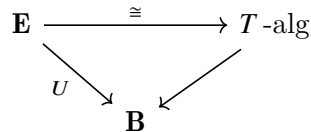
\[
    \begin{tikzcd}
     𝐄 \ar[rr, "≅"] \ar[rd, swap, "U"] && T\alg \ar[dl]
     \\ & 𝐁
    \end{tikzcd}
  \]
\end{defi}

\begin{rem}\label{rem:strictlymonadic}
  This notion of monadicity is quite strict. Some authors prefer a
  relaxed version where isomorphism is replaced with equivalence.
\end{rem}

\begin{defi}
  A (co)limit is \alert{absolute} if it is preserved by all functors.
\end{defi}
\begin{exa}
  A split coequaliser is a coequaliser
  \[
    \begin{tikzcd}
     A \ar[r, shift left=.75ex, "f"] \ar[r, shift right=.75ex, "g", swap] & B
     \ar[r, "e"]
     &C,
    \end{tikzcd}
  \]
  such that
  \begin{itemize}
  \item $e$ has a section $s∶ C → B$, 
  \item $f$ has a section $t∶ B → A$ with
    \item $g∘t = s ∘ e$.
  \end{itemize}
  Such a coequaliser is absolute~\cite[Corollary~VI.6]{MacLane:cwm}.

  On the other hand, the initial object $∅$ in $𝐒𝐞𝐭$, viewed as the
  colimit of the unique functor from the empty category, is not
  absolute, since it is not preserved by the constant terminal
  endofunctor on $𝐒𝐞𝐭$.
\end{exa}
\begin{thmC}[{Beck's monadicity
  theorem~\cite[Theorem~VI.7]{MacLane:cwm}}] 
\label{thm:beck} \hfill

  A functor $U∶ 𝐄 → 𝐁$ is monadic iff
  \begin{itemize}
  \item it has a left adjoint, and
  \item it creates coequalisers for those parallel pairs
    $f,g∶ E₁ ⇉ E₂$ for which $(U(f),U(g))$ has an absolute coequaliser
    in $𝐁$.
  \end{itemize}
\end{thmC}
We will use  the following consequence.
\begin{prop}\label{prop:monadic}
  Given a commuting triangle of functors
  \begin{center}
    \diag{%
      𝐀\& \& 𝐁 \\
      \& 𝐂 %
    }{%

      (m-1-1) edge[labela={H}] (m-1-3) %
      edge[labelbl={F}] (m-2-2) %
      (m-1-3) edge[labelbr={G}] (m-2-2) %
    }
  \end{center}
  between cocomplete categories, if $F$ and $G$ are monadic, then so
  is $H$.

  If furthermore $F$ and $G$ are finitary, then so is $H$.
\end{prop}
\begin{proof}
  For monadicity, Borceux~\cite[Corollary~4.5.7]{BorceuxII} gives a
  proof in the weakly monadic case (see
  Remark~\ref{rem:strictlymonadic}). This is a straightforward
  adaptation.
  Finitariness follows from the next lemma.
\end{proof}

\begin{lem}\label{lem:finitary}
  Given a commuting triangle of functors as in
  Proposition~\ref{prop:monadic}, if $F$ is finitary and $G$ creates
  filtered colimits, then $H$ is finitary.
\end{lem}
\begin{proof}
  Given a colimiting cocone $c∶ J → A$ for a filtered diagram
  $J∶ 𝐃 → 𝐀$, this cocone is preserved by $F$ and created by $G$.  So
  $F(c)$ is colimiting, and has a unique antecedent by $G$, which is
  again colimiting.  But $H(c)$ is an antecedent, hence has to be
  \emph{the} antecedent, and so is colimiting as desired.
\end{proof}

\subsection{Monads vs.\ monadic functors}\label{ss:street}
In this section, we recall the equivalence between (finitary) monadic
functors and (finitary) monads.

Let us fix a locally finitely presentable category $𝐂$.
\begin{defi}
  Let $𝐌𝐧𝐝(𝐂)$ denote the category of monads on $𝐂$.  When
  $𝐂$ is clear from context, we sometimes abbreviate this to just
  $𝐌𝐧𝐝$.

  Furthermore, let $𝐌𝐧𝐝_f(𝐂)$ denote the full subcategory of finitary
  monads on $𝐂$.
\end{defi}

\begin{defi} 
  Let $𝐌𝐨𝐧𝐚𝐝𝐢𝐜/𝐂$ denote the full subcategory of $𝐂𝐀𝐓/𝐂$ spanning
  (strictly) monadic functors.

  Furthermore, let $𝐌𝐨𝐧𝐚𝐝𝐢𝐜_f/𝐂$ denote the full subcategory spanning
  (strictly) monadic functors that are finitary, or equivalently (by
  Corollary~\ref{cor:Ucreates:Tpreserves}), whose underlying monad is
  finitary.
\end{defi}

Let us readily make the following observation.
\begin{lem}\label{lem:allmonadic} The functor underlying any morphism in
  $𝐌𝐨𝐧𝐚𝐝𝐢𝐜_f/𝐂$ is itself monadic.
\end{lem}
\begin{proof} By Proposition~\ref{prop:monadic}.
\end{proof}

In fact, a lot of our understanding of $𝐌𝐨𝐧𝐚𝐝𝐢𝐜_f/𝐂$ will follow from
the following equivalence.
\begin{prop}\label{prop:street}
  The functor \[(-)\alg∶ \op{𝐌𝐧𝐝(𝐂)} → 𝐌𝐨𝐧𝐚𝐝𝐢𝐜/𝐂\] mapping any monad
  $T$ to the forgetful functor $T\alg → 𝐂$ is an equivalence.
\end{prop}
\begin{proof}
  The functor is essentially surjective by definition of monadic
  functors. It is also full and faithful by
  \cite[Proposition~5.3]{Barr1970CoequalizersAF}.
\end{proof}

\begin{cor}\label{cor:street}
  Let $𝐂$ be cocomplete. The equivalence
  \[(-)\alg∶ \op{𝐌𝐧𝐝(𝐂)} → 𝐌𝐨𝐧𝐚𝐝𝐢𝐜/𝐂\]
  of Proposition~\ref{prop:street} lifts to a functor
  \[(-)\alg∶ \op{𝐌𝐧𝐝_f(𝐂)} → 𝐌𝐨𝐧𝐚𝐝𝐢𝐜_f/𝐂\rlap{,}\] which is again an
  equivalence.
\end{cor}
\begin{proof}
  For any finitary monad $T$, the forgetful functor is finitary by
  Corollary~\ref{cor:Ucreates:Tpreserves}, which proves that the
  functor lifts as claimed. Furthermore, the lifted functor is clearly
  fully faithful.  Finally, it is essentially surjective because if a
  monadic functor is finitary, then so is the induced monad, again by
  Corollary~\ref{cor:Ucreates:Tpreserves}.
\end{proof}

\subsection{Limits of finitary monadic functors}\label{ss:monadic}
In this section, we recall a well-known result about limits of
(finitary) monadic functors, namely that they are computed as in
$𝐂𝐀𝐓$:

\begin{prop}\label{prop:monadic:eq}
  The forgetful functor $𝐌𝐨𝐧𝐚𝐝𝐢𝐜_f/𝐂 → 𝐂𝐀𝐓/𝐂$ creates limits.
\end{prop}

The rest of this section will be devoted to the proof, but first let us
record the following.
\begin{cor}\label{cor:monadicequaliser}
  The forgetful functor $𝐌𝐨𝐧𝐚𝐝𝐢𝐜_f/𝐂 → 𝐂𝐀𝐓$ creates equalisers.  More
  precisely, given monadic functors $F₁∶𝐃₁ → 𝐂$ and $F₂ ∶ 𝐃₂ → 𝐂$ 
  and functors $G₁,G₂ ∶ 𝐃₁→ 𝐃₂$ such that $F₂ ∘ Gᵢ = F₁$, 
  if $𝐀 → 𝐃₁$ is the equaliser of $G₁$ and $G₂$ in $𝐂𝐀𝐓$, then the
  composite $𝐀 → 𝐃₁ → 𝐂$ is finitary monadic and underlies the
  equaliser of $G₁$ and $G₂$ in $𝐌𝐨𝐧𝐚𝐝𝐢𝐜_f/𝐂$.
\end{cor}
\begin{proof}
  This follows straightforwardly from
  Proposition~\ref{prop:monadic:eq} and the fact that the forgetful
  functor from any slice category to the base category creates
  equalisers (see the proof of~\cite[Proposition~2.16.3]{BorceuxI}).
\end{proof}

Returning to the proof of Proposition~\ref{prop:monadic:eq}, let us
start with the following two lemmas.
\begin{lem}\label{lem:lfpmonads} Finitary monads over any locally
  finitely presentable category form a locally finitely presentable
  category.
\end{lem}
\begin{proof}[Proof sketch (see~\cite{MonadicityMonads})] Let $𝐂$ be
  locally presentable. Then $𝐂_f$ is small, so $[𝐂_f,𝐂]$ is again
  locally presentable.  Now, finitary monads on $𝐂$ are equivalently
  monoids for the composition tensor product in $[𝐂_f,𝐂]$, hence also
  algebras for a finitary monad on a locally presentable category,
  which allows us to conclude by the following lemma.
\end{proof}

\begin{lem}\label{lem:algfinitary}
  The category of algebras for a finitary monad on any locally
  finitely presentable category is again locally finitely presentable.
\end{lem}
\begin{proof} This is~\cite[Remark~2.78]{Adamek} with $λ = ω$.
\end{proof}

We now need to recall three standard definitions, and prove two more lemmas.
\begin{defi}
  A functor $F∶ 𝐂 → 𝐃$ is \alert{conservative} iff, for any morphism
  $f∶ C → C'$ in $𝐂$ such that $F(f)$ is an isomorphism, so is $f$.
\end{defi}
\begin{defi}
  A functor $F∶ 𝐂 → 𝐃$ is \alert{amnestic} iff, for any isomorphism
  $i∶ C → C'$ in $𝐂$ such that $F(C) = F(C')$ and $F(i) = \id$, we
  have $C=C'$ and $i = \id$.
\end{defi}
\begin{lem}\label{lem:consamn}
  For any conservative and amnestic functor $F∶ 𝐂 → 𝐃$ and morphism
  $f$ in $𝐂$, if $F(f)$ is an identity, then so is $f$.
\end{lem}
\begin{proof}
  By conservativeness, $f$ is an isomorphism. But then by amnesia $f$
  is an identity.
\end{proof}
\begin{defi}
  A functor $F∶ 𝐂 → 𝐃$ is an  \alert{iso-fibration} iff, for any
  isomorphism of the form $i∶ D → F(C)$, there exists an isomorphism
  $j∶ C' → C$ such that $F(j) = i$.  
\end{defi}
\begin{lem}\label{lem:amnestic:creates}
  Any continuous, conservative, amnestic iso-fibration from a complete category
  creates limits.
\end{lem}
\begin{proof}
  Consider any continuous, amnestic iso-fibration $𝐅∶ 𝐂 → 𝐃$ with $𝐂$
  complete, and any functor $J∶ 𝐗 → 𝐂$ such that $FJ$ has a limiting
  cone, say $δ∶ D → FJ$.  Then, because $𝐂$ is complete, $J$ also has
  a limiting cone, say $γ∶ C → J$. Because $F$ is continuous,
  $F(γ)∶ F(C) → FJ$ is again limiting, hence we get an isomorphism
  $i∶ D → F(C)$ of cones over $FJ$.  Because $F$ is an iso-fibration,
  we then lift $i$ to an isomorphism $j∶ C' → C$ in $𝐂$ such that
  $F(j) = i$.  The cone $γj∶ C' → J$ is thus limiting, and an
  antecedent of $δ$.  It thus remains to show that it is the only
  antecedent of $δ$.  Let thus $γ''∶ C'' → J$ be any antecedent of
  $δ$. Because $γj$ is limiting, there exists a unique cone morphism
  $m∶ C'' → C'$. But now $F(m)$ is a cone endomorphism $δ → δ$, hence
  $F(m) = \id$. By Lemma~\ref{lem:consamn}, we then get $C'' = C'$ and
  $m=id$, thus proving that $γ'' = γj$ as desired.
\end{proof}

At last, we have:
  \begin{proof}[Proof of Proposition~\ref{prop:monadic:eq}]
    By Corollary~\ref{cor:street}, $𝐌𝐨𝐧𝐚𝐝𝐢𝐜_f/𝐂$ is equivalent to
    $\op{𝐌𝐧𝐝_f(𝐂)}$. But the latter is complete as the opposite of
    $𝐌𝐧𝐝_f(𝐂)$, which is locally finitely presentable by
    Lemma~\ref{lem:lfpmonads}.  Thus, $𝐌𝐨𝐧𝐚𝐝𝐢𝐜_f/𝐂$ is complete.

    Moreover, the forgetful functor $𝐌𝐨𝐧𝐚𝐝𝐢𝐜_f/𝐂 → 𝐂𝐀𝐓/𝐂$ is
    continuous. Indeed, it is equivalent (in $𝐂𝐀𝐓^{→}$) to the
    composite \[\op{𝐌𝐧𝐝_f(𝐂)} ↪ \op{𝐌𝐧𝐝(𝐂)} → 𝐂𝐀𝐓/𝐂\rlap{,}\] whose
    first component is continuous
    by~\cite[Proposition~5.6]{blackwell_phd}, while the second is
    by~\cite[Proposition~26.3]{KellyUnified}.

    Finally, the forgetful functor $𝐌𝐨𝐧𝐚𝐝𝐢𝐜_f/𝐂 → 𝐂𝐀𝐓/𝐂$ is a conservative,
    amnestic iso-fibration:
    \begin{itemize}
    \item it is conservative and amnestic as a full subcategory
      embedding, and
    \item an iso-fibration because the subcategory in question is
      replete (otherwise said, monadic functors are closed under
      isomorphisms).
    \end{itemize}
    The result thus follows by Lemma~\ref{lem:amnestic:creates}.
\end{proof}

\section{\Transition monads}
\label{sec:operational-monads}
In this section, we introduce the main new mathematical notion of the
paper: \transition monads.
In~§\ref{ss:overview-transmon} we give an informal description 
and in~§\ref{ss:def-transmon}, we give our formal definition.
In~§\ref{ss:transmon-relmon}, we provide an equivalent definition based on
the notion of relative monad. Finally, in~§\ref{ss:irrelevant-trans-monads},
we sketch a proof-irrelevant variant of \transition monads.

\subsection{Overview of \transition monads}
\label{ss:overview-transmon}
\hfill
\medskip

\mypar{Placetakers and states}

In standard $λ$-calculus, we have terms, variables are
placeholders for terms, and transitions relate a source term to a
target term.  In a general \transition monad we still have variables and
transitions, but placetakers for variables and endpoints of transitions
can be of a different nature, which we phrase as follows: variables
are placeholders for \alert{placetakers}, while transitions relate a
 \alert{source state} with a \alert{target state}.  

\mypar{The categories for placetakers and for states}

In standard $λ$-calculi, we have a set $𝕋$ of types for terms (and
variables); for instance in the untyped version, $𝕋$ is a singleton.
Accordingly, terms form a \alert{monad} on the category $𝐒𝐞𝐭^𝕋$.  In a
general \transition monad we have a set $ℙ$ of placetaker types, and
placetakers form a \alert{monad} on the category $𝐒𝐞𝐭^ℙ$; similarly,
we have a set $𝕊$ of transition types, and the family of possible
states (depending on a given family of variables) forms an object in
$𝐒𝐞𝐭^𝕊$.  For example, for the simply-typed $λ$-calculus, $ℙ=𝕊$ is the
set of simple types.

\mypar{The object of variables}

In our  view of the untyped $λ$-calculus, there is a
(variable!) set of variables and everything is parametric in this `variables
set'.
Similarly, in a general \transition monad $R$, there is a `variables object' $V$ in $𝐒𝐞𝐭^ℙ$
and everything is functorial in this variables object. In particular, we
have a placetaker object $T_R(V)$ in $𝐒𝐞𝐭^ℙ$ and a source (resp.\ target) state
object in $𝐒𝐞𝐭^𝕊$, both depending upon the variables object.

\mypar{The state functors $S₁$ and $S₂$} 

While in the $λ$-calculus,
states are the same as placetakers, in a general \transition monad,
they may differ, and more precisely both state objects are derived
from the placetaker object by applying the \alert{state functors}
$S₁,S₂∶ 𝐒𝐞𝐭^ℙ → 𝐒𝐞𝐭^𝕊$.  

\mypar{The transition structure  }

In standard $λ$-calculi, there is a (typed!) set of transitions, which
yields a graph on the set of terms. That is to say, if $V$ is the
variables object, and $LC(V)$ the placetaker object, there is a
transition object $\Red(V)$ equipped with two maps
$\src_V, \tgt_V∶ \Red(V) → LC(V)$.  Note that we consider
`proof-relevant' transitions here, in the sense that two different
transitions may have the same source and target
(see~\ref{ss:irrelevant-trans-monads} for the proof-irrelevant variant).

In a general \transition monad $R$, we still have a transition object
$\Red_R (V)$, which now lives in $𝐒𝐞𝐭^𝕊$, together with
state objects $S₁ ( T_R (V))$  and $S₂ ( T_R (V))$,
so that the pairing
$⟨\src_V,trg_V⟩$ forms a morphism
 $ \Red_R(V)  → S₁(T_R(V)) \times S₂(T_R(V))$.

 One main feature of \transition monads is that transitions are closed
 under substitution. Technically, this is realised by taking the
 transition object $\Red_R(V)$ to be a $T_R$-module, and that the
 morphism $ \Red_R(V) → S₁(T_R(V)) \times S₂(T_R(V))$ to be a
 $T_R$-module morphism.

\subsection{The definition of \transition monad}
\label{ss:def-transmon}
Here is our formal definition:

  \begin{defi}
    Given two sets $ℙ$ and $𝕊$, a  \alert{\transition monad}
    over $(ℙ,𝕊)$ consists of
    \begin{itemize}
     \item a finitary monad $T$ on $𝐒𝐞𝐭^ℙ$, called the
       \alert{placetaker} monad,
     \item two finitary functors $S₁, S₂∶ 𝐒𝐞𝐭^ℙ → 𝐒𝐞𝐭^𝕊$, called
       \alert{state} functors, and
    \item a \alert{transition structure} $ R
      \xrightarrow{⟨\src, \tgt⟩} S₁T  \times S₂T$ consisting of
      \begin{itemize}
      \item 
        a finitary $T$-module $R∶ 𝐒𝐞𝐭^ℙ → 𝐒𝐞𝐭^𝕊$, called the
        \alert{transition} module,
      \item a \alert{source} $T$-module morphism $\src∶ R → S₁ T$,
        recalling from Notation~\ref{not:freemodule} that $S₁T$ is the
        free $T$-module on $S₁$,
      \item a \alert{target} $T$-module morphism $\tgt∶ R → S₂ T$.
      \end{itemize}
    \end{itemize}
    \end{defi}

  \begin{defi}\label{def:omndef}
    For any sets $ℙ$ and $𝕊$, finitary monad $T$ over $𝐒𝐞𝐭^ℙ$, and
    finitary functors $S₁,S₂∶ 𝐒𝐞𝐭^ℙ → 𝐒𝐞𝐭^𝕊$, we let
    $𝐓𝐫𝐚𝐧𝐬𝐒𝐭𝐫𝐮𝐜𝐭_{ℙ,𝕊}(T,S₁,S₂)$ denote the class of \transition
    structures over $T$, $S₁$, and $S₂$.

    Furthermore, let $𝐓𝐫𝐚𝐧𝐬𝐌𝐧𝐝_{ℙ,𝕊}$ denote the coproduct
    $∑_{T,S₁,S₂} 𝐓𝐫𝐚𝐧𝐬𝐒𝐭𝐫𝐮𝐜𝐭_{ℙ,𝕊}(T,S₁,S₂)$.
  \end{defi}

  In Definition~§\ref{def:omndof}, we will upgrade these classes into
  categories with the same names.

\subsection{\Transition monads as relative monads }
\label{ss:transmon-relmon}
In our definition of \transition monad, we have required the two state
modules to take the form $S₁T$ and $S₂T$, but this is not essential in
our development. Moreover, it probably leaves some relevant examples
out of reach. We think in particular about the standard presentation
of the $π$-calculus as a labelled transition system. Our original
reason for this design choice was purely aesthetic: it ensures that
our \transition monads are \alert{relative} monads, as were the
reduction monads introduced in~\cite{AHLM19} (in particular the
untyped $λ$-calculus). These were monads relative to the `discrete
graph' functor from sets to graphs, and in our extended context, we
have to replace graphs by \alert{$S$-graphs}. Let us provide more
detail.

Let us first 
recall~\cite{DBLP:journals/corr/AltenkirchCU14} that, given any
functor $J∶ 𝐂 → 𝐃$, a \alert{monad relative to $J$}, or
\alert{$J$-relative monad}, consists of
\begin{itemize}
\item an object mapping $T∶ 𝐨𝐛(𝐂) → 𝐨𝐛(𝐃)$, together with
\item morphisms $η_X∶ J(X) → T(X)$, and
\item for each morphism $f∶ J(X) → T(Y)$, an \alert{extension}
  $f^⋆∶ T(X) → T(Y)$,
\end{itemize}
satisfying coherence conditions.  Any $J$-relative monad $T$ has an
underlying functor $𝐂 → 𝐃$, and is said \alert{finitary} when this
functor is. Note that a monad is nothing but a $J$-relative monad, for $J$
the identity endofunctor.

Now we define $S$-graphs:  
 
\begin{defi}\label{def:sgraphs}
  For any pair $S = (S₁,S₂)$ of functors $𝐒𝐞𝐭^ℙ → 𝐒𝐞𝐭^𝕊$, an \alert{$S$-graph} 
  over an object $V ∈ 𝐒𝐞𝐭^ℙ$ consists of
\begin{itemize}
\item an object $E$ (of \alert{edges})  in $𝐒𝐞𝐭^𝕊$, and
\item a morphism $∂∶ E → S₁(V)×S₂(V)$.
\end{itemize}
Accordingly, an \alert{$S$-graph} consists of an object $V ∈ 𝐒𝐞𝐭^ℙ$ and an $S$-graph
over $V$.
\end{defi}

We can now say  that  in a general \transition monad, transitions form
an $S$-graph over the placetaker object (the whole thing depending
upon the variables object…). Before proceeding, we must introduce the 
category of $S$-graphs:
 a morphism
 $G → G'$
 consists of a morphism for vertices
$f∶ V_G → V_{G'}$ together with a morphism for edges $g∶ E_G → E_{G'}$
making the following diagram commute. 
\begin{center}
  \diag|baseline=(m-1-1.base)|{%
    E_G \& E_{G'} \\
    S₁(V_G)×S₂(V_G) \& S₁(V_{G'})×S₂(V_{G'}) }{%
    (m-1-1) edge[labela={g}] (m-1-2) %
    edge[labell={∂_G}] (m-2-1) %
    (m-2-1) edge[labelb={S₁(f)×S₂(f)}] (m-2-2) %
    (m-1-2) edge[labelr={∂_{G'}}] (m-2-2) %
  }
\end{center}

\begin{prop}
  For any pair $S = (S₁,S₂)$ of functors $𝐒𝐞𝐭^ℙ → 𝐒𝐞𝐭^𝕊$, $S$-graphs
  form a category $S\Gph$.
\end{prop}

The proof is a straightforward verification.

We will consider monads relative to the following functors:

\begin{defi}
  For any pair $S = (S₁,S₂)$ of functors $𝐒𝐞𝐭^ℙ → 𝐒𝐞𝐭^𝕊$, the
  \alert{discrete $S$-graph} functor $J_S∶ 𝐒𝐞𝐭^ℙ → S\Gph$ maps any
  $V ∈ 𝐒𝐞𝐭^ℙ$ to the $S$-graph over $V$ with no edges.
\end{defi}

Now we are ready to deliver our characterisation of 
\transition monads as relative monads:

\begin{prop}\label{prop:synthana}
Given finitary functors $S₁,S₂∶ 𝐒𝐞𝐭^ℙ → 𝐒𝐞𝐭^𝕊$, \transition monads 
with state  functors $S₁$ and $S₂$ are exactly
monads relative to the discrete $S$-graph functor  
for $S = (S₁,S₂)$, such that the induced
functor $𝐒𝐞𝐭^ℙ → S\Gph$ is finitary.
\end {prop}

The proof consists merely in unfolding the definitions. 
Since we do not use this result, we do not give further details.

\subsection{The proof-irrelevant variant}
\label{ss:irrelevant-trans-monads}

Although we have chosen a  ``proof-relevant'' notion of \transition monad, 
we can sketch a presentation of a ``proof-irrelevant'' variant.

  \begin{defi}
    Let $𝐈𝐓𝐫𝐚𝐧𝐬𝐒𝐭𝐫𝐮𝐜𝐭_{ℙ,𝕊}(T,S₁,S₂)$ denote the subset of
    $𝐓𝐫𝐚𝐧𝐬𝐒𝐭𝐫𝐮𝐜𝐭_{ℙ,𝕊}(T,S₁,S₂)$ consisting of \transition monads
    $⟨\src,\tgt⟩∶R → S₁T × S₂T$ such that $⟨\src,\tgt⟩$ is a pointwise
    inclusion.
  \end{defi} 
  \begin{rem}
\label{rk:natural-retract}    
We have a natural retraction 
    \[𝐓𝐫𝐚𝐧𝐬𝐒𝐭𝐫𝐮𝐜𝐭_{ℙ,𝕊}(T,S₁,S₂) ↠ 𝐈𝐓𝐫𝐚𝐧𝐬𝐒𝐭𝐫𝐮𝐜𝐭_{ℙ,𝕊}(T,S₁,S₂)\rlap{,}\] which
maps \atransition structure $∂∶ R → S₁T × S₂T$ to the monomorphism
$\overline{R} ↪ S₁T × S₂T$ obtained from the (strong epi)-mono
factorisation; this factorisation
exists~\cite[Proposition~1.61]{Adamek} since the category of finitary
$𝐒𝐞𝐭$-valued $T$-modules is a presheaf
category~\cite[Definition~2.71]{AhrensPhD}.  In
Proposition~\ref{prop:irrelevant-trans-monads}, we will upgrade this
retraction into a coreflection of categories.
  \end{rem}

\section{Examples of \transition monads}\label{s:examples}
In this section, we present very informally
the announced examples of \transition
monads. This presentation should eventually be compared to the one via
signatures given in~§\ref{s:applications}.

\subsection{The call-by-value, simply-typed, big-step \lam-calculus}\label{ss:stlambda}
The notion of transition monad accounts for many different variants
of the \lam-calculus.  Let us detail the case of the  simply-typed,
call-by-value, big-step $λ$-calculus.  Most often, big-step semantics
describes evaluation of closed terms.
Here we consider a variant describing the evaluation
of open terms~\cite{cbvsolv,Lassencbv}.  In this setting, the main
subtlety lies in the fact that variables are only placeholders for
values.

We 
fix some set of base types, ranged over by $ι$ and define successively types,
values and terms, typing contexts, well-typed terms, and transitions, as follows.
\begin{defi}
  The set $ℙ$ of \alert{types}, ranged over by $A$, $B$, is defined
  inductively by  \[A, B \Coloneqq ι ｜ A → B\rlap{.}\]
\end{defi}
\begin{defi}
  For any set, say $X$, of \alert{variables}, we then define
  \alert{values}, ranged over by $v$, $w$, and general \alert{terms},
  ranged over by $e$, $f$ in a mutually inductive way as follows.
\begin{displaymath}
  \begin{array}{rcll}
    v,w & \Coloneqq & x ｜ λx.e \\
    e,f & \Coloneqq & v ｜ e\ f
  \end{array}
\end{displaymath}
Here, $x$ ranges over $X$, and Terms are considered equivalent modulo
$α$-conversion. Furthermore, one may define capture-avoiding
substitution, as usual.
\end{defi}
\begin{defi}
  A \alert{typing context} is a type-indexed family of sets, i.e., an
  object of $𝐒𝐞𝐭^ℙ$.  For any $Γ ∈ 𝐒𝐞𝐭^ℙ$, $A ∈ ℙ$, and $x ∉ Γ_A$, we
  let $Γ,x:A$ denote $Γ$ augmented with $x$ over $A$.
\end{defi}

\begin{rem}
  The extended context $Γ,x:A$ is isomorphic to $Γ + 𝐲_A$, where $𝐲$
  denotes the Yoneda embedding $ℙ ↪ 𝐒𝐞𝐭^ℙ$, viewing $ℙ$ as a discrete
  category.  Indeed, because $ℙ$ is discrete, $𝐲_A(B)$ is empty when
  $A≠B$, and $𝐲_A(A)$ is a singleton. Thus $Γ+𝐲_A$ is $Γ$, plus one
  element of type $A$.
 \end{rem}

\begin{defi}
  \alert{Well-typed} terms are inductively defined by the following rules.
\begin{mathpar}
  \inferrule{ }{Γ ⊢ x : A}~(x ∈ Γ_A) \and
  \inferrule{Γ, x : A ⊢ e : B}{Γ ⊢ λx.e : A → B} \and
  \inferrule{Γ ⊢ e : A → B \\ Γ ⊢ f : A}{Γ ⊢ e\ f  : B}
\end{mathpar}
\end{defi}

\begin{defi}
\alert{Transitions} are inductively defined by the following rules.
\begin{mathpar}
  \inferrule{ }{v ⇓ v} \and
  \inferrule{e₁ ⇓ λx.e₃ \\ e₂ ⇓ w \\ e₃[x ↦ w] ⇓ v}{e₁\ e₂ ⇓ v}
\end{mathpar}
\end{defi}

This calculus forms \atransition monad as follows.

\mypar{Placetakers and transition types} As foreshadowed by the notation,
because variables and values are  indexed by (simple) types, we take
$ℙ = 𝕊$ to be the set of types.

\mypar{Placetaker monad} The placetaker monad $T$ over $𝐒𝐞𝐭^ℙ$ is
given by well-typed values: given any $Γ ∈ 𝐒𝐞𝐭^ℙ$, the placetaker
object $T(Γ) ∈ 𝐒𝐞𝐭^ℙ$ assigns to each type $A$ the set $T(Γ)_A$ of
value typing derivations, i.e., the set of typing derivations with
conclusion of the form $Γ ⊢ v : A$.

\mypar{State functors}
In big-step semantics, transitions relate terms to values.  Hence,
we are seeking state functors $S₁,S₂∶ 𝐒𝐞𝐭^ℙ→𝐒𝐞𝐭^ℙ$ such that
$S₁(T(Γ))_A$ is the set of typing derivations of type $A$ with free
variables in $Γ$, and $S₂(T(Γ))_A$ is the subset of value typing
derivations therein.  For $S₂$, we should clearly take the identity
functor since $T$ consists of all value typing derivations.  For $S₁$,
we first observe that $λ$-terms can be described as \alert{application
  binary trees} whose leaves are values (internal nodes being typed
applications). More formally, let $S₁(Γ)_A$ denotes the set of typing
derivations with conclusion of the form $Γ ⊢_{S₁} b : A$ inductively
generated by the following rules.
\begin{mathpar}
  \inferrule{x ∈ Γ_A }{Γ ⊢_{S₁} x : A} \and
  \inferrule{Γ ⊢_{S₁} b : A→B \\ Γ ⊢_{S₁} b' : A}{Γ ⊢_{S₁} b\ b' : B}~⋅
\end{mathpar}
If $Γ$ is a typing context, $S₁(T(Γ))$ is indeed the set of general
typing derivations with variables in $Γ$.

\mypar{Transition module} Finally, the transition module 
maps any $Γ ∈ 𝐒𝐞𝐭^ℙ$ to the family, over all $A ∈ ℙ$, of transition
proofs of some $e ⇓ v$ with $e ∈ S₁(T(Γ))_A$ and $v ∈ T(Γ)_A$.  Such
transition proofs are stable under value substitution, so we obtain
\atransition monad.

\begin{rem}
  Transition proofs are \emph{not} stable under general substitution.
  E.g., the following proof
  \begin{mathpar}
    \inferrule*{
      \inferrule{ }{λx.x ⇓ λx.x} \\
      \inferrule{ }{y ⇓ y} \\
      \inferrule{ }{y ⇓ y}
    }{
      (λx.x)\ y ⇓ y
    }
  \end{mathpar}
  of $(λx.x)\ y ⇓ y$ becomes invalid if we replace $y$ with any
  non-value term $e$. Indeed, $e ⇓ e$ does not hold.  And in fact, the
  conclusion itself becomes invalid: we cannot have $(λx.x)\ e ⇓ e$ since
  evaluation results are all values.
\end{rem}

\subsection{The \lbmu-calculus}\label{ss:lbmu}
The $\overline{λ}μ$-calculus, introduced by
Herbelin~\cite{DBLP:phd/hal/Herbelin95}, models the computational
contents of cut elimination in the sequent calculus.  Following Vaux's
presentation~\cite[§2.4.4]{Iouri}, its grammar is given by
\begin{center}
  $\begin{array}[t]{c}
    \mbox{Commands} \\
    c  \Coloneqq  ⟨e|π⟩
  \end{array}$
\hfil
  $\begin{array}[t]{c}
    \mbox{Programs} \\
    e  \Coloneqq  x ｜ μα.c ｜ λx.e
  \end{array}$
    \hfil
    $\begin{array}[t]{c}
    \mbox{Stacks} \\
    π  \Coloneqq  α ｜ e⋅π,
      \end{array}$
\end{center}
where $x$ and $α$ range over two disjoint sets of variables, called
\alert{stack} and \alert{program} variables, respectively.  Both
constructions $μ$ and $λ$ bind their variable in the body.  Reduction
 is generated by the following two
basic transition rules concerning commands:
\begin{mathpar}
{⟨μα.c|π⟩ → c[α↦π]} \and
{⟨λx.e|e'⋅π⟩ → ⟨e[x↦e']|π⟩}.
\end{mathpar}
Reduction may occur ``everywhere'', so it involves programs and stacks as well.

Let us show how this gives rise to \atransition monad.

\mypar{Placetaker types}
In the transition rules above, we see that placetakers may be programs  or stacks.
So, we take  two placetaker types:
$ℙ ≔ 2 = \ens{𝐩,𝐬}$.  A variables object is an element of $𝐒𝐞𝐭^ℙ$, that
is, a pair of sets: the first one gives the available free program
variables, and the second one the available free stack variables.

\mypar{Placetaker monad} The syntax may be viewed as a monad
$T∶ 𝐒𝐞𝐭^ℙ→𝐒𝐞𝐭^ℙ$: given a variables object $X = (X_𝐩,X_𝐬) ∈ 𝐒𝐞𝐭^ℙ$, the
placetaker object $(T(X)_𝐩, T(X)_𝐬)∈𝐒𝐞𝐭^ℙ$ consists of the sets of
program and stack terms with free variables in $X$, up to bound
variable renaming. As usual, monad multiplication is given by
capture-avoiding substitution.

\mypar{Transition types and state functors} As mentioned above,
transitions involve programs and stacks as well as commands.  Thus, we
take three transition types: $𝕊 = 3= \ens{𝐜,𝐩,𝐬}$.  Furthermore,
commands are pairs of a program and a stack, so that, setting
$S₁(A) = S₂(A) = (A_𝐩 × A_𝐬, A_𝐩, A_𝐬)$, we get
$Sᵢ(T(X)) = (T(X)_𝐩 × T(X)_𝐬, T(X)_𝐩, T(X)_𝐬)$ for $i = 1,2$, as
desired.

\mypar{Transition module} Finally, transitions with free variables in
$X$ form a triple of graphs with vertices respectively in
$T(X)_𝐩 × T(X)_𝐬$ (the set of commands taking free variables in $X$), $T(X)_𝐩$, and $T(X)_𝐬$.
This family is
natural in $X$ and commutes with substitution, hence forms a
$T$-module morphism, which completes our \transition monad.

\subsection{The \pii-calculus}\label{ss:expi}
For an example involving equations on placetakers, let us recall the
following standard presentation of (a simple variant of) the
$π$-calculus~\cite{DBLP:books/daglib/0004377}. The syntax for
\alert{processes} is given by 
\[P,Q \Coloneqq x ｜ 0 ｜ (P|Q) ｜ νa.P ｜ \abar⟨b⟩.P ｜ a(b).P\rlap{,}\]
where $x$ ranges over \alert{process variables},
$a$ and $b$ range over \alert{channel names}, and $b$ is bound
in $νb.P$ and $a(b).P$.  Processes are identified when related by
the smallest context-closed equivalence relation $≡$ satisfying
\begin{center}
 $0|P ≡ P$
\hfil $P|Q ≡ Q|P$ \hfil $P|(Q|R) ≡ (P|Q)|R$ \hfil
$(νa.P)|Q ≡ νa.(P|Q)$,
\end{center}
where in the last equation $a$ should not occur free in $Q$.
Transitions are then given by the following rules.
\begin{mathpar}
\inferrule{ }{\abar⟨b⟩.P | a(c).Q ⟶ P|(Q[c↦b])} \and
\inferrule{P ⟶ Q}{P|R ⟶ Q|R} \and \inferrule{P ⟶ Q}{νa.P ⟶ νa.Q} 
\end{mathpar}
\begin{rem}
  Please note that there are no context rules for inputs or outputs,
  so that nothing happens under them.
\end{rem}
The $π$-calculus gives rise to \atransition monad as follows.

\mypar{Placetaker types} We consider two placetaker
types, one for channels and one for processes. Hence,
$ℙ = 2 = \ens{𝐜,𝐩}$.

\mypar{Placetaker monad} Then, the syntax may be viewed as a monad
$T∶ 𝐒𝐞𝐭^ℙ→𝐒𝐞𝐭^ℙ$: given a variables object $X = (X_𝐜,X_𝐩) ∈ 𝐒𝐞𝐭^ℙ$, the
placetaker object $T(X)=(X_𝐜, T(X)_𝐩) ∈ 𝐒𝐞𝐭^ℙ$ consists of the sets of
channels and processes with free variables in $X$ (modulo $≡$). Note
that $T(X)_𝐜=X_𝐜$ as there is no operation on channels.

\mypar{Transition type and state functors} Transitions relate processes, so we take
$𝕊 = 1$ and $S₁(X) = S₂(X) = X_𝐩$.

\mypar{Transition module} Transitions are generated by the above three
rules, and obviously stable under substitution. We thus obtain
\atransition monad.

Let us now describe a second way of organising the $π$-calculus as
\atransition monad, this time over a single placetaker type.
For this, we consider the same calculus, albeit without process variables, so that the syntax
becomes
\[P,Q ∶∶= 0 ｜ (P|Q) ｜ νa.P ｜ \abar⟨b⟩.P ｜ a(b).P\rlap{.}\]

\mypar{Placetaker} In this variant of the $π$-calculus, all that needs
substitution is channels, so we set $ℙ = 1$.

\mypar{Placetaker monad} Since the syntax contains no channel
constructor, the placetaker monad is merely the identity monad.

\mypar{Transition types and state functors} However, since transitions relate processes, we
need to fit the syntax into the state functors. We thus take $𝕊 = 1$
and $S₁(X) = S₂(X)$ to be the set of processes with free channels in
$X$.

\mypar{Transition module} Transitions are generated as above, and
stable under channel renaming, hence again form \atransition monad.

\begin{rem}
  Neither of our presentations would work in the presence of the
  mismatch operator~\cite[p13]{DBLP:books/daglib/0004377}, which
  breaks stability of reduction under renaming.
\end{rem}

  \subsection{Positive GSOS systems}\label{ss:gsos}
  An example involving labelled transitions is given by Positive GSOS
  systems~\cite[p246]{GSOS}. They specify labelled transition systems,
  in which transitions have the shape $\labredrule{a}{e}{f}$, where
  $a$ is drawn from some fixed set $𝔸$ of \alert{labels}.  Let us
  further fix a set $O$ of \alert{operations} with arities in $ℕ$ and
  an infinite set of \alert{variables}, ranged over by $x$, $y$,…

\begin{defi}
A \alert{Positive GSOS rule} has the shape
\begin{equation}
\inferrule{
  … \\ \labredrule{a_{i,j}}{xᵢ}{y_{i,j}} \\ … \\ (i ∈ n, j ∈ nᵢ)}{\labredrule{c}{\ope(x₁,…,xₙ)}{e}}\rlap{,}
\label{eq:GSOS}
\end{equation}
where
\begin{itemize}
\item $\ope ∈ O$ is an operation with arity $n$,
\item for all $i ∈ n$, $nᵢ$ is a natural number,
\item the variables $xᵢ$ and $y_{i,j}$ are all distinct,  and
\item $e$ is an expression potentially depending on all the variables.
\end{itemize}
A \alert{Positive GSOS system} is a set of Positive GSOS rules.
\end{defi}
The semantics of a Positive GSOS rule is that of a ``rule scheme'', in the following sense.
\begin{defi}
  The labelled transition relation \alert{generated} by a Positive GSOS
  system is the smallest $𝔸$-labelled transition system on expressions
  generated by $O$, such that for all rules~\eqref{eq:GSOS}, and
  expressions $e₁,…,eₙ$ and $e_{i,1},…,e_{i,nᵢ}$ for all $i ∈ n$, if
  $\labredrule{a_{i,j}}{eᵢ}{e_{i,j}}$ for all $i ∈ n$ and $j ∈ nᵢ$,
  then
  \[\labredrule{c}{\ope(e₁,…,eₙ)}{e[…, xᵢ ↦ eᵢ, …, y_{i,j} ↦ e_{i,j}, …]}.\]
\end{defi}

Otherwise said, each rule scheme~\eqref{eq:GSOS} induces a rule 
\begin{mathpar}
  \inferrule{… \\ \labredrule{a_{i,j}}{eᵢ}{e_{i,j}} \\ … \\ (i ∈ n, j ∈ nᵢ)
  }{
    \labredrule{c}{\ope(e₁,…,eₙ)}{e[ …,xᵢ ↦ eᵢ, …, y_{i,j} ↦ e_{i,j}, …]}
  }~·
\end{mathpar}
\begin{rem}
  The general notion of GSOS system includes negative rules, which
  means rules that may have premises of the shape $xᵢ \nxto{a}$\ .
  Their semantics is significantly more involved, so we leave their
  integration as an open problem.
\end{rem}
Each Positive GSOS system yields \atransition monad as follows.

\mypar{Placetaker and transition types} We take
$ℙ=1$, because we are in an untyped setting, and $𝕊 = 1$ because
states are terms.

\mypar{Placetaker monad} The selected family of operations (or rather
arities) specifies the term monad $T$.

\mypar{State functors} In order to take labels into account, we take
$S₁(X) = X$ and $S₂(X) = 𝔸×X$. Transitions thus form a set over
$X×(𝔸×X)$ as desired.

\begin{rem}
  We could as well take $S₁(X) = 𝔸×X$ and $S₂(X)=X$, as ultimately
  only the product $S₁(X)×S₂(X)$ matters.  
\end{rem}

\mypar{Transition structure} As before, we take as transitions the set
of all proofs generated by the rules, which is indeed stable under
substitution by construction.

\subsection{The differential \lam-calculus}\label{ss:lambdadiff}
Let us finally sketch how  differential
$λ$-calculus~\cite{LambdaDiff} provides a further example with
$S₁≠S₂$. Following Vaux~\cite[§6]{Iouri}, its syntax may be defined by
  \begin{displaymath}
    \begin{array}[t]{rcll}
      e,f,g & \Coloneqq & x ｜ λx.e ｜ e⟨U⟩ ｜ D e ⋅ f & \mbox{(terms)} \\
      U,V & \Coloneqq & 0 ｜ e + U & \mbox{(multiterms)}\rlap{,}
    \end{array}
  \end{displaymath}
  where terms and multiterms are considered equivalent up to the
  following equations.
  \begin{mathpar}
    e + e' + U = e' + e + U \and D (D e · f) · g = D (D e · g) · f.
  \end{mathpar}

  The definition of transitions is based on two auxiliary
  constructions:
    \begin{enumerate}
    \item \alert{Unary multiterm substitution}
      $e[x↦ U]$
      of a multiterm $U$ for a variable $x$ in a term $e$,
      which returns a multiterm (not to be confused with unary monadic
      substitution, which handles the particular case where $U$ is a
      mere singleton).

      \item \alert{Partial derivative} $\frac{∂e}{∂x} ⋅ U$
      of a term $e$
      w.r.t.\ a term variable $x$ along a multiterm $U$. This again
      returns a multiterm.
    \end{enumerate}
    Both are defined by induction on $e$ (see~\cite[Definition~6.3
    and~6.4]{Iouri}).

    We may now define the transition relation as the smallest
    context-closed relation satisfying the rules below.
    \begin{center}
      \hfil $\inferrule{}{(λx.e)⟨U⟩ → e[x ↦ U]}$ \hfil
      $\inferrule{}{D(λx.e)⋅f → λx.\left (\frac{∂e}{∂x}⋅f \right )}$
  \end{center}
  The compact formulation of the second rule
  relies on the abbreviation $ λx.(e₁+…+eₙ) ≔ λx.e₁ + … + λx.eₙ$. 

  Let us now sketch how  this forms \atransition monad.
  
  \mypar{Placetaker monad} Terms induce a monad $T$ on $𝐒𝐞𝐭$, which we
  take as the placetaker monad (hence $ℙ=1$).

  \mypar{State functors} Transitions relate terms to multiterms, hence
  $𝕊=1$, $S₁$ is the identity, and $S₂=\oc$ is the functor mapping any
  set $X$ to the set of (finite) multisets over $X$.

  \mypar{Transition structure} Transitions are stable under
  substitution by terms, hence we again have \atransition monad.

\section{Signatures for \transition monads}\label{s:sigtrans}
In the previous section, we have shown that several significant
(abstract) programming languages may be organised as \transition
monads. We are now interested in specifying such languages by
signatures. We first introduce registers, as announced in the
introduction, and then our register for \transition monads.

\subsection{Signatures registers}
In this section, we introduce signatures registers, which are a
formalisation of the notion of signature, at least in the context of
initial algebra semantics.

The idea is that each signature $S$ over a fixed category $𝐂$ should
give rise to a particular category $S\alg$, equipped with a
``forgetful'' functor to $𝐂$, the specified object $\spec{S}$ being
the carrier of the initial object of $S\alg$:

\begin{defi}
  \Asemsig over a category $𝐂$ consists of a category $𝐄$ with an
  initial object, equipped with a functor $𝐄 → 𝐂$.  We denote by
  $\SemSig_𝐂$ the class of \semsigs over $𝐂$.
\end{defi}
\begin{nota}
  We denote the components of any \semsig $S$ over a category $𝐂$ by
  $S\alg$ and $𝐔_S$, so that $S$ is precisely $𝐔_S∶ S\alg → 𝐂$.
  Accordingly, we call objects of $S\alg$ \alert{$S$-algebras}, or
  \alert{models} of $S$.
\end{nota}

\begin{defi}
  A \alert{register} $𝐑$ for a given category $𝐂$ consists of
  \begin{itemize}
  \item a class $𝐒𝐢𝐠_𝐑$ of \alert{signatures}, and
  \item a \alert{semantics} map $⟦-⟧_𝐑 ∶ 𝐒𝐢𝐠_𝐑 → \SemSig_𝐂$.
  \end{itemize}
\end{defi}
  \begin{terminology}\label{term:spec}
    We say that any $S ∈ 𝐒𝐢𝐠_𝐑$ is a \alert{signature for
      $\spec{S}≔𝐔_{S}(0)$} ($0$ here denotes the initial object in
    $S\alg$), or alternatively that $S$ \alert{specifies} $\spec{S}$.
    Finally, when we define our registers below, we first introduce
    signatures and associate a functor $𝐄 → 𝐂$ to each signature.  It
    then remains to prove that $𝐄$ has an initial object: as mentioned
    in~§\ref{s:intro}, we call such proofs \alert{validity proofs}.
\end{terminology}
Most of our registers will be monadic in the following sense.
\begin{defi}
  A register $𝐑$ is \textbf{monadic} when the \semsig $𝐄 → 𝐂$
  associated to any signature in $𝐒𝐢𝐠_𝐑$ is finitary and monadic.
\end{defi}

\begin{rem}
  Strictly speaking, we should call this ``finitary monadic''. We omit
  the ``finitary'' for readability. 
\end{rem}

\begin{nota}\label{not:star}
  When a register is monadic, we denote by $S^⋆$ the monad induced by
  any signature $S$, in the sense that we have an isomorphism
  $S\alg ≅ S^⋆\alg$ of categories over $𝐂$ (i.e., which commutes with
  forgetful functors). In such cases we have
  \[S^⋆(0_𝐂) ≅ \spec{S} = 𝐔_S(0_{S\alg}).\]
\end{nota}

Let us now introduce a simple notion of morphism between registers.
\begin{defi}\label{def:compilation}
  A \emph{compilation} from a register $R₁$ on a category $𝐂$ to a
  register $R₂$ on the same category is a map $c∶𝐒𝐢𝐠_{R₁} → 𝐒𝐢𝐠_{R₂}$
  preserving the semantics up to isomorphism, in the sense that for
  any $Σ ∈ 𝐒𝐢𝐠_{R₁}$, there is an isomorphism 
  $⟦ Σ ⟧_{R₁} ≅ ⟦ c(Σ) ⟧_{R₂}$ as objects of $𝐂𝐀𝐓/𝐂$.
  We say that $R₁$ is a \alert{subregister} of $R₂$ if there
  exists a compilation of $R₁$ to $R₂$.
\end{defi}

Let us finish this subsection by recasting a well-known fact as the
definition of a register.  The well-known fact is the following.
\begin{propC}[{\cite[p62]{Reiterman77}}]\label{prop:reiterman}
  For any finitary endofunctor $F$ on a cocomplete category $𝐂$, the
  forgetful functor $F\alg → 𝐂$ is monadic, and the left adjoint maps
  any object $C ∈ 𝐂$ to the initial algebra $F^*(C)$ of the functor
  $A ↦ C + F(A)$, i.e., the least fixed point $μA. (C + F(A))$, with
  $F$-algebra structure given by
  \[F(F^*(C)) ↪ C + F(F^*(C)) ≅ F^*(C).\]
\end{propC}
And here comes the register:
\begin{defi}\label{def:endos}
  For a given cocomplete category $𝐂$, we define the monadic register
  $𝐄𝐅_𝐂$, called the \alert{endofunctor register}, as
  follows.
  \begin{description}
  \item[Signatures] A signature is a finitary endofunctor on $𝐂$.
  \item[Semantics] The \semsig associated to any finitary endofunctor
    $F$ is the forgetful functor $U_F∶ F\alg → 𝐂$.
\end{description}
\end{defi}
 \begin{proof}[{\rm \bf Validity proof:}\nopunct]
   By Proposition~\ref{prop:reiterman}.
 \end{proof}

 \begin{rem}\label{rk:ambiguity}
   Let $F$ be any finitary endofunctor on a cocomplete category $𝐂$.
   Please note the difference between $F^⋆(0)$ and $F^*(0)$: $F^⋆(0)$
   denotes the object specified by $F$ \emph{qua} signature of $𝐄𝐅_𝐂$,
   while $F^*(0)$ denotes the (carrier of the) initial $F$-algebra.
   In this case, of course, the denotations coincide, but this will no
   longer be the case, for instance, in~§\ref{ss:ia:mnd}.  There,
   $F^⋆(0)$ will denote the initial \alert{$F$-monoid}, while $F^*(0)$
   will still denote the initial $F$-algebra.
 \end{rem}

 Let us conclude by naming all registers defined by compilation into
 $𝐄𝐅_𝐂$.
 \begin{defi}\label{d:endofunctorial}
   We call \alert{endofunctorial} all subregisters of $𝐄𝐅_𝐂$.
 \end{defi}

\subsection{A register for \transition monads}
Our goal is to define a register for \transition monads.  Thus, we
should at least organise them into a category in the first place:
  \begin{defi}\label{def:omndof}
    For any sets $ℙ$ and $𝕊$, finitary monad $T$ over $𝐒𝐞𝐭^ℙ$, and
    finitary functors $S₁,S₂∶ 𝐒𝐞𝐭^ℙ → 𝐒𝐞𝐭^𝕊$, let
    \[𝐓𝐫𝐚𝐧𝐬𝐒𝐭𝐫𝐮𝐜𝐭_{ℙ,𝕊}(T,S₁,S₂) = T\Mod_f (𝐒𝐞𝐭^𝕊) /S₁T×S₂T \]denote
    the slice of the category of finitary, $𝐒𝐞𝐭^𝕊$-valued $T$-modules
    over $S₁T×S₂T$.

    And for any sets $ℙ$ and $𝕊$, we let
    $𝐓𝐫𝐚𝐧𝐬𝐌𝐧𝐝_{ℙ,𝕊} ≔ ∑_{T,S₁,S₂} 𝐓𝐫𝐚𝐧𝐬𝐒𝐭𝐫𝐮𝐜𝐭_{ℙ,𝕊}(T,S₁,S₂)$.
  \end{defi}
  \begin{rem}\label{rk:induction}
    Reduction monads~\cite{AHLM19} correspond to the case when
    $S₁=S₂ = \Id $, with $ℙ=𝕊 = 1$, but contrary to the present work, morphisms there
    can live between reduction monads with different underlying
    monads. We don't need such morphisms in the present work because 
     we enforce that models of a signature share the same underlying monad. This allows for
    a simpler notion of signature, at the cost of reducing the scope
    of the recursion principle.
  \end{rem}

  In the coming sections, we will introduce
  \begin{itemize}
  \item a register $𝐑𝐞𝐠𝐌𝐧𝐝_f(𝐒𝐞𝐭^ℙ)$ for finitary monads on the category $𝐒𝐞𝐭^ℙ$,
  \item a register $𝐑𝐞𝐠[𝐒𝐞𝐭^ℙ,𝐒𝐞𝐭^𝕊]_f$ for finitary functors $𝐒𝐞𝐭^ℙ → 𝐒𝐞𝐭^𝕊$, and
  \item for any finitary monad $T$ and functors $S₁$ and $S₂$, a
    register $𝐑𝐞𝐠𝐓𝐫𝐚𝐧𝐬𝐒𝐭𝐫𝐮𝐜𝐭_{ℙ,𝕊}(T,S₁,S₂)$ for
    \transition structures over $T$, $S₁$, and $S₂$.
  \end{itemize}
  
  Assuming this is done, we may already give our register for \transition monads:
  \begin{defi}\label{d:omndreg}
  We define the register $𝐑𝐞𝐠𝐓𝐫𝐚𝐧𝐬𝐌𝐧𝐝_{ℙ,𝕊}$ for the category $𝐓𝐫𝐚𝐧𝐬𝐌𝐧𝐝_{ℙ,𝕊}$ as follows.
  \begin{description}
  \item[Signatures] A signature, which we call
    a \alert{\transition signature}, consists of
    \begin{itemize}
    \item a signature $Σ$ of $𝐑𝐞𝐠𝐌𝐧𝐝_f(𝐒𝐞𝐭^ℙ)$, specifying a finitary monad $T$ on $𝐒𝐞𝐭^ℙ$,
      
    \item signatures $Σ₁$ and $Σ₂$ of $𝐑𝐞𝐠[𝐒𝐞𝐭^ℙ,𝐒𝐞𝐭^𝕊]_f$, specifying
      functors $S₁,S₂∶ 𝐒𝐞𝐭^ℙ → 𝐒𝐞𝐭^𝕊$, and
    \item a signature $Σ_{\Red}$ of $𝐑𝐞𝐠𝐓𝐫𝐚𝐧𝐬𝐒𝐭𝐫𝐮𝐜𝐭_{ℙ,𝕊}(T,S₁,S₂)$.
    \end{itemize}
  \item[Semantics] The \semsig associated to a signature $(Σ,Σ₁,Σ₂,Σ_\Red)$ is
    \[Σ_\Red\alg \xto{𝐔_{Σ_\Red}} 𝐓𝐫𝐚𝐧𝐬𝐒𝐭𝐫𝐮𝐜𝐭_{ℙ,𝕊}(T,S₁,S₂) ↪ 𝐓𝐫𝐚𝐧𝐬𝐌𝐧𝐝_{ℙ,𝕊}.\]
  \end{description}
\end{defi}
 \begin{proof}[{\rm \bf Validity proof:}\nopunct]
   We need to prove that $Σ_\Red\alg$ has an initial object; but this
   follows from $𝐑𝐞𝐠𝐓𝐫𝐚𝐧𝐬𝐒𝐭𝐫𝐮𝐜𝐭_{ℙ,𝕊}(T,S₁,S₂)$ being a register, which will be
   proved below.
 \end{proof}

 It remains to introduce the registers $𝐑𝐞𝐠𝐌𝐧𝐝_f(𝐒𝐞𝐭^ℙ)$ for monads,
 $𝐑𝐞𝐠[𝐒𝐞𝐭^ℙ,𝐒𝐞𝐭^𝕊]_f$ for functors, and
 $𝐑𝐞𝐠𝐓𝐫𝐚𝐧𝐬𝐒𝐭𝐫𝐮𝐜𝐭_{ℙ,𝕊}(T,S₁,S₂)$ for transition structures. The most
 novel is clearly the latter. It is furthermore
 independent from the others, so we introduce it first.

 For the reader's convenience, we list our registers in
 Figure~\ref{fig:registers}, together with corresponding categories.
 \begin{figure}[tp]
   {\renewcommand{\tabcolsep}{3pt}
   \begin{tabular}{cccc}
     \toprule
     Main Register & Signatures &  Base category & Where \\ \midrule
     $𝐑𝐞𝐠𝐓𝐫𝐚𝐧𝐬𝐌𝐧𝐝_{ℙ,𝕊}$ & \Transition signature & $𝐓𝐫𝐚𝐧𝐬𝐌𝐧𝐝_{ℙ,𝕊}$ & \ref{d:omndreg} \\
     $𝐑𝐞𝐠𝐈𝐓𝐫𝐚𝐧𝐬𝐌𝐧𝐝_{ℙ,𝕊}$ & \emph{idem} & $𝐈𝐓𝐫𝐚𝐧𝐬𝐌𝐧𝐝_{ℙ,𝕊}$ & \ref{def:irrelreg} \\
     $𝐑𝐞𝐠𝐓𝐫𝐚𝐧𝐬𝐒𝐭𝐫𝐮𝐜𝐭_{ℙ,𝕊}(T,S₁,S₂)$ & Families of rules & $𝐓𝐫𝐚𝐧𝐬𝐒𝐭𝐫𝐮𝐜𝐭_{ℙ,𝕊}(T,S₁,S₂)$ & \ref{def:regtransstruct} \\
     $𝐑𝐞𝐠𝐌𝐧𝐝_f(𝐒𝐞𝐭^ℙ)$ & Eq.\ modular signatures & $𝐌𝐧𝐝_f(𝐒𝐞𝐭^ℙ)$ & \ref{def:modmnd}\\
     $𝐑𝐞𝐠[𝐒𝐞𝐭^ℙ,𝐒𝐞𝐭^𝕊]_f$ & Eq.\ facet-based signatures & $[𝐒𝐞𝐭^ℙ,𝐒𝐞𝐭^𝕊]_f$ &
     \ref{def:PModFun} \\ \midrule
     Auxiliary Register & Signatures & Base category & Where \\ \midrule
     $𝐑^*$  & Families of sig's of $𝐑$ & Same as $𝐑$ & \ref{def:copairingreg} \\
     $𝐑𝐞𝐠𝐓𝐫𝐚𝐧𝐬𝐒𝐭𝐫𝐮𝐜𝐭⁰_{ℙ,𝕊}(T,S₁,S₂)$ & Rules  & $𝐓𝐫𝐚𝐧𝐬𝐒𝐭𝐫𝐮𝐜𝐭_{ℙ,𝕊}(T,S₁,S₂)$ & \ref{def:Rule} \\
     $𝐑𝐞𝐠⁰(𝐂/X)$ & Abstract simple rules & $𝐂/X$ & \ref{def:rule0} \\
     $𝐑𝐞𝐠¹(𝐂/X)$ & Abstract medium rules & $𝐂/X$ & \ref{def:MSlice} \\
     $𝐑𝐞𝐠(𝐂/X)$ & Abstract rules & $𝐂/X$ & \ref{def:Slice} \\
     $𝐑𝐞𝐠𝐌𝐧𝐝⁰_f(𝐒𝐞𝐭^ℙ)$ & Modular signatures & $𝐌𝐧𝐝_f(𝐒𝐞𝐭^ℙ)$ & \ref{def:elem} \\
     $𝐑𝐞𝐠⁰[𝐒𝐞𝐭^ℙ,𝐒𝐞𝐭^𝕊]_f$ & Facet-Based signatures & $[𝐒𝐞𝐭^ℙ,𝐒𝐞𝐭^𝕊]_f$
                                                 & \ref{def:elemfunc} \\
                                                 \midrule
     General Register & Signatures & Base category & Where \\ \midrule
     $𝐄𝐅_𝐂$ & Finitary endofunctors & $𝐂$ &
                                                           \ref{def:endos} \\
     $𝐄𝐒_𝐂$ & Equational systems & $𝐂$ & \ref{def:eqsysreg} \\
     $𝐏𝐒𝐄𝐅_𝐂$ & Pointed strong endos &
                                              $𝐌𝐨𝐧(𝐂)$ &
                                                         \ref{def:fpt} \\
     $𝐌𝐄𝐒_𝐂$ & Monoidal eq.\ systems &  $𝐌𝐨𝐧(𝐂)$   & \ref{def:moneqsysreg} \\
     \bottomrule
   \end{tabular}
   }

   \smallskip
   For $𝐑^*$, $𝐑$ is assumed to be an endofunctorial register
   (Definition~\ref{d:endofunctorial}).
   \caption{Registers}
   \label{fig:registers}
 \end{figure}

 \section{Registers for transition structures}\label{s:slices}
 In this section, we define the register
 $𝐑𝐞𝐠𝐓𝐫𝐚𝐧𝐬𝐒𝐭𝐫𝐮𝐜𝐭_{ℙ,𝕊}(T,S₁,S₂)$ for \transition structures on fixed
 $T,S₁,S₂$.  Since these are $𝐒𝐞𝐭^𝕊$-valued $T$-modules over $S₁T×S₂T$,
 this register should specify objects of the slice
 category $T\Mod_f(𝐒𝐞𝐭^𝕊)/S₁T×S₂T$.  We will in fact design a register for
 more general slice categories $𝐂 / X$, where $𝐂$ is a locally
 finitely presentable category and $X$ an object of $𝐂$.  The desired register
 $𝐑𝐞𝐠𝐓𝐫𝐚𝐧𝐬𝐒𝐭𝐫𝐮𝐜𝐭_{ℙ,𝕊}(T,S₁,S₂)$ will then be obtained as an instance
 by taking $𝐂 = T\Mod_f$ and $X = S₁T×S₂T$.

 In~§\ref{ss:small}, we first present a basic register $𝐑𝐞𝐠⁰(𝐂/X)$
 that would work for the untyped case without variable binding. Then,
 in~§\ref{ss:binding}, we extend $𝐑𝐞𝐠⁰(𝐂/X)$ to a more powerful
 register $𝐑𝐞𝐠¹(𝐂/X)$ that deals with variable binding.  Observing
 that this register is slightly heavy to use in the typed setting, we
 design a more convenient variant, called $𝐑𝐞𝐠(𝐂/X)$.  Finally,
 in~§\ref{sss:formatrules}, we focus on instances of this register to
 categories of the form $𝐂 = T\Mod_f/(S₁T × S₂T)$.  In this case, we
 introduce special notation, which allows us to write signatures
 essentially as in usual operational semantics literature.  Finally,
 in~§\ref{ss:irr-cats}, we derive a register for the proof-irrelevant
 variant of \transition monads.

\subsection{Small register for slice categories}\label{ss:small}
In this section, we introduce a first, limited register for slice
categories.
\begin{exa}\label{ex:stapp}
  Let $ℙ$ denote the set of simple types over a given set of basic
  types, as in~§\ref{ss:stlambda}, and consider the arity for
  application, i.e., the endofunctor $Σ_{\app}∶ 𝐒𝐞𝐭^ℙ → 𝐒𝐞𝐭^ℙ$ defined
  by \[Σ_{\app}(X)(A) = ∑_B X(B→A)×X(B)\rlap{.} \]Across the
  equivalence $𝐒𝐞𝐭^ℙ ≃ 𝐒𝐞𝐭/ℙ$, one way of presenting this endofunctor,
  which is perhaps closer to the syntactic inference rule for
  application, is as the span
  \[ℙ² \xot{⟨\arr,π₂⟩} ℙ² \xto{π₁} ℙ\rlap{,}\]
  where $\arr(A,B) = (B→A)$.
  Indeed, $Σ_{\app}(X)$ corresponds to
  \begin{itemize}
  \item taking the product of $X → ℙ$ with itself in the arrow category,
  \item pulling back along $⟨\arr,π₂⟩$, and
  \item postcomposing with $π₁$.
  \end{itemize}
  To see this, let us observe that after pulling back, we obtain a
  family $X'$ over $ℙ²$ such that $X'(A,B) = X(B→A)×X(B)$.
  Postcomposing, we take the disjoint union over $B$ as desired.
\end{exa}

Generalising from this example, we obtain the following ``small''
register.
\begin{defi}\label{def:rule0}
  For any locally finitely presentable category $𝐂$ and object
  $X ∈ 𝐂$, we define the endofunctorial register $𝐑𝐞𝐠⁰(𝐂/X)$ for the
  slice category $𝐂/X$ as follows.
\begin{description}
\item[Signatures]
  A signature consists of:
  \begin{itemize}
  \item a \alert{metavariable} object $V$;
    \item a list of \alert{premise} morphisms $(V\xrightarrow{sᵢ} X)_{i∈n}$
      denoted by $V \xrightarrow{\vec{s}} Xⁿ$;
      \item a \alert{conclusion} morphism $V\xrightarrow{t} X$. 
  \end{itemize}
\item[Semantics] A model of a signature
  $(Xⁿ \xleftarrow{\vec{s}} V \xrightarrow{t} X)$ is an algebra for
  the endofunctor mapping any $Y \xrightarrow{p} X$ to
  $Y' → V \xrightarrow{t} X$, where $Y' → V$ denotes the following
  pullback.
  \begin{equation*}
    \Diag{%
      \pbk{m-2-1}{m-1-1}{m-1-2} %
    }{%
      Y' \& Yⁿ \\
       V \& Xⁿ
    }{%
      (m-1-1) edge[labela={}] (m-1-2) %
      edge[labell={}] (m-2-1) %
      (m-2-1) edge[labelb={\vec{s}}] (m-2-2) %
      (m-1-2) edge[labelr={pⁿ}] (m-2-2) %
    }%
  \end{equation*}
  Morphisms of models are morphisms of algebras.
\end{description}
\end{defi}
\begin{proof}[{\rm \bf Validity proof:}\nopunct]
  All we have to show is that the induced endofunctor is finitary.
  But this functor is a composite of three functors
  \[𝐂/X \xto{(-)ⁿ/X} 𝐂/Xⁿ \xto{\vec{s}^*} 𝐂/V \xto{t_!} 𝐂/X\rlap{,}\]
  where
  \begin{itemize}
  \item $(-)ⁿ/X$ denotes $n$-fold self-product in the arrow category,
  \item $\vec{s}^*$ denotes pullback along $\vec{s}$, and
  \item $t_!$ denotes postcomposition with $t$.
  \end{itemize}
  The last two functors, as left adjoints, are cocontinuous, hence
  finitary. Finally, the forgetful functor $𝐂/Xⁿ → 𝐂$ creates
  colimits, so it suffices to show that the composite functor
  $𝐂/X \xto{(-)ⁿ/X} 𝐂/Xⁿ → 𝐂$ is finitary.  But this in turn is the
  composite of $𝐂/X → 𝐂 \xto{(-)ⁿ} 𝐂$.  Because the first component
  creates, hence preserves all colimits, it suffices to show that
  $(-)ⁿ$ is finitary, which holds since filtered colimits commute with
  finite limits in locally finitely presentable
  categories~\cite[Proposition~1.59]{Adamek}.
\end{proof}

\begin{exa}\label{ex:Rule0}
  Let $ℙ=1$ so that $𝐒𝐞𝐭^ℙ ≅ 𝐒𝐞𝐭$, and $T$ denote the monad for pure
  $λ$-calculus syntax.  Taking $𝐂 = T\Mod_f(𝐒𝐞𝐭)$ and $X $ the product module
  $T² = T×T$, let us
  consider the left congruence rule for application
  \begin{mathpar}
    \inferrule{M ⟶ N}{M\ P ⟶ N\ P}
  \end{mathpar}
  as a signature of $𝐑𝐞𝐠⁰(𝐂/X)$.  For this, we take
  \begin{itemize}
  \item as metavariable module $V = T³$,
  \item a single premise given by $⟨π₁,π₂⟩∶ T³ → T²$, and
  \item as conclusion the morphism $t∶ T³ → T²$ mapping any $(M,N,P)$ to
    $(M\ P,N\ P)$.
  \end{itemize}
  Let us now see what it means for a module morphism $R → T²$ to be a
  model. In this case, the pullback along $⟨π₁,π₂⟩$ yields
  $R×T$, and so a model structure amounts to a morphism
  $R×T → R$ making the following square commute.
  \begin{center}
    \diag{%
      R×T \& R \\
      T²×T \& T² %
    }{%
      (m-1-1) edge[labela={}] (m-1-2) %
      edge[labell={}] (m-2-1) %
      (m-2-1) edge[labelb={t}] (m-2-2) %
      (m-1-2) edge[labelr={}] (m-2-2) %
    }
  \end{center}
  Unfolding the definition, such a map associates to any transition
  $r ∈ R(X)$ over $(M,N) ∈ T²(X)$ and term $P ∈ T(X)$ a transition
   over $t(M,N,P) = (M\ P,N\ P)$, as desired.
\end{exa}

\subsection{Binding register for slice categories}\label{ss:binding}
In this section, we observe that the register $𝐑𝐞𝐠⁰(𝐂/X)$ is not
expressive enough in the presence of variable binding, hence we refine
it.
\begin{exa}\label{ex:xi}
  Let $ℙ=1$ so that $𝐒𝐞𝐭^ℙ ≅ 𝐒𝐞𝐭$, and $T$ denote the monad for pure
  $λ$-calculus syntax.  Taking $𝐂 = T\Mod_f$ and $X = T²$ as in
  Example~\ref{ex:Rule0}, let us consider the $ξ$-rule
  \begin{mathpar}
    \inferrule{M → N}{λx.M → λx.N}~·
  \end{mathpar}
  The natural metavariable module for this is $T^{(1)}×T^{(1)}$,
  because $M$ and $N$ have an additional, fresh variable, and the
  natural premise would be the identity map thereon.  However, the
  register $𝐑𝐞𝐠⁰(𝐂/X)$ only allows premises to target powers of $X$.
\end{exa}
In order to rectify the situation, we refine $𝐑𝐞𝐠⁰(𝐂/X)$
to obtain the following more general register.
\begin{defi}\label{def:MSlice}
  For any locally finitely presentable category $𝐂$ and object
  $X ∈ 𝐂$, we define the endofunctorial register $𝐑𝐞𝐠¹(𝐂/X)$ for the
  slice category $𝐂/X$ as follows.  
\begin{description}
\item[Signatures]
  A signature consists of:
  \begin{itemize}
  \item a \alert{metavariable} object $V$;
  \item a list of \alert{premise} morphisms
    $(V\xrightarrow{sᵢ} FᵢX)_{i∈n}$,
      where each $Fᵢ$ is a finitary endofunctor on $𝐂$,
      denoted by $V \xrightarrow{\vec{s}} ∏ᵢ FᵢX$;
    \item a \alert{conclusion} morphism $V\xrightarrow{t} X$.
  \end{itemize}
\item[Semantics]
  A model of a signature $(∏ᵢ FᵢX \xleftarrow{\vec{s}} V \xrightarrow{t} X)$
  is an algebra for the functor mapping $Y \xrightarrow{p} X$ to
  $Y' → V \xrightarrow{t} X$,
  where $Y'$ denotes the pullback
  \begin{equation*}
    \Diag{%
      \pbk{m-2-1}{m-1-1}{m-1-2} %
    }{%
      Y' \& ∏ᵢFᵢ(Y) \\
       V \& ∏ᵢFᵢ(X)
    }{%
      (m-1-1) edge[labela={}] (m-1-2) %
      edge[labell={}] (m-2-1) %
      (m-2-1) edge[labelb={\vec{s}}] (m-2-2) %
      (m-1-2) edge[labelr={∏ᵢFᵢ(p)}] (m-2-2) %
    }%
  \end{equation*}
\end{description}
\end{defi}
\begin{proof}[{\rm \bf Validity proof:}\nopunct]
  Similar to Definition~\ref{def:rule0}, with the following composite endofunctor.
  \begin{center}
    \hfill $𝐂/X \xto{(∏ᵢFᵢ(-))/X} 𝐂/{\textstyle ∏}ᵢFᵢ(X) \xto{\vec{s}^*} 𝐂/V \xto{t_!} 𝐂/X$ \qedhere
\end{center}
\end{proof}
\begin{exa}
  Let us now treat the $ξ$-rule, rectifying Example~\ref{ex:xi}. 
  \begin{itemize}
  \item We
  first take as metavariable module $V ≔ T^{(1)}×T^{(1)}$, as planned.
\item We then take as unique premise the identity on $V$.  For this we
  should justify that $V$ does have the desired form $F(T²)$. This is
  the case with $F(M) = M^{(1)}$ since $(T²)^{(1)} = (T^{(1)})²$.
\item Finally, we take as conclusion
  $V ≔ T^{(1)}×T^{(1)} \xto{λ×λ} T×T$.
  \end{itemize}
\end{exa}
\begin{rem}
\label{rk:deriv-adjoint}
  We may generalise this register to permit a conclusion between terms
  with additional free variables, by having the conclusion morphism
  target $RX$ rather than $X$, for some finitary right adjoint functor
  $R$.  For example, exploiting the adjunction
  $-^{(1)} ⊢ - × T$ in the category of $T$-modules~\cite[Proposition 13]{DBLP:conf/csl/AhrensHLM18},
  the application of untyped $λ$-calculus can be
  viewed as an operation $\app_{alt}∶ T → T^{(1)}$.
  Anticipating on~§\ref{sss:formatrules}, the corresponding modified
  $β$-rule is
  \begin{mathpar}
    \inferrule{ }{\app_{alt}(abs(t)) ↝ t}~·
  \end{mathpar}
  The point here is that for any set $X$ and $t ∈ T^{(1)}(X)$, both
  terms $\app_{alt}(abs(t))$ and $t$ lie in $T^{(1)}(X)$.
\end{rem}

\subsection{Typed variant}\label{sss:bindingrules}
In this section, we observe that the medium register $𝐑𝐞𝐠¹(𝐂/X)$
is slightly inconvenient in a typed setting, and propose our last register
$𝐑𝐞𝐠(𝐂/X)$ for slice categories.
\begin{exa}\label{ex:typedrule}
  As in~§\ref{ss:stlambda}, let $ℙ$ denote the set of simple types
  over a given set of ground types and $T∶ 𝐒𝐞𝐭^ℙ → 𝐒𝐞𝐭^ℙ$ denote the
  monad for simply-typed $λ$-calculus values. Furthermore, let us
  recall the first state functor $S₁∶ 𝐒𝐞𝐭^ℙ → 𝐒𝐞𝐭^ℙ$: $S₁(X)$ is the
  set of typed application binary trees with leaves in $X$.  Let us
  consider the $β$-rule
  \begin{mathpar}
    \inferrule{e₁ ⇓ λ(e₃) \\ e₂ ⇓ w \\ e₃[w] ⇓ v}{e₁\ e₂ ⇓ v}~·
  \end{mathpar}
  Implicitly, this is in fact a family of rules indexed over all pairs
  $(A,B) ∈ ℙ²$ of types. For any such $(A,B)$, we have
  $e₁ ∈ S₁T(X)_{A → B}$, $e₂ ∈ S₂T(X)_A$, and $e₁\ e₂ ∈ S₁T(X)_B$.
  But these are all $𝐒𝐞𝐭$-valued modules, while the transition module
  $R → S₁T×S₂T$ is $𝐒𝐞𝐭^ℙ$-valued.

  In order to work with $𝐒𝐞𝐭$-valued modules, we now want to introduce
  a refinement of the register $𝐑𝐞𝐠¹(𝐂/X)$.  Let us first sketch on
  this example how it should look like.  First of all, because
  $e₁\ e₂$ and $v$ have type $B$, we would like to replace the
  $𝐒𝐞𝐭^ℙ$-valued module $S₁T×S₂T$ with the $𝐒𝐞𝐭$-valued $(S₁T×S₂T)_B$.
  We would then use as metavariable module the product
  \[V =  (S₁T)_{A→B} × (S₁T)_A × (S₁T)^{(A)}_{B} × T_{B} × T_A\rlap{,}\]
  whose elements are tuples $(e₁,e₂,e₃,v,w)$ as in the rule.  The
  conclusion of our rule should then consist of a morphism
  $V → (S₁T×S₂T)_B$, and similarly for the premises (see
  Example~\ref{ex:betaformatted} below).  The crucial ingredient here
  is the functor $(-)_B∶ T\Mod_f(𝐒𝐞𝐭^𝕊) → T\Mod_f(𝐒𝐞𝐭)$. Furthermore,
  in order to prove the existence of an initial model, it is important
  that this functor has a left adjoint $(-)·𝐲_B$, as
  in Proposition~\ref{prop:adjyev}.
\end{exa}

Abstracting over this situation, we are led to:
\begin{defi}\label{def:Slice}
  For any locally finitely presentable category $𝐂$ and object
  $X ∈ 𝐂$, we define the endofunctorial register $𝐑𝐞𝐠(𝐂/X)$ for the
  slice category $𝐂/X$ as follows.  
\begin{description}
\item[Signatures]
  A signature, called a \alert{rule}, consists of:
  \begin{itemize}
  \item a category $𝐃$ and a right adjoint $E∶ 𝐂 → 𝐃$;
  \item a \alert{metavariable} object $V ∈ 𝐃$;
  \item a list of \alert{premise} morphisms
    $(V \xrightarrow{sᵢ} FᵢX)_{i∈n}$,
      where each $Fᵢ∶ 𝐂 → 𝐃$ is a finitary functor,
      denoted by $V \xrightarrow{\vec{s}} ∏ᵢ FᵢX$;
    \item a \alert{conclusion} morphism $V\xrightarrow{t} E(X)$.
  \end{itemize}
\item[Semantics] Let us consider any signature $S$, consisting of
  $E∶ 𝐂 → 𝐃$, with left adjoint $J∶ 𝐃 → 𝐂$, and morphisms
  $(∏ᵢ FᵢX \xleftarrow{\vec{s}} V \xrightarrow{t} E(X))$.  Then, $S$
  induces a functor $Σ_S∶ 𝐂/X → 𝐃/E(X)$ mapping any
  $Y \xrightarrow{p} X$ to $Y' → V \xrightarrow{t} E(X)$,
  where $Y'$ denotes the following pullback.
  \begin{equation*}
    \Diag{%
      \pbk{m-2-1}{m-1-1}{m-1-2} %
    }{%
      Y' \& ∏ᵢFᵢ(Y) \\
       V \& ∏ᵢFᵢ(X)
    }{%
      (m-1-1) edge[labela={}] (m-1-2) %
      edge[labell={}] (m-2-1) %
      (m-2-1) edge[labelb={\vec{s}}] (m-2-2) %
      (m-1-2) edge[labelr={∏ᵢFᵢ(p)}] (m-2-2) %
    }%
  \end{equation*}
  Composing with
  the composite
  \begin{equation}
    𝐃/E(X) \xto{\widetilde{(-)}} J/X → 𝐂/X\rlap{,}\label{eq:composite}
  \end{equation}
  where the first functor denotes transposition, we obtain an
  endofunctor $\widetilde{Σ_S}$, and define the \semsig associated to
  $S$ to be the forgetful functor
    \[\widetilde{Σ_S}\alg → 𝐂/X.\]
  
\end{description}
\end{defi}
\begin{rem}
    Equivalently, a model is a morphism $p∶ Y → X$ equipped with a map
  $k$ making the following triangle commute,
  \begin{center}
    \diag{%
      Y' \& \& E(Y) \\
      \& E(X) %
    }{%

      (m-1-1) edge[labela={k}] (m-1-3) %
      edge[labelbl={Σ_S(Y,p)}] (m-2-2) %
      (m-1-3) edge[labelbr={E(p)}] (m-2-2) %
    }
  \end{center}
  and a morphism of models $(Y,p) → (Z,q)$ is a morphism $f$ in $𝐂/X$
  making the following square commute.
  \begin{center}
    \diag{%
      Y' \& E(Y) \\
      Z' \& E(Z) %
    }{%
      (m-1-1) edge[labela={k_Y}] (m-1-2) %
      edge[labell={f'}] (m-2-1) %
      (m-2-1) edge[labelb={k_Z}] (m-2-2) %
      (m-1-2) edge[labelr={E(f)}] (m-2-2) %
    }
  \end{center}
\end{rem}
\begin{proof}[{\rm \bf Validity proof:}\nopunct]
  The endofunctor $\widetilde{Σ_S}$ is obtained by composing $Σ_S$
  with~\eqref{eq:composite}, so it suffices to show that both of these
  functors are finitary.

  The composite~\eqref{eq:composite} is in fact cocontinuous, because
  colimits in slice categories are computed on domains, and, on
  domains, \eqref{eq:composite} acts like the left adjoint $J$.
  
  Let us now consider the functor $Σ_S$. It is a composite of
  three functors, so it suffices to show that each of these functors
  is finitary.
  The last two, pullback along $\vec{s}$ and postcomposition with $t$,
  have right adjoints, hence are even cocontinuous.
  Finally, the first functor $𝐂/X → 𝐃/∏ᵢ Fᵢ(X)$, mapping any
  $Y \xto{p} X$ to $∏ᵢFᵢ(Y) \xto{∏ᵢ Fᵢ(p)} ∏ᵢFᵢ(X)$ is finitary
  because colimits in slice categories are computed on domains, and,
  on domains, this functor acts like $Y ↦ ∏ᵢ Fᵢ(Y)$, which is a finite
  product of finitary functors -- and filtered colimits commute with
  finite products in locally finitely presentable categories.
\end{proof}

\begin{exa}\label{ex:betaformatted}
  Let us now treat the big-step $β$-rule, finishing Example~\ref{ex:typedrule}.
  For any types $A$ and $B$, we define the following rule:
  \begin{itemize}
  \item we take $𝐃 = T\Mod_f(𝐒𝐞𝐭)$ and
    $E∶ T\Mod_f(𝐒𝐞𝐭^𝕊) → T\Mod_f(𝐒𝐞𝐭)$ to be pointwise evaluation at
    $B$, which is indeed right adjoint to $(-)·𝐲_B$
    (Proposition~\ref{prop:adjyev});
  \item we further take $V ≔ (S₁T)_{A→B} × (S₁T)_A × (S₁T)^{(A)}_{B} × T_{B} × T_A$ as metavariable object;
  \item we have three premises $V → S₁T×S₂T$, which, at any
    $X ∈ 𝐒𝐞𝐭^ℙ$, respectively map any $(e₁,e₂,e₂,w,v) ∈ V(X)$ to:
    \begin{itemize}
    \item $(e₁,λ_{A,B}(e₃))$,
    \item $(e₂,w)$, and
    \item $(e₃[w],v)$;
    \end{itemize}
  \item the conclusion maps any such $(e₁,e₂,e₂,w,v) ∈ V(X)$ to $(\app_{A,B}(e₁, e₂),v)$.
  \end{itemize}
  Using the notation of~§\ref{sss:formatrules} below, this will look
  much like the standard, syntactic rule.
\end{exa}

\begin{defi}\label{def:Rule}
  Let $𝐑𝐞𝐠𝐓𝐫𝐚𝐧𝐬𝐒𝐭𝐫𝐮𝐜𝐭⁰_{ℙ,𝕊}(T,S₁,S₂) = 𝐑𝐞𝐠(T\Mod_f(𝐒𝐞𝐭^𝕊)/S₁T×S₂T)$.
\end{defi}

Finally, we would like signatures to consist of families of rules.
For this, we use the following generic construction of registers.
\begin{defi}\label{def:copairingreg}
  For any endofunctorial register $𝐑$ for a category with coproducts, we denote by $𝐑^*$ the
  endofunctorial register whose signatures are families of signatures
  in $𝐒𝐢𝐠_𝐑$, and whose semantics maps any family to the coproduct of
  associated endofunctors.
\end{defi}

We may at last define our register for the category
$𝐓𝐫𝐚𝐧𝐬𝐒𝐭𝐫𝐮𝐜𝐭_{ℙ,𝕊}(T,S₁,S₂)$ of \transition structures, which we
recall is by definition the slice category $T\Mod_f(𝐒𝐞𝐭^𝕊)/S₁T×S₂T$.
For this, we take as signatures all families of signatures in
$𝐑𝐞𝐠𝐓𝐫𝐚𝐧𝐬𝐒𝐭𝐫𝐮𝐜𝐭⁰_{ℙ,𝕊}(T,S₁,S₂)$:
\begin{defi}\label{def:regtransstruct}
  Let $𝐑𝐞𝐠𝐓𝐫𝐚𝐧𝐬𝐒𝐭𝐫𝐮𝐜𝐭_{ℙ,𝕊}(T,S₁,S₂) = 𝐑𝐞𝐠𝐓𝐫𝐚𝐧𝐬𝐒𝐭𝐫𝐮𝐜𝐭⁰_{ℙ,𝕊}(T,S₁,S₂)^*$.
\end{defi}

\subsection{A format for displaying signatures in rule-based registers}\label{sss:formatrules}
In all of our examples of signatures in
$𝐑𝐞𝐠𝐓𝐫𝐚𝐧𝐬𝐒𝐭𝐫𝐮𝐜𝐭⁰_{ℙ,𝕊}(T,S₁,S₂)$, the metavariable object $V$ is a
functor to $𝐒𝐞𝐭$, so the premises and conclusion are set-maps (which
are in fact module morphisms).  In this case, we adopt the following
notational conventions.
\begin{itemize}
\item For each premise or conclusion
 $
  \begin{array}[t]{rcll}
    V & → & W \\
    x & ↦ & e
  \end{array}$
  of a rule,
  we write $x ∶ V ⊢ e ∶ W$.
\item Furthermore, we organise the premises and conclusion as usual: 
  \begin{center}
  $\inferrule{x ∶ V ⊢ e₁ ∶ W₁ \\ … \\ x ∶ V ⊢ eₙ
    ∶ Wₙ}{x ∶ V ⊢ e ∶ W}~\rlap{,}$
\end{center}
or just
\quad $\inferrule{e₁ \\ … \\ eₙ}{e}$ \quad
when the rest may be inferred from context.
\end{itemize}

Moreover, in our examples $M = S₁T×S₂T$, so each element $e$ is in
fact a pair $(L,R)$, which we generally denote with an arrow, e.g., by
$L ↝ R$, $L → R$,….
\begin{exa}
  The big-step $β$-rule from Example~\ref{ex:betaformatted} reads as
  follows.
  \begin{mathpar}
    \inferrule{ e₁ \leadsto λ_{A, B}(e₃) \\ e₂ \leadsto w \\ e₃[w] \leadsto v }{ \app_{A,B}(e₁, e₂) \leadsto v }
  \end{mathpar}
\end{exa}

\begin{rem}
  The module $V$ is often a product and thus $x$ is a tuple.
\end{rem}

\begin{remC}[{\cite{AHLM19}}]
  In practice, there are several choices for building the transition
  rule out of such a schematic presentation, depending on the order of
  metavariables. This order is irrelevant: all interpretations yield
  isomorphic semantics, in the obvious sense.
\end{remC}

\begin{rem}
  This format could be generalised to any metavariable object, by
  using the internal language of categories.
\end{rem}

\subsection{Proof-irrelevant variant}\label{ss:irr-cats}
In this section, we introduce a register for the proof-irrelevant
variant of \transition monads.  The idea is very simple: we keep the
same signatures as in the proof-relevant setting, and interpret each
signature in a proof-irrelevant way. This is done by constructing a
functor from proof-relevant \transition monads to proof-irrelevant
ones.
  \begin{prop}
    \label{prop:irrelevant-trans-monads}
    Let $𝐈𝐓𝐫𝐚𝐧𝐬𝐒𝐭𝐫𝐮𝐜𝐭_{ℙ,𝕊}(T,S₁,S₂)$ denote the full subcategory of
    \transition structures $⟨\src,\tgt⟩∶ R → S₁T × S₂T$ such that
    $⟨\src,\tgt⟩$ is a pointwise inclusion.  Then, the embedding
    $U_{T,S₁,S₂} ∶ 𝐈𝐓𝐫𝐚𝐧𝐬𝐒𝐭𝐫𝐮𝐜𝐭_{ℙ,𝕊}(T,S₁,S₂) ↪ 𝐓𝐫𝐚𝐧𝐬𝐒𝐭𝐫𝐮𝐜𝐭_{ℙ,𝕊}(T,S₁,S₂)$ is
    reflective.  Consequently, letting
    $𝐈𝐓𝐫𝐚𝐧𝐬𝐌𝐧𝐝_{ℙ,𝕊} ≔ ∑_{T,S₁,S₂} 𝐈𝐓𝐫𝐚𝐧𝐬𝐒𝐭𝐫𝐮𝐜𝐭_{ℙ,𝕊}(T,S₁,S₂)$, the induced
    embedding $𝐈𝐓𝐫𝐚𝐧𝐬𝐌𝐧𝐝_{ℙ,𝕊} ↪ 𝐓𝐫𝐚𝐧𝐬𝐌𝐧𝐝_{ℙ,𝕊}$ is also reflective.
  \end{prop}
  \begin{proof}
    Refining Remark~\ref{rk:natural-retract},
    by~\cite[Definition~2.71]{AhrensPhD}, the category of finitary
    $T$-modules is a category of presheaves on a small category.
    Thus, by~\cite[Example~4.3.10.g]{BorceuxI}, all epimorphisms are
    strong.  Furthermore, by~\cite[Proposition~4.4.3]{BorceuxI}, it
    admits (strong epi)-mono factorisations.
    
    Using this, we define the left adjoint
    \[L∶𝐓𝐫𝐚𝐧𝐬𝐒𝐭𝐫𝐮𝐜𝐭_{ℙ,𝕊}(T,S₁,S₂) → 𝐈𝐓𝐫𝐚𝐧𝐬𝐒𝐭𝐫𝐮𝐜𝐭_{ℙ,𝕊}(T,S₁,S₂) \]to map
    \atransition structure $∂∶ R → S₁T × S₂T$ to the monomorphism
    $\overline{R} ↪ S₁T × S₂T$ obtained from the (strong epi)-mono
    factorisation of $∂$.  Then, the natural bijection
    \[
    𝐓𝐫𝐚𝐧𝐬𝐒𝐭𝐫𝐮𝐜𝐭_{ℙ,𝕊}(T,S₁,S₂)( R₁, U R₂)≃
    𝐈𝐓𝐫𝐚𝐧𝐬𝐒𝐭𝐫𝐮𝐜𝐭_{ℙ,𝕊}(T,S₁,S₂)(L R₁, R₂)\] follows from the lifting
    property of strong epimorphisms.
  \end{proof}
  
  Let us now introduce the relevant register.  For this, we first
  observe that postcomposition with a functor $F∶ 𝐂 → 𝐃$ turns a
  register for $𝐂$ into one for $𝐃$.
  \begin{defi}[Post-composition register]
    For any functor $F∶ 𝐂 → 𝐃$ and register $𝐑$ on $𝐂$, let $F_!(𝐑)$
    denote the register for $𝐃$ with $𝐒𝐢𝐠_{F_!(𝐑)} = 𝐒𝐢𝐠_{𝐑}$ and
    $⟦s⟧_{F_!(𝐑)} = F ∘ ⟦s⟧_𝐑$.
  \end{defi}

  \begin{defi}\label{def:irrelreg}
    The register $𝐑𝐞𝐠𝐈𝐓𝐫𝐚𝐧𝐬𝐌𝐧𝐝_{ℙ,𝕊}$ is defined as
    $F_{!}(𝐑𝐞𝐠𝐓𝐫𝐚𝐧𝐬𝐌𝐧𝐝_{ℙ,𝕊})$, where $F∶ 𝐓𝐫𝐚𝐧𝐬𝐌𝐧𝐝_{ℙ,𝕊} → 𝐈𝐓𝐫𝐚𝐧𝐬𝐌𝐧𝐝_{ℙ,𝕊}$
    denotes the reflection.
  \end{defi}

\section{Registers for monads}\label{ss:regmonads}
In this section and the next, we design the missing registers,
respectively for monads and state functors.  The registers are mostly
adapted from existing constructions and results in the
literature~\cite{AHLM:2sig,fiore:presheaf,DBLP:conf/lics/Fiore08,FioreHurEquational}.
The novelty here lies in our new explicit description of
initial algebras.

The basic idea for specifying operations in our register for monads is
that the arity of an operation consists of
\begin{itemize}
\item an ``input'' ($𝐒𝐞𝐭$-valued) \alert{parametric} module, in the
  sense of~§\ref{sss:param}, together with
\item an ``output'' placetaker type $p ∈ ℙ$.
\end{itemize}
The role of parametric modules lies in specifying how capture-avoiding
substitution should interact with operations. A similar role is played
in~\cite{fiore:presheaf,DBLP:conf/lics/Fiore08} by ``pointed strong''
endofunctors; we explain the connection in~§\ref{ss:PSEF}.

In~§\ref{ss:regmon}, we construct a first register $𝐑𝐞𝐠𝐌𝐧𝐝⁰_f(𝐒𝐞𝐭^ℙ)$,
which only allows to specify operations. We then deal with equations
in~§\ref{ss:ModMnd}.  Finally, we characterise initial algebras
in~§\ref{ss:ia:mnd}.  All proofs are deferred
to~§\ref{s:friendliness}.

  \subsection{The register \texorpdfstring{$𝐑𝐞𝐠𝐌𝐧𝐝⁰_f(𝐒𝐞𝐭^ℙ)$}{RegMnd0f(Set)} for specifying operations}
  \label{ss:regmon}
  This section is devoted to defining the monadic register
  $𝐑𝐞𝐠𝐌𝐧𝐝⁰_f(𝐒𝐞𝐭^ℙ)$.
  
  Signatures will rely on parametric modules, but these need to be
  restricted in order to ensure existence of an initial algebra (and
  even monadicity).

  \begin{prop}
    For any set $ℙ$, $p₁,…,pₙ ∈ ℙ$, and finitary functor
    $F∶ 𝐒𝐞𝐭^ℙ → 𝐒𝐞𝐭$, the assignment
    $T,X ↦ F(T(X + ∑_{i ∈ n} 𝐲_{pᵢ}))$ defines a parametric module
    denoted by $(F∘ Θ)^{(p₁,…,pₙ)}$.
  \end{prop}
  \begin{proof}
    By Example~\ref{ex:basic-param-mod}.
  \end{proof}

  \begin{defi}
    A parametric module is \alert{elementary} if it is isomorphic to
    some finite product of parametric modules of the shape
    $(F ∘ Θ)^{(p₁,…,pₙ)}$.
  \end{defi}
  
  \begin{exa}
    Typically, taking $F(X) = X(p)$, any finite product of parametric
    modules of the shape $Θ_p^{(p₁,…,pₙ)}$, for some $p,p₁,…,pₙ ∈ ℙ$,
    is elementary.
  \end{exa}

\begin{exa}
  Recall that the idea of our register is that an operation will be
  specified by two parametric modules, one for the source and another
  (very simple) for the target.  Let us give the parametric modules
  for a few operations from our examples.

  \begin{center}
  \begin{tabular}[t]{cccc} \toprule Language &
    Operation & Source & Target \\ \midrule
                     Pure $\overline{λ}μ$ & Push $e·π$ & $Θ_𝐩×Θ_𝐬$ & $Θ_𝐬$ \\
                     Pure $\overline{λ}μ$ & Abstraction $λx.e$  & $Θ^{(𝐩)}_𝐩$ & $Θ_𝐩$ \\
                     $π$-calculus & Input $a(b).P$ & $Θ_𝐜 × Θ^{(𝐜)}_𝐩$
                       & $Θ_𝐩$ \\ \bottomrule
    \end{tabular}
  \end{center}
  \bigskip

  In the above table, $𝐩$, $𝐬$, and $𝐜$ are placetaker types, and the
  subscripts on $Θ$ refer to the notation $Mₚ$ introduced in
  Example~\ref{ex:basic-param-mod}.  Thus, e.g.,
  $(Θ_𝐩×Θ_𝐬)(T)(X) = T(X)_𝐩 × T(X)_𝐬$.
  \end{exa}

  \begin{defi}
    Given any set $ℙ$, a \alert{modular signature}
    is a family of pairs $(d,p)$ where
    \begin{itemize}
    \item $d$ is an elementary parametric module, and
    \item $p ∈ ℙ$.
    \end{itemize}
    Given any modular signature $S = (dᵢ,pᵢ)_{i ∈ I}$, an
    \alert{$S$-algebra} is a finitary monad $T$ equipped with $T$-module
    morphisms $dᵢ(T) → T_{pᵢ}$ for all $i ∈ I$.  An $S$-algebra
    morphism is a monad morphism commuting with these morphisms.  We
    denote by $S\alg → 𝐌𝐧𝐝_f(𝐒𝐞𝐭^ℙ)$ the forgetful functor.
  \end{defi}

  \begin{defi}\label{def:elem}
    We define the monadic register $𝐑𝐞𝐠𝐌𝐧𝐝⁰_f(𝐒𝐞𝐭^ℙ)$ for
    $𝐌𝐧𝐝_f(𝐒𝐞𝐭^ℙ)$
    as follows, for any set $ℙ$.
    \begin{description}
    \item[Signatures] A signature is a modular signature.
    \item[Semantics] The \semsig associated to any signature $S$ is
      the forgetful functor $S\alg → 𝐌𝐧𝐝_f(𝐒𝐞𝐭^ℙ)$.
    \end{description}
\end{defi}
\begin{proof}[{\rm \bf Validity proof:}\nopunct]
  By Corollary~\ref{cor:SModMnd} below.
\end{proof}

\subsection{The register \texorpdfstring{$𝐑𝐞𝐠𝐌𝐧𝐝_f(𝐒𝐞𝐭^ℙ)$}{RegMndf(Set)}}\label{ss:ModMnd}
We now define our register $𝐑𝐞𝐠𝐌𝐧𝐝_f(𝐒𝐞𝐭^ℙ)$, where a signature will
consist of a signature of $𝐑𝐞𝐠𝐌𝐧𝐝⁰_f(𝐒𝐞𝐭^ℙ)$, plus a family of
``equations''.  An equation is essentially a pair of ``derived
operations'' with a common ``arity''. The arity consists of an input
arity, which intuitively models the metavariables of the equation, and
an output arity, which models the output type.  The input arity will
be an elementary parametric module $d$, and the output arity will be a
placetaker type $p ∈ ℙ$. For a family of equations, the arity thus is
a family of such pairs $(d,p)$ i.e., a modular signature.  An equation
will then consist of two derived operations with the same arity, in
the following sense.
\begin{defi}\label{def:derivedops}
  Given any modular signatures $S$ and $S'$, an \alert{$S$-derived
    operation of arity $S'$} is a functor $L∶ S\alg → S'\alg$ over
  $𝐌𝐧𝐝_f(𝐒𝐞𝐭^ℙ)$, i.e., making the following triangle commute.
    \begin{center}
      \diag{%
        S\alg \& \& S'\alg \\
        \& 𝐌𝐧𝐝_f(𝐒𝐞𝐭^ℙ)
      }{%
        (m-1-1) edge[labela={L}] (m-1-3) %
        edge[labelbl={}] (m-2-2) %
        (m-1-3) edge[labelbr={}] (m-2-2) %
      }
    \end{center}
    We call operations of $S$ \alert{basic}, by contrast with the
    \alert{derived} operations of $S'$.

  More concretely, we may introduce derived operations in two stages,
  as follows:
    \begin{itemize}
    \item an \alert{$S$-module morphism} $M → N$ between parametric
      modules $M$ and $N$ is a natural family of morphisms
      $(α_T ∶M(T) ⟶ N(T))_{T∈ S\alg}$, such that $α_T$ is a $T$-module
      morphism, for each $S$-algebra $T$;
    \item letting $S' = (Vⱼ,qⱼ)_{j ∈ J}$, an $S$-derived operation $L$
      of arity $S'$ is a family of parametric module morphisms $Vⱼ → Θ_{qⱼ}$, for
      all $j ∈ J$.
    \end{itemize}
  \end{defi}
  \begin{rem}
    Concretely, an $S$-derived operation $L$ of arity $S'$ associates
    to each $S$-algebra $T$ a family of $T$-module morphisms
    $Vⱼ(T) → T_{qⱼ}$, naturally in $T$.

    Equivalently, by Proposition~\ref{prop:adjyev}, an $S$-derived
    operation of arity $S'$ associates to each $T$ a single,
    $𝐒𝐞𝐭^ℙ$-valued $T$-module morphism $𝐇_{S'}(T) → T$, where
    $𝐇_{S'}(T) ≔ ∑_{j ∈ J} Vⱼ(T)·𝐲_{qⱼ}$.
  \end{rem}
\begin{exa}\label{ex:eqn}
  Consider associativity of parallel composition in the $π$-calculus,
  $P|(Q|R) ≡ (P|Q)|R$: the metavariables are $P$, $Q$, and $R$.  The
  corresponding input arity is $Θ_𝐩³$, and the output arity is $𝐩$.
  Recalling~§\ref{ss:expi}, and anticipating on~§\ref{ss:signature:pi}
  below, 
  the basic signature $S$ contains
  in particular an operation $\paral∶ Θ_𝐩² → Θ_𝐩$ for parallel
  composition, and the
   derived operations for associativity respectively
  map any algebra $(T, \paral_T∶ T_𝐩² → T_𝐩,…)$ to
  \begin{center}
    $T³_𝐩 \xto{\paral × T_𝐩} T²_𝐩 \xto{\paral} T_𝐩$ \hfil and \hfil
    $T³_𝐩 \xto{T_𝐩 × \paral} T²_𝐩 \xto{\paral} T_𝐩$.
  \end{center}
\end{exa}

Returning to the general case, we now define our notion of signature for monads.
\begin{defi}
  An \alert{equational modular signature} consists of
  \begin{itemize}
  \item a modular signature $S$ called the modular signature for
    \alert{operations}, or the \alert{operations} modular signature,
  \item a modular signature $S'$ called the modular signature for
    \alert{equations}, or the \alert{equations} modular signature,
    together with
  \item a pair of $S$-derived operations of arity $S'$.
\end{itemize}
  For any equational modular signature $E = (S,S',L,R)$, an
  \alert{$E$-algebra} is an $S$-algebra $T$ such that the $S'$-algebra
  structures $L(T)$ and $R(T)$ coincide, i.e., $L(T) = R(T)$.  A
  morphism of $E$-algebras is a morphism of $S$-algebras.  We let
  $E\alg$ denote the category of $E$-algebras and morphisms between
  them.
\end{defi}

Let us at last define our register.
\begin{defi}\label{def:modmnd}
  We define the monadic register $𝐑𝐞𝐠𝐌𝐧𝐝_f(𝐒𝐞𝐭^ℙ)$ for $𝐌𝐧𝐝_f(𝐒𝐞𝐭^ℙ)$
  as follows, for any set $ℙ$.
  \begin{description}
  \item[Signatures] A signature is an equational modular signature.
  \item[Semantics] The \semsig associated to any equational modular
    signature $E$ is the forgetful functor $E\alg → 𝐌𝐧𝐝_f(𝐒𝐞𝐭^ℙ)$.
\end{description}
\end{defi}
\begin{proof}[{\rm \bf Validity proof:}\nopunct]
  This is Corollary~\ref{cor:modmnd}\ref{item:ems:monadic} below.
\end{proof}

Let us conclude this subsection by introducing some convenient
notation for specifying equations. 

\begin{nota}[Format for equations]\label{not:monadformat}
  We write any equational modular signature whose equations modular
  signature is a singleton $S' = (V,p)$, say with derived operations
  given by
   \[\begin{array}[t]{rcl}
       V & → & Θₚ² \\
       x & ↦ & (L,R)\rlap{,}
   \end{array}\]
     as \[x∶ V ⊢ L ≡ R : Θₚ \](leaving the operations modular signature implicit),
     or even just $L ≡ R$ when the rest may be inferred.

     Furthermore, given any common (implicit) operations modular
     signature $S$, any family $(x∶ Vᵢ ⊢ Lᵢ ≡ Rᵢ : Θ_{pᵢ})_{i ∈ I}$ 
     will accordingly denote the equational modular signature
     \begin{itemize}
     \item whose equations modular signature is $S' = (Vᵢ,pᵢ)_{i ∈ I}$, and 
     \item whose $S$-derived operations of arity $S'$ are given at any
       $S$-algebra $T$ by the morphisms $Lᵢ(T)∶ Vᵢ(T) → T_{pᵢ}$ and
       $Rᵢ(T)∶ Vᵢ(T) → T_{pᵢ}$, for each $i ∈ I$.
     \end{itemize}
\end{nota}
\begin{exa}\label{ex:paranot}
  We write associativity from Example~\ref{ex:eqn} as just
  \[\paral (P,\paral(Q,R)) ≡ \paral (\paral (P,Q),R). \]In this case, the argument
  $x$ is the triple $(P,Q,R)$.
\end{exa}

\subsection{Explicit description of initial algebras}\label{ss:ia:mnd}
In this final subsection, we provide an explicit description of the
initial $E$-algebra, for any equational modular signature $E$.  We
first deal with the case without equations, recalling the standard
identification of the initial $S$-algebra as a free algebra for a
suitable endofunctor.  In the presence of equations, we then
characterise the initial $E$-algebra as a coequaliser of free
algebras.

In order to compute the initial algebra for a modular signature $S$,
we first observe that the induced homogeneous parametric module
$𝐇_S$ in fact comes from an endofunctor.

  \begin{defi}\label{def:SigmaS}
    For any modular signature $S$, we define its \alert{associated
      endofunctor} $Σ_S$ on $[𝐒𝐞𝐭^ℙ,𝐒𝐞𝐭^ℙ]_f$ as follows.
    \begin{itemize}
    \item For any elementary $d = (F ∘ Θ)^{(p₁,…,pₙ)}$ and $r ∈ ℙ$,
      let \[Σ_{(d,r)}(P) = (F ∘ P)^{(p₁,…,pₙ)}·𝐲ᵣ\rlap{.}\]
    \item For a family $S = (dᵢ,rᵢ)$, let
      $Σ_S = ∑_{i ∈ I} Σ_{(dᵢ,rᵢ)}.$
  \end{itemize}
  Explicitly, we have for any $(d,r)$:
    \begin{center}
    $Σ_{(d,r)}(P)(X) ≔ (F(P(X + ∑ᵢ 𝐲_{pᵢ})))·𝐲ᵣ\rlap{.}$
  \end{center}
  In perhaps more elementary terms, we have
    \[
  \begin{array}{rcll}
    Σ_{(d,r)}(P)(X)(r) & = & F(P(X + ∑ᵢ 𝐲_{pᵢ})) \\
    Σ_{(d,r)}(P)(X)(r') & = & ∅ & \mbox{for $r' ≠ r$.} 
  \end{array}\]
  \end{defi}
By construction, we have:
\begin{prop}
  For any modular signature $S$,  the following square commute,
  \begin{center}
    \diag{%
      𝐌𝐧𝐝_f(𝐒𝐞𝐭^ℙ) \& 𝐌𝐨𝐝(𝐒𝐞𝐭^ℙ,𝐒𝐞𝐭^ℙ) \\
      {[𝐒𝐞𝐭^ℙ,𝐒𝐞𝐭^ℙ]_f} \&
      {[𝐒𝐞𝐭^ℙ,𝐒𝐞𝐭^ℙ]_f} %
    }{%
      (m-1-1) edge[labela={𝐇_S}] (m-1-2) %
      edge[labell={}] (m-2-1) %
      (m-2-1) edge[labelb={Σ_S}] (m-2-2) %
      (m-1-2) edge[labelr={}] (m-2-2) %
    }
  \end{center}
  where the right-hand functor maps any pair $(T,M)$ to $M$.
\end{prop}

Let us now turn to the explicit description of the initial
$S$-algebra. The mathematical contents essentially
date back to~\cite{fiore:presheaf}.
  \begin{prop}\label{prop:S}
    For any modular signature $S$, the forgetful functor
    $S\alg → 𝐌𝐧𝐝_f(𝐒𝐞𝐭^ℙ)$ is monadic, and furthermore the free
    $Σ_S$-algebra $Σ_S^*(\id)$ on the identity has a canonical
    $S$-algebra structure, which is initial.
  \end{prop}
  \begin{proof}
    This is Corollary~\ref{cor:SModMnd} below.
  \end{proof}

  We now seek an explicit description of the initial
  $E$-algebra, for any equational modular signature $E$.

\begin{defi}\label{def:SE}
  Let $E = (S,S',L,R)$ denote any equational modular signature, with
  $S = (Vᵢ,pᵢ)_{i ∈ I}$ and $S' = (Wⱼ,qⱼ)_{j ∈ J}$.
\begin{itemize} 
\item Let $S+S'$ denote the ``disjoint union'', i.e., the modular
  signature $(Uₖ,rₖ)_{k ∈ K}$, where
  \begin{itemize}
  \item $K = I+J$,
  \item $(Uₖ,rₖ)$ is
    \begin{itemize}
    \item $(Vᵢ,pᵢ)$ if $k = in₁(i)$ and
    \item  $(Wⱼ,qⱼ)$ if $k = in₂(j)$.
    \end{itemize}
  \end{itemize}
\item The $S'$-algebra structures given by $L(\spec{S})$ and
  $R(\spec{S})$ on $\spec{S}$, together with its canonical $S$-algebra
  structure, yield two $(S+S')$-algebra structures.
  By initiality of $\spec{S+S'}$, we thus obtain
  two $(S+S')$-algebra morphisms
  \begin{equation} \label{eq:LR''mon}
    \check{L},\check{R}∶  \spec{S+S'} → \spec{S}\rlap{,}
  \end{equation}
  or equivalently, by Proposition~\ref{prop:S},
  \begin{equation} \label{eq:LR'''mon}
    \check{L},\check{R}∶  (Σ_S+Σ_{S'})^*(\id) → Σ_S^*(\id).
  \end{equation}
  \end{itemize}
\end{defi}

\begin{thm}\label{thm:modmnd}
  For any equational modular signature $E = (S,S',L,R)$, the
  coequaliser of the pair
  $\check{L},\check{R}∶ (Σ_S+Σ_{S'})^*(\id) → Σ_S^*(\id)$ in
  $[𝐒𝐞𝐭^ℙ,𝐒𝐞𝐭^ℙ]_f$ admits a unique $S$-algebra structure such that
  the coequalising morphism is an $S$-algebra morphism.  Furthermore,
  this structure is in fact an $E$-algebra structure. Finally, it
  makes the coequaliser into an initial $E$-algebra.
\end{thm}
\begin{proof}
  By Corollary~\ref{cor:modmnd} below.
\end{proof}

  \begin{rem}\hfill
    \begin{itemize}
    \item The coequaliser may be computed as a pointwise quotient of
      $Σ_S^*(\id)$, as follows. As explained by
      Hamana~\cite{DBLP:conf/aplas/Hamana04}, $Σ_S^*(\id)$
      may be defined as a term language, each $Σ_S^*(\id)(X)(p)$ being
      the set of terms of type $p$, with free variables of each type
      $q ∈ ℙ$ in $X(q)$ --- as in~§\ref{ss:stlambda}, the pair $(X,p)$
      may be thought of as a sequent $X ⊢ p$. Now, coequalisers (as
      all limits and colimits) are pointwise in functor
      categories~\cite[§V.4]{MacLane:cwm}, so for any $X ∈ 𝐒𝐞𝐭^ℙ$ and
      $p ∈ ℙ$, $\spec{E}(X)(p)$ is the coequaliser of
        \[(Σ_S+Σ_{S'})^*(\id)(X)(p) → Σ_S^*(\id)(X)(p)\]
      in $𝐒𝐞𝐭$. And this is well known to be the quotient of
      $Σ_S^*(\id)(X)(p)$ by the smallest equivalence relation identifying all
      $\check{L}_{X,p}(e)$ with $\check{R}_{X,p}(e)$, for all
      $e ∈ (Σ_S+Σ_{S'})^*(\id)(X)(p)$. Intuitively, following Hamana
      again, $(Σ_S+Σ_{S'})^*(\id)$ is an extension of the term
      language $Σ_S^*(\id)$ with operations, say $oⱼ$, with arities
      $(Vⱼ,qⱼ)$, for all $j$ (where $S' = (Vⱼ,qⱼ)_{j ∈ J}$ as before),
      and $\check{L}$ and $\check{R}$ inductively translate this
      language to $Σ_S^*(\id)$ by interpreting $oⱼ$ using
      $Lⱼ(Σ_S^*(\id))$ and $Rⱼ(Σ_S^*(\id))$, respectively.
    \item We could consider computing the coequaliser of a simpler
      parallel pair
      \begin{equation}
        \spec{S'} → \spec{S}\rlap{,}\label{eq:LRsimplistic}
      \end{equation}
      constructed similarly. Let us show on a simple example that this
      does not compute the desired functor.  The intuition is that
      this coequaliser identifies terms modulo a relation which is not
      a congruence.

      Let $ℙ = 1$, and $S$ consist of a single, binary operation,
      i.e., $S = \{ (Θ²,⋆) \}$.  Thus, an $S$-algebra is merely a
      monad on sets, equipped with a binary $T$-module morphism
      $b∶ T² → T$.  Furthermore, let $S'$ consist of a single,
      ternary operation $t$. Finally, for any $S$-algebra $(T,b)$, let
      $L(T)$ and $R(T)$ denote the $T$-module morphisms
      \begin{mathpar}
        T³ \xto{b×T} T² \xto{b} T \and
        T³ \xto{T×b} T² \xto{b} T.
      \end{mathpar}
      By Proposition~\ref{prop:S}, for any set $X$, $\spec{S}(X)$
      consists of binary trees with leaves in $X$.  Similarly,
      $\spec{S'}(X)$ consists of ternary trees with leaves in
      $X$.  Now, the parallel pair~\eqref{eq:LRsimplistic} maps
      ternary trees to binary trees, replacing any ternary node
      $(t₁,t₂,t₃)$, respectively with $((t₁,t₂),t₃)$ and
      $(t₁,(t₂,t₃))$. Thus, e.g., assuming $x ∈ X$, the binary trees
      $(((x,x),x),x)$ and $((x,(x,x)),x)$, having an even number of
      leaves, are not in the image of the parallel pair, hence are not
      identified in the coequaliser.
      
    \item A perhaps higher-level understanding of this, which will be
      developped in~§\ref{s:friendliness}, is as follows.  The derived
      operations induce monad morphisms $(S')^⋆ → S^⋆$, of which the
      desired monad $E^⋆$ is the coequaliser in
      $𝐌𝐧𝐝_f(𝐒𝐞𝐭^ℙ)$. However, coequalisers of monads are generally
      not pointwise~\cite{DBLP:conf/lics/AdamekMBL12}.  Fortunately,
      \alert{reflexive} coequalisers are, so the desired monad may be
      computed as the pointwise coequaliser of the obvious parallel
      pair $S^⋆ + (S')^⋆ → S^⋆$.  Finally, roughly because $(-)^⋆$ is
      a left adjoint in this case, we have $S^⋆ + (S')^⋆ ≅ (S+S')^⋆$,
      which directly leads to our formula.
    \end{itemize}
  \end{rem}

\section{Registers for (state) functors}\label{s:regfun}
In this section, we define a register $𝐑𝐞𝐠[𝐒𝐞𝐭^ℙ,𝐒𝐞𝐭^𝕊]_f$, which is a
variant of $𝐑𝐞𝐠𝐌𝐧𝐝_f(𝐒𝐞𝐭^ℙ)$ for the case of state functors.  Operations,
equations, and models will be defined exactly as for monads, and a
signature in $𝐑𝐞𝐠[𝐒𝐞𝐭^ℙ,𝐒𝐞𝐭^𝕊]_f$ will again consist of families of
operations and equations, the only difference being that instead of
parametric modules, we will use \alert{facets}.  We
start by presenting the auxiliary notion of facets,
which is a simple variant of the parametric modules
of~§\ref{ss:regmonads}. Next, we introduce a register $𝐑𝐞𝐠⁰[𝐒𝐞𝐭^ℙ,𝐒𝐞𝐭^𝕊]_f$
for operations, and then a refinement $𝐑𝐞𝐠[𝐒𝐞𝐭^ℙ,𝐒𝐞𝐭^𝕊]_f$ with equations.
Again, validity proofs are deferred to~§\ref{s:friendliness}.

\subsection{Facets}
We start in this section by introducing facets.  Let us first explain
how they naturally come up in the case of call-by-value, simply-typed
$λ$-calculus.  We have seen in~§\ref{ss:stlambda} that the source
state functor $S₁$ for call-by-value, simply-typed $λ$-calculus
consists of application binary trees.  The goal now is to design a
register for specifying such a functor $S₁$.  Intuitively, it has two
(type-indexed families of) operations:
\begin{itemize}
\item a first operation for  injecting values into application binary trees,
  of type $X_A → S₁(X)_A$ for all $X ∈ 𝐒𝐞𝐭^ℙ$ and $A ∈ ℙ$, and
\item a second operation for application, of type
  $S₁(X)_{A → B} × S₁(X)_A → S₁(X)_B$, for all $X ∈ 𝐒𝐞𝐭^ℙ$ and $A,B ∈ ℙ$.
\end{itemize}
The type for application is really similar to what we had
in~§\ref{ss:regmonads}: for any type $A ∈ ℙ$, denoting
by $Θ_A∶ [𝐒𝐞𝐭^ℙ,𝐒𝐞𝐭^ℙ] → [𝐒𝐞𝐭^ℙ,𝐒𝐞𝐭]$ the functor defined by
\[Θ_A(S)(X) = S(X)_A\rlap{,}\]
application of a function of type $A → B$ has arity
\[Θ_{A→B}×Θ_A → Θ_B.\]
The type for value injection does not make any sense in the context of
modules over monads, because the functor $M(T)(X) = X_A$ does not form
a module. But here in the context of state functors, we may well
define $𝐈_A∶ [𝐒𝐞𝐭^ℙ,𝐒𝐞𝐭^ℙ] → [𝐒𝐞𝐭^ℙ,𝐒𝐞𝐭]$ by
\[𝐈_A(S)(X) = X_A\rlap{,}\]
for any $A ∈ ℙ$.

Such functors $Θ_A$ and $𝐈_A$ are examples of facets, which we now
introduce more formally.
\begin{defi}
  For any categories $𝐂$ and $𝐃$, a \alert{facet} for $[𝐂,𝐃]_f$ is a
  finitary functor $[𝐂,𝐃]_f → [𝐂,𝐒𝐞𝐭]_f$. 

  Let $𝐅𝐚𝐜𝐞𝐭(𝐂,𝐃) ≔ [[𝐂,𝐃]_f, [𝐂,𝐒𝐞𝐭]_f]_f$ denote the
  category of facets.
\end{defi}

\begin{nota}
  We abbreviate $𝐅𝐚𝐜𝐞𝐭(𝐒𝐞𝐭^ℙ,𝐒𝐞𝐭^𝕊)$ to $𝐅𝐚𝐜𝐞𝐭(ℙ,𝕊)$, for readability.
  \end{nota}

  \begin{defi}
    Here are a few basic constructions of facets, the first three in the
    general case, and the next three for $[𝐒𝐞𝐭^ℙ,𝐒𝐞𝐭^𝕊]_f$.
  \begin{itemize}
  \item Any functor $F∶ 𝐃 → 𝐒𝐞𝐭$ induces a facet $Φ_F$
    defined by $Φ_F(P) = F∘P$.
  \item Similarly, any functor $G∶ 𝐂 → 𝐒𝐞𝐭$ induces a facet $Ψ_G$
    defined by $Ψ_G(P)(C) = G(C)$.
  \item For any facets $F$ and $G$, the product $F×G$ in the
    (functor) category of facets is again a facet.
  \item For any $p ∈ ℙ$, let $Θₚ = Φ_{\evₚ}$.
  \item For any $s ∈ 𝕊$, let $𝐈ₛ = Ψ_{\evₛ}$.
  \item For any facet $F$ for $[𝐒𝐞𝐭^ℙ,𝐒𝐞𝐭^𝕊]$ and $p₁,…,pₙ ∈ ℙ$, let
    \[F^{(p₁,…,pₙ)}(P)(X) = F(P(X + 𝐲_{p₁} + … + 𝐲_{pₙ})).\]
\end{itemize}
\end{defi}
\begin{rem}
  Let $Θ$ denote the identity endofunctor on $[𝐒𝐞𝐭^ℙ,𝐒𝐞𝐭^𝕊]_f$,
  and \[𝐈∶ [𝐒𝐞𝐭^ℙ,𝐒𝐞𝐭^𝕊]_f → [𝐒𝐞𝐭^ℙ,𝐒𝐞𝐭^ℙ]_f \]the constant functor
  mapping anything to the identity endofunctor.  Post-composing with
  evaluation at any $p ∈ ℙ$, resp.\ $s ∈ 𝕊$, we recover the facets
  $𝐈ₚ$ and $Θₛ$.
\end{rem}
\begin{rem}
  The notation $Θ$ introduced above for the identity endofunctor
  is compatible with the one denoting the parametric module
  mapping a monad $T$ on $𝐒𝐞𝐭^ℙ$ to $T$ as a module over itself
  (Example~\ref{ex:basic-param-mod})
  in the sense that, for example, the functor underlying the module $T$ coincides
  with the functor underlying the monad $T$.
\end{rem}
\begin{exa}\label{ex:functors}
  The arities for application binary trees will be given by $𝐈_A → Θ_A$
  and $Θ_{A→B}×Θ_A → Θ_B$, for all types $A$ and $B$.
\end{exa}

\subsection{The register \texorpdfstring{$𝐑𝐞𝐠⁰[𝐒𝐞𝐭^ℙ,𝐒𝐞𝐭^𝕊]_f$}{Reg0[SetP,SetS]f} for specifying operations}
In order to adapt the notion of signature for operations from
$𝐑𝐞𝐠𝐌𝐧𝐝⁰_f(𝐒𝐞𝐭^ℙ)$, we merely need to adapt the notion of elementariness,
which becomes the following:
\begin{defi}
  A facet for $[𝐒𝐞𝐭^ℙ,𝐒𝐞𝐭^𝕊]_f$ is \alert{elementary} if it is
  isomorphic to some finite product of facets of the shape
  $(H∘⟨𝐈,Θ⟩)^{(p₁,…,pₙ)}$ for some $p₁,…,pₙ ∈ ℙ$ and finitary functor
  $H∶ 𝐒𝐞𝐭^ℙ×𝐒𝐞𝐭^𝕊 → 𝐒𝐞𝐭$.
\end{defi}
\begin{rem}
  The notation $(H∘⟨𝐈,Θ⟩)^{(p₁,…,pₙ)}$ deserves some explanation.
  We have $𝐈∶ [𝐒𝐞𝐭^ℙ,𝐒𝐞𝐭^𝕊] → [𝐒𝐞𝐭^ℙ,𝐒𝐞𝐭^ℙ]$
  and $Θ∶ [𝐒𝐞𝐭^ℙ,𝐒𝐞𝐭^𝕊] → [𝐒𝐞𝐭^ℙ,𝐒𝐞𝐭^𝕊]$, and we mean
  \[(H∘⟨𝐈,Θ⟩)^{(p₁,…,pₙ)}(P)(C) ≔ H (C,P(C)).\]
\end{rem}
\begin{exa}
  Typically, any product of facets of the shape $𝐈ₚ^{(p₁,…,pₙ)}$ or
  $Θ_s^{(p₁,…,pₙ)}$, for some $p,p₁,…,pₙ ∈ ℙ$ and $s ∈ 𝕊$, is
  elementary.
\end{exa}

\begin{defi}
  A \alert{facet-based signature} is a family of pairs $(d,s)$ consisting 
  of an elementary facet $d$ and a transition type $s ∈ 𝕊$.
\end{defi}
    
\begin{defi}
  For any facet-based signature $S = (dᵢ,sᵢ)_{i ∈ I}$, an
  \alert{$S$-algebra} is a finitary functor $F∶ 𝐒𝐞𝐭^ℙ → 𝐒𝐞𝐭^𝕊$,
  equipped with natural transformations $dᵢ(F) → Θ_{sᵢ}(F)$ for all 
  $i ∈ I$.  A \alert{morphism of $S$-algebras} is a natural
  transformation commuting with these morphisms.  Let $S\alg$ denote
  the category of $S$-algebras, and $U_S∶ S\alg → [𝐒𝐞𝐭^ℙ,𝐒𝐞𝐭^𝕊]_f$ the
  forgetful functor.
\end{defi}

  \begin{defi}\label{def:elemfunc}
    For any sets $ℙ$ and $𝕊$, we define the monadic register
    $𝐑𝐞𝐠⁰[𝐒𝐞𝐭^ℙ,𝐒𝐞𝐭^𝕊]_f$ as follows.
    \begin{description}
    \item[Signatures] A signature is a facet-based signature.
    \item[Semantics] The \semsig associated to any facet-based
      signature $S$ is the forgetful functor
      $S\alg → [𝐒𝐞𝐭^ℙ,𝐒𝐞𝐭^𝕊]_f$.
\end{description}
\end{defi}
\begin{proof}[{\rm \bf Validity proof:}\nopunct]
  By Proposition~\ref{prop:SigmaSMon:Salgpre} below.
\end{proof}

\subsection{The register \texorpdfstring{$𝐑𝐞𝐠[𝐒𝐞𝐭^ℙ,𝐒𝐞𝐭^𝕊]_f$}{𝐑𝐞𝐠[SetP,SetS]f}}

\begin{defi}
  Given any facet-based signatures $S$ and $S'$, an \alert{$S$-derived
    operation of arity $S'$} is a functor $L∶ S\alg → S'\alg$ over
  $[𝐒𝐞𝐭^ℙ,𝐒𝐞𝐭^𝕊]_f$, i.e., making the following triangle commute.
    \begin{center}
      \diag{%
        S\alg \& \& S'\alg \\
        \& {[𝐒𝐞𝐭^ℙ,𝐒𝐞𝐭^𝕊]_f}
      }{%
        (m-1-1) edge[labela={L}] (m-1-3) %
        edge[labelbl={}] (m-2-2) %
        (m-1-3) edge[labelbr={}] (m-2-2) %
      }
    \end{center}
    We call \alert{basic} the operations of $S$, by contrast with
    the \alert{derived} operations of $S'$.

    More concretely, as in Definition~\ref{def:derivedops}, we may
    introduce derived operations in two stages, as follows:
    \begin{itemize}
    \item an \alert{$S$-facet morphism} $M → N$ between facets $M$ and
      $N$ is a natural family of natural transformations
      $(α_F ∶M(F) ⟶ N(F))_{F∈ S\alg}$;
    \item letting $S' = (Vⱼ,qⱼ)_{j ∈ J}$, an $S$-derived operation $L$
      of arity $S'$ is a family of $S$-facet morphisms $Vⱼ → Θ_{qⱼ}$,
      for all $j ∈ J$.
    \end{itemize}
\end{defi}
\begin{rem}
  Concretely, an $S$-derived operation $L$ of arity $S'$ associates to
  each $S$-algebra $F$ a family of $S$-facet morphisms
  $Vⱼ(F) → F_{qⱼ}$, naturally in $F$.
\end{rem}

\begin{defi}
  An \alert{equational facet-based signature} consists of
  \begin{itemize}
  \item a facet-based signature $S$ called the \alert{operations}
    facet-based signature,
  \item a facet-based signature $S'$ called the \alert{equations}
    facet-based signature, and
  \item a pair of $S$-derived operations of arity $S'$.
\end{itemize}
  For any equational facet-based signature $E = (S,S',L,R)$, an
  \alert{$E$-algebra} is an $S$-algebra $F$ such that the $S'$-algebra
  structures $L(F)$ and $R(F)$ coincide, i.e., $L(F) = R(F)$.  A
  morphism of $E$-algebras is a morphism of $S$-algebras.  We let
  $E\alg$ denote the category of $E$-algebras and morphisms between
  them.
\end{defi}

\begin{defi}\label{def:PModFun}
  For any sets $ℙ$ and $𝕊$, we define the monadic register
  $𝐑𝐞𝐠[𝐒𝐞𝐭^ℙ,𝐒𝐞𝐭^𝕊]_f$ for $[𝐒𝐞𝐭^ℙ,𝐒𝐞𝐭^𝕊]_f$ as follows.
  \begin{description}
  \item[Signatures] A signature is an equational facet-based signature.
  \item[Semantics]
The \semsig associated to any equational facet-based
    signature $E$ is the forgetful functor $E\alg → [𝐒𝐞𝐭^ℙ,𝐒𝐞𝐭^𝕊]_f$.
\end{description}
\end{defi}
\begin{proof}[{\rm \bf Validity proof:}\nopunct]
By Corollary~\ref{cor:PSEmonadic} below.
\end{proof}

\begin{rem}\label{rk:statefuns}
  Any finitary functor $F$ admits a trivial signature consisting of
  the family $((F_s ∘ 𝐈) → Θ_s)_{s ∈ 𝕊}$ of operations, which does not
  prevent other signatures from being more convenient (see the case of
  Example~\ref{ex:functors}).  Here are a few examples
  from~§\ref{s:examples}:
  \begin{center}
    $\begin{array}{lcc}
       \toprule
       \mbox{Language} & \mbox{State functor} & \mbox{Trivial signature} \\
       \midrule
       \Tstrut\Bstrut\overline{λ}μ &
                                        \begin{array}{l}
                                          S₁(X) = S₂(X) \\
                                          = (X_𝐩 × X_𝐬, X_𝐩, X_𝐬)
                                          \end{array}
 &
   \begin{array}{rcl}
     ⟨-|-⟩∶& 𝐈_𝐩 × 𝐈_𝐬& → Θ_𝐜 \\
     η_𝐩 ∶ &𝐈_𝐩& → Θ_𝐩 \\
     η_𝐬 ∶ &𝐈_𝐬& → Θ_𝐬
   \end{array}
       \\ \midrule
      π & S₁(X) = S₂(X) = X_𝐩 & 𝐈_𝐩 → Θ \\ \midrule
      \mbox{Call-by-value, simply-typed $λ$} & S₂(X) = X & ηₜ∶ 𝐈ₜ → Θₜ \mbox{\ \
      (for all $t$)}\\ \midrule
      \mbox{Positive GSOS specifications} & S₁(X) = X & 𝐈 → Θ \\
                     & S₂(X) = 𝔸×X & 𝔸×𝐈 → Θ \\ \bottomrule
    \end{array}$
  \end{center}
\end{rem}

\begin{nota}\label{not:endoformat}  
  In examples, we will present equational facet-based signatures as
  families of equations. For this, we will use
  Notation~\ref{not:monadformat}, which extends to facet-based
  equations. E.g., Example~\ref{ex:paranot} applies verbatim for
  associativity of multiset union in the target state functor for
  differential $λ$-calculus.
\end{nota}

\subsection{Explicit description of initial algebras}
In this section, we state monadicity of the register
$𝐑𝐞𝐠[𝐒𝐞𝐭^ℙ,𝐒𝐞𝐭^𝕊]_f$ of equational facet-based signatures, and give our
explicit description of initial algebras.  We first deal with facet-based
signatures $S$, identifying the initial $S$-algebra as a free algebra
for a suitable endofunctor.  We then characterise the initial
$E$-algebra as a coequaliser of free algebras, for any equational
facet-based signature $E$.

  \begin{defi}\label{def:SigmaSpre}
    For any facet-based signature $S$, we define its \alert{associated
      endofunctor} $Σ_S$ on $[𝐒𝐞𝐭^ℙ,𝐒𝐞𝐭^𝕊]_f$ as follows.
    \begin{itemize}
    \item For any elementary $d = (F ∘ ⟨𝐈,Θ⟩)^{(p₁,…,pₙ)}$ and $s ∈ 𝕊$, 
      let
        \[Σ_{(d,s)}(P) = (F ∘ ⟨\id,P⟩)^{(p₁,…,pₙ)}·𝐲_s\rlap{.}\]
  \item     For a family $S = (dᵢ,sᵢ)$, let
    $Σ_S ≔ ∑_{i ∈ I} Σ_{(dᵢ,sᵢ)}.$
  \end{itemize}
  \end{defi}
  \begin{rem}
    Explicitly, we have for any $(d,s)$:
    \begin{center}
    $Σ_{(d,s)}(P)(X) ≔ (F(X + ∑ᵢ 𝐲_{pᵢ},P(X + ∑ᵢ 𝐲_{pᵢ})))·𝐲_s\rlap{.}$
  \end{center}
  \end{rem}

    \begin{prop}\label{prop:monadic:premodular}
    For any facet-based signature $S$, the forgetful functor
    $U_S∶ S\alg → [𝐒𝐞𝐭^ℙ,𝐒𝐞𝐭^𝕊]_f$ is monadic. Furthermore, the free
    $S$-algebra on a functor $F ∈ [𝐒𝐞𝐭^ℙ,𝐒𝐞𝐭^𝕊]_f$ is the free
    $Σ_S$-algebra $Σ_S^*(F)$, as characterised in
    Proposition~\ref{prop:reiterman}.
  \end{prop}
  \begin{proof}
    A direct consequence of Proposition~\ref{prop:SigmaSMon:Salgpre}
    below.
  \end{proof}

  We now want to characterise the initial $E$-algebra, for any
  equational facet-based signature $E$. The development closely
  follows~§\ref{ss:ia:mnd}.

\begin{defi}
  Let $E = (S,S',L,R)$ denote any equational facet-based signature,
  with $S = (Vᵢ,pᵢ)_{i ∈ I}$ and $S' = (Wⱼ,qⱼ)_{j ∈ J}$.
  \begin{itemize} 
  \item Let $S+S'$ denote the ``disjoint union'', i.e., the
    facet-based signature $(Uₖ,rₖ)_{k ∈ K}$, where
    \begin{itemize}
    \item $K = I+J$,
  \item $(Uₖ,rₖ)$ is
    \begin{itemize}
    \item $(Vᵢ,pᵢ)$ if $k = in₁(i)$ and
    \item  $(Wⱼ,qⱼ)$ if $k = in₂(j)$.
    \end{itemize}
  \end{itemize}
\item The $S'$-algebra structures given by $L(\spec{S})$ and
  $R(\spec{S})$ on $\spec{S}$, together with its canonical $S$-algebra
  structure, yield two $(S+S')$-algebra structures.
  By initiality of $\spec{S+S'}$, we thus obtain
  two $(S+S')$-algebra morphisms
  \begin{equation} \label{eq:LR''monfacets}
    \check{L},\check{R}∶  \spec{S+S'} → \spec{S}.
  \end{equation}    
\end{itemize}
\end{defi}

\begin{thm}\label{prop:statefuns}
    For any equational facet-based signature $E= (S,S',L,R)$,
  the coequaliser of the pair
  $\check{L},\check{R}∶ \spec{S+S'} → \spec{S}$ in $[𝐒𝐞𝐭^ℙ,𝐒𝐞𝐭^𝕊]_f$
  admits a unique $S$-algebra structure such that the coequalising
  morphism is an $S$-algebra morphism.  Furthermore, this structure is
  in fact an $E$-algebra structure. Finally, it makes the coequaliser
  into an initial $E$-algebra.
\end{thm}

\section{General registers}\label{sec:gen-signatures}
We have now introduced all the needed registers for our register
$𝐑𝐞𝐠𝐓𝐫𝐚𝐧𝐬𝐌𝐧𝐝_{ℙ,𝕊}$ of Definition~\ref{d:omndreg} to make sense,
including two registers featuring equations, respectively for monads
(Definition~\ref{def:modmnd}) and functors
(Definition~\ref{def:PModFun}).  In this section, in preparation for
the missing validity proofs and the announced explicit descriptions of
initial algebras, we introduce a fundamental register featuring
equations, for a general category, $𝐄𝐒_𝐂$, whose signatures are Fiore
and Hur's equational systems~\cite{FioreHurEquational}.  In order to
deal more specifically with monads, we then introduce a second
register, $𝐏𝐒𝐄𝐅_𝐂$, based on Fiore, Plotkin, and Turi's pointed strong
endofunctors, which incorporates variable binding and substitution.
We then ``merge'' both registers into a single register $𝐌𝐄𝐒_𝐂$,
suited for syntax with variable binding, substitution, and equations.
Proofs are again deferred to~§\ref{s:friendliness}.

\subsection{Equational systems}\label{ss:ES}
A general device for constructing monadic functors is Fiore and
Hur's \alert{equational systems}~\cite{FioreHurEquational}.  In this
section, we view equational systems on a category $𝐂$ as
the signatures of a register $𝐄𝐒_𝐂$ for $𝐂$, called the \alert{equational
  register}, which refines with equations the endofunctorial register
$𝐄𝐅_𝐂$ (see Definition~\ref{def:endos}).

Briefly, an equational system consists of two parts:
\begin{itemize}
\item an endofunctor $Σ$ on $𝐂$, which intuitively specifies operations,
\item and equations.
\end{itemize}

We present them in a slightly non-standard way, and make the link with
the original in Propositions~\ref{prop:alteqsys}
and~\ref{prop:alteqsysalg} below.
\begin{defi}
  For any endofunctor $Γ$ and monad $T$ on a category $𝐂$, a
  \alert{functorial $T$-term} of arity $Γ$ is a natural transformation
  $Γ → T$.
\end{defi}
\begin{exa}\label{ex:assocES}
  Let us consider the finitary endofunctor $Σ$ on sets defined by
  $Σ(X) = X²$, so that $Σ^*(X)$ denotes the set of binary trees with
  leaves in $X$, as generated by the following grammar,
  \[e,f \Coloneqq x ｜ \ope(e,f)\]
  where $x$ ranges over $X$.
  Since $Σ$-algebras are sets equipped with a binary operation, a
  natural equation to impose on them is associativity.  Taking $Γ(X)=X³$,
  the relevant functorial $Σ^*$-terms $L$ and $R$ of arity $Γ$ are defined by
  $L(x₁,x₂,x₃) = \ope(\ope(x₁,x₂),x₃)$, and
  $R(x₁,x₂,x₃) = \ope(x₁,\ope(x₂,x₃))$.
\end{exa}
\begin{exa}
  We will do this right in Example~\ref{ex:betaeq} below, but for
  illustrative purposes, let us describe a failed attempt at
  specifying pure $λ$-terms modulo $β$ by an equational system.  We
  first take $𝐂 = [𝐒𝐞𝐭,𝐒𝐞𝐭]_f$ to consist of finitary endofunctors on
  sets, and $Σ(X)(n) = n + X(n)² + X(n+1)$. The first summand models
  variables, the second application, and the third, abstraction.
  Indeed, any algebra $X$ comes equipped with maps
  $\appₙ∶ X(n)² → X (n)$ and $λₙ∶ X(n+1) → X(n)$, for all $n$. The
  first member of the $β$-equation, $\appₙ(λₙ(e),f)$, would be
  modelled by setting $Γ(X)(n) = X(n+1) × X(n)$ and taking the natural
  transformation $Γ → Σ^*$ mapping any $(e,f) ∈ X(n+1)×X(n)$ to
  $\appₙ(λₙ(e),f)$ (omitting the monad unit $η∶ \id → Σ^*$).  This
  works fine, but we have neglected to build substitution into the
  model, so we cannot define the right-hand side of $β$.  This will be
  rectified in Example~\ref{ex:betaeq}, after observing that (part of)
  $Σ$ is in fact pointed strong in Example~\ref{ex:lam}.
\end{exa}
The following should now look natural.
\begin{defi}\label{def:eqsys}
  An \alert{equational system} $𝔼 = (𝐂∶ Γ ⊢ L = R : Σ)$ consists of
  \begin{itemize}
  \item a locally finitely presentable category $𝐂$, 
  \item finitary endofunctors $Σ$ and $Γ$, together with
  \item functorial $Σ^*$-terms $L$ and $R$ of arity $Γ$.
\end{itemize}
\end{defi}

This differs slightly from Fiore and Hur's~\cite{FioreHurEquational}
definition, so let us readily bridge the gap.

\begin{prop}\label{prop:alteqsys}
  Given any finitary endofunctors $Σ$ and $Γ$ on a locally finitely
  presentable category $𝐂$, there are natural isomorphisms (in $Σ$ and
  $Γ$) between
  \begin{enumerati}
  \item \label{item:GS} functorial $Σ^*$-terms of arity $Γ$, i.e., natural
    transformations $Γ → Σ^*$,
  \item \label{item:monads} monad morphisms $Γ^* → Σ^*$,
  \item \label{item:monadic} functors $Σ\alg → Γ\alg$ over $𝐂$,
  \item \label{item:reflmonads} monad morphisms $(Σ+Γ)^* → Σ^*$ with
    section the canonical morphism $Σ^* → (Σ+Γ)^*$,
  \item \label{item:reflmonadic} functors $Σ\alg → (Σ+Γ)\alg$ over $𝐂$
    which are sections of the canonical functor $(Σ+Γ)\alg → Σ\alg$,
    and
  \item \label{item:ded} natural transformations $Γ ∘ U_Σ → U_Σ$.
  \end{enumerati}
\end{prop}
\begin{proof} \ \hfill
  \begin{itemize}
  \item \ref{item:GS} $⇔$ \ref{item:monads} This is precisely the
    universal property of $Γ^*$.
  \item \ref{item:monads} $⇔$ \ref{item:monadic} This is precisely
    Corollary~\ref{cor:street}.
  \item \ref{item:monads} $⇔$ \ref{item:reflmonads} By
    $(Σ+Γ)^* ≅ Σ^* + Γ^*$, which holds (in $𝐌𝐧𝐝_f(𝐂)$) because $(-)^*$
    is a left adjoint, hence preserves coproducts.
  \item \ref{item:monadic} $⇔$ \ref{item:reflmonadic} By
    $(Σ+Γ)\alg ≅ {Σ\alg} ×_𝐂 {Γ\alg}$, with the canonical functor
    $(Σ+Γ)\alg → Σ\alg$ as left projection.
  \item \ref{item:monadic} $⇔$ \ref{item:ded} By merely unfolding
    definitions. \qedhere
\end{itemize}
\end{proof}

\begin{rem}\label{rk:FH}
  The expert will have noticed that Fiore and Hur's equational systems
  use~\ref{item:monadic}, in a slightly more general setting:
  they do not assume $𝐂$ to be locally finitely presentable, nor $Σ$
  and $Γ$ to be finitary.
\end{rem}

Let us now define the models of an equational system.

\begin{defi}
  For any equational system $𝔼 = (𝐂∶ Γ ⊢ L = R : Σ)$, an
  \alert{$𝔼$-algebra} is a $Σ$-algebra $X$ whose induced $Σ^*$-algebra
  structure coequalises
  \[L_X,R_X∶ Γ(X) ⇉ Σ^*(X).\]
  Let $𝔼\alg$ denote the category of $𝔼$-algebras, with $Σ$-algebra
  morphisms between them, and let $U_𝔼∶ 𝔼\alg → 𝐂$ denote the
  forgetful functor.
\end{defi}

Let us readily transfer this definition across the various
correspondences of Proposition~\ref{prop:alteqsys}.
\begin{prop}\label{prop:alteqsysalg}
  For any equational system $𝔼 = (𝐂∶ Γ ⊢ L = R : Σ)$ and $Σ^*$-algebra
  $a∶ Σ^*(X) → X$, the following are equivalent:
  \begin{enumerata}
  \item \label{item:alg:GS} $a$ coequalises $L_X,R_X∶ Γ(X) → Σ^*(X)$,
  \item \label{item:alg:monads} $a$ coequalises the corresponding morphisms $Γ^*(X) → Σ^*(X)$,
  \item \label{item:alg:monadic} the induced $Σ$-algebra structure
    $Σ(X) → X$ belongs to the equaliser of the corresponding functors
    $Σ\alg → Γ\alg$ over $𝐂$,
  \item \label{item:alg:reflmonads} $a$ coequalises the corresponding
    morphisms $(Σ+Γ)^*(X) → Σ^*(X)$,
  \item \label{item:alg:reflmonadic} the induced $Σ$-algebra structure
    $Σ(X) → X$ belongs to the equaliser of the corresponding functors
    $Σ\alg → (Σ+Γ)\alg$ over $𝐂$,
  \item \label{item:alg:ded} the corresponding natural transformations
    $Γ ∘ U_Σ → U_Σ$ have the same components at the induced
    $Σ$-algebra structure $Σ(X) → X$.
  \end{enumerata}
\end{prop}
\begin{proof} \ \hfill
  \begin{itemize}
  \item \ref{item:alg:GS} $⇔$ \ref{item:alg:monadic} Each functor
    $Σ\alg → Γ\alg$ corresponding to $K ∈ \{L,R\}$ maps the induced
    $Σ$-algebra to $Γ(X) \xto{K} Σ^*(X) \xto{a} X$, so both sides
    unfold to the same thing.
  \item \ref{item:alg:monads} $⇒$ \ref{item:alg:GS} Follows from the
    fact that precomposing the induced morphisms $Γ^* → Σ^*$ by
    $Γ → Γ^*$ yields the original $L$ and $R$ by construction.
  \item \ref{item:alg:monadic} $⇒$ \ref{item:alg:monads} This holds
    because~\ref{item:alg:monads} is equivalent to equality of induced
    $Γ^*$-algebra structures, which is further equivalent to equality
    of induced $Γ$-algebra structures.  
  \end{itemize}
  The rest follows similarly.
\end{proof}
\begin{rem}
  The equivalence of~\ref{item:alg:GS} and~\ref{item:alg:monadic}
  entails that our notion of algebra coincides with Fiore and Hur's
  (apart from the differences noted in Remark~\ref{rk:FH}).
\end{rem}

\begin{rem}
  Families $(tᵢ = uᵢ)ᵢ$ of equations are covered by taking the
  coproduct of all involved endofunctors, say $Γᵢ$, and the pointwise
  copairing of functorial terms.
\end{rem}

\begin{defi}
  \label{def:eqsysreg}
  For a given locally finitely presentable category $𝐂$, we define the
  monadic register $𝐄𝐒_𝐂$, called the \alert{equational register}, as
  follows.
  \begin{description}
  \item[Signatures] A signature is an equational system.
  \item[Semantics] The \semsig associated to any signature $𝔼 $ is the
    forgetful functor $𝔼\alg → 𝐂$.
  \end{description}
\end{defi}
\begin{proof}[{\rm \bf Validity proof:}\nopunct]
  By Theorem~\ref{thm:ES} below.
\end{proof}

\begin{exa}
 Given a cocomplete category $𝐂$, the endofunctor register
  $𝐄𝐅_𝐂$ is a subregister of $𝐄𝐒_𝐂$, by mapping any finitary
  endofunctor $Σ$ to the (finitary) equational system given by
  taking $Γ = 0$, and the unique $L$ and $R$.
\end{exa}

\subsection{The register \texorpdfstring{$𝐏𝐒𝐄𝐅_𝐂$}{PSEFC} for monoids}\label{ss:PSEF}
In the previous subsection, we reviewed a fundamental register
available for general locally finitely presentable categories. Here,
we review another fundamental register called $𝐏𝐒𝐄𝐅_𝐂$, available for the
category of monoids in a convenient monoidal category. This register
is essentially due to Fiore, Plotkin, and Turi~\cite{fiore:presheaf}.
Signatures in the register $𝐏𝐒𝐄𝐅_𝐂$ will be pointed strong
endofunctors on the monoidal category $𝐂$.

The point of the register $𝐏𝐒𝐄𝐅_𝐂$ is to specify monads.  In fact,
monads may be specified by the previous register, $𝐄𝐒_𝐂$, but at the
cost of including in the signature
\begin{itemize}
\item operations for the monad multiplication and unit, and
\item equations for associativity and unitality.
\end{itemize}
Instead, the register $𝐏𝐒𝐄𝐅_𝐂$ will directly specify monoids in any
(sufficiently nice) monoidal category $𝐂$, hence be more
economical. This in particular covers the case of finitary monads on a
locally finitely presentable category.

Let us first describe the announced monadic \semsig.  We start by
recalling what a pointed strong endofunctor is, and showing how it
yields a parametric module (in the sense of
Definition~\ref{def:paramX}) on $𝐂$.

\mypar{Pointed strong endofunctors}
\begin{defi}
  \label{d:pstrength-endo}
  A \alert{pointed strong endofunctor} on a monoidal category
  $(𝐂,⊗,I)$ is an endofunctor $Σ$ equipped with a \alert{pointed
    strength}, i.e., a natural transformation with components
  \[st_{X,Y}∶ Σ(X)⊗Y → Σ(X⊗Y)\]
between functors $𝐂 × I/𝐂 → 𝐂$, making the following diagrams commute
(pointed objects $e_Y∶ I → Y$ are denoted by their codomains $Y$ for readability),
\begin{center}
  \diag{%
    (Σ(X)⊗Y)⊗Z \& \& Σ(X)⊗(Y⊗Z) \\
    Σ(X⊗Y)⊗Z \& Σ((X⊗Y)⊗Z) \& Σ (X⊗(Y⊗Z)) %
  }{%
    (m-1-1) edge[labela={α_{Σ(X),Y,Z}}] (m-1-3) %
    edge[labell={st_{X,Y}⊗Z}] (m-2-1) %
    (m-2-1) edge[labelb={st_{X⊗Y,Z}}] (m-2-2) %
    (m-2-2) edge[labelb={Σ(α_{X,Y,Z})}] (m-2-3) %
    (m-1-3) edge[labelr={st_{X,Y\dot{⊗}Z}}] (m-2-3) %
  }%
  \hfil
  \diag{%
    Σ(X) \& \& Σ(X)⊗I \\
     \& Σ(X⊗I) %
   }{%
     (m-1-1) edge[labela={ρ_{Σ(X)}}] (m-1-3) %
     edge[labelbl={Σ(ρ_X)}] (m-2-2) %
     (m-1-3) edge[labelbr={st_{X,\dot{I}}}] (m-2-2) %
  }%
  \end{center}
  where
  $I/𝐂$ inherits monoidal structure $(\dot{I},\dot{⊗},…)$ from
    $𝐂$ in the obvious way~\cite{Weber:genmorph}.

\end{defi}

\begin{rem}
  There is a simpler
  notion of \alert{strong} endofunctor, which also involves a natural
  transformation $st_{X,Y}∶ Σ(X)⊗Y → Σ(X⊗Y)$, although with $Y$ not
  pointed.  Strong endofunctors embed into pointed strong
  endofunctors, which are thus more general. The generalisation is
  necessary to cover variable binding, as shown by the following
  example.
\end{rem}

\begin{exa}\label{ex:lam}
  For defining $λ$-calculus syntax, we take $𝐂 ≔ [𝐒𝐞𝐭,𝐒𝐞𝐭]_f$ to
  be the category of finitary endofunctors on sets with the
  composition tensor product, and $Σ(X)(n) = X(n)² + X(n+1)$ (using
  the equivalence $[𝐒𝐞𝐭,𝐒𝐞𝐭]_f ≃ [𝐒𝐞𝐭_f,𝐒𝐞𝐭]$, where $𝐒𝐞𝐭_f$ is the
  full category spanning finite ordinals).  The tensor product is
  defined on $[𝐒𝐞𝐭_f,𝐒𝐞𝐭]$ by the coend formula~\cite[§IX.6]{MacLane:cwm}
  \[(X⊗Y)(n) = ∫ᵐX(m)×Y(n)ᵐ\rlap{,}\]
  which is in one-to-one correspondence with $X (Y (n))$.
  Elements of $(X⊗Y)(n)$ may be thought of as explicit
  substitutions $x⦇σ⦈$, with $x ∈ X(m)$ and $σ∶ m → Y(n)$, considered
  equivalent up to standard equations, compactly summarised by
  $(f·x)⦇σ⦈ = x⦇σ∘f⦈$, for all $f∶ m → p$, $x ∈ X(m)$ and
  $y∶ p → Y(n)$, where $f·x$ is shorthand for $X(f)(x)$.

  A monoid structure on any $X ∈ 𝐂$ thus amounts to
  \begin{itemize}
  \item a substitution operation $X⊗X → X$ mapping any such explicit
    substitution $x⦇σ⦈ ∈ (X⊗X)(n)$ to some proper substitution, which
    we denote by $x[σ] ∈ X(n)$,
  \item together with a morphism $I → X$, which, because
    $I(n) = \id(n) = n$, amounts to identifying available variables
    within each $X(n)$.
  \end{itemize}
  These data are required to satisfy the usual associativity and
  unitality conditions, which amount to standard substitution lemmas.

  Returning to pointed strengths
  $(st_{X,Y})ₙ∶ Σ(X) (Y(n)) → Σ (X (Y (n)))$, intuitively, they
  describe the behaviour of substitution w.r.t.\ application and
  abstraction.  E.g., for application, we define it to map
  $(in₁(x₁,x₂))⦇σ⦈$ to $in₁(x₁⦇σ⦈,x₂⦇σ⦈)$, which will ensure the usual
  equation $(e₁ e₂)[σ] = e₁[σ] e₂[σ]$.  Abstraction is the point where
  \alert{pointedness} of $st$ comes into play.  Indeed, supposing that
  $Y$ is equipped with a point $e_Y∶ I → Y$, we may define
  $σ^↑∶ m+1 → Y(n+1)$ by copairing of
  \begin{center}
    $m \xto{σ} Y(n) \xto{Y(in₁)} Y(n+1)$ \hfil and \hfil
    $1 \xto{(e_Y)₁} Y(1) \xto{Y(in₂)} Y(n+1)$.
  \end{center}
  We use this in defining the pointed strength to map any $in₂(x)⦇σ⦈$,
  where $x ∈ X(m+1)$ and $σ∶ m → Y(n)$, to $in₂(x⦇σ^↑⦈)$. This will
  ensure the usual equation $λ(e)[σ] = λ(e[σ^↑])$.
\end{exa}

\mypar{The parametric module $Σ\pmod$} Let us now show how every
pointed strong endofunctor induces a parametric module.
\begin{defi}\label{def:sigmas}
We define the parametric
module $Σ\pmod$ on $𝐂$ as assigning to any monoid $X ∈ 𝐂$, the object
$Σ(X)$, equipped with the action given by the composite
  \[Σ(X)⊗ X → Σ(X⊗X) → Σ(X).\]
\end{defi}
The verification that this assignment indeed defines a parametric module is
straightforward.

\mypar{The category $Σ\Mon$ of models} We now introduce
$Σ$-monoids. These are exactly the classical $Σ$-monoids, which we
present from a point of view better suited to our purpose. They are
monoids equipped with a ``compatible'' algebra structure:
\begin{defi}
  A \alert{$Σ$-monoid} is a monoid $X$, equipped with an $X$-module morphism
  $ν_X∶ Σ\pmod(X) → X$. A $Σ$-monoid morphism is a monoid morphism $f∶ X →
Y$
  making the following square commute.
  \begin{center}
    \diag{%
      Σ\pmod(X) \&       Σ\pmod(Y) \\
      X \& Y %
    }{%
      (m-1-1) edge[labela={Σ\pmod(f)}] (m-1-2) %
      edge[labell={ν_X}] (m-2-1) %
      (m-2-1) edge[labelb={f}] (m-2-2) %
      (m-1-2) edge[labelr={ν_Y}] (m-2-2) %
    }
  \end{center}
  We let $Σ\Mon$ denote the category of $Σ$-monoids and morphisms between
them, while $U_Σ∶ Σ\Mon → 𝐌𝐨𝐧(𝐂)$ denotes the obvious forgetful functor.
\end{defi}

Let us first prove that this agrees with the standard definition.
\begin{prop}
  The category $Σ\Mon$ is isomorphic over $𝐌𝐨𝐧(𝐂)$ to the following category.
  \begin{itemize}
  \item Objects are monoids $X$ equipped with $Σ$-algebra structure
    $ν_X∶ Σ(X) → X$ making the diagram
  \begin{equation}
    \diag{%
      Σ(X)⊗X \& Σ (X⊗X) \& Σ(X) \\
      X⊗X \& \& X %
    }{%
      (m-1-1) edge[labela={st_{X,X}}] (m-1-2) %
      edge[labell={ν_X ⊗ X}] (m-2-1) %
      (m-1-2) edge[labela={Σ(m_X)}] (m-1-3) %
      (m-2-1) edge[labelb={m_X}] (m-2-3) %
      (m-1-3) edge[labelr={ν_X}] (m-2-3) %
    }
    \label{eq:monoidalg}
  \end{equation}
  commute.
\item Morphisms are morphisms in $𝐂$ which are both monoid and
  $Σ$-algebra morphisms.
  \end{itemize}
\end{prop}
\begin{proof}
  By definition of $Σ\pmod$, it is equivalent for a morphism
  $Σ(X) → X$ to be an $X$-module morphism, and for the
  condition~\eqref{eq:monoidalg} to hold.
\end{proof}
Let us finally check that we have defined \asemsig.

\begin{thm}\label{thm:sigmamonadic}
  For any finitary, pointed strong endofunctor on a convenient
  monoidal category $𝐂$, the category $Σ\Mon$ and the forgetful
  functor $U_Σ∶ Σ\Mon → 𝐌𝐨𝐧(𝐂)$ form a monadic \semsig.  In fact, the
  forgetful functor $U_Σ∶ Σ\Mon → 𝐌𝐨𝐧(𝐂)$ is finitary monadic.
\end{thm}
\begin{proof}
  We have a commuting triangle
  \begin{center}
    \diag{%
      Σ\Mon \& \& 𝐌𝐨𝐧(𝐂) \\
      \& 𝐂\rlap{,} %
    }{%
      (m-1-1) edge[labela={U_Σ}] (m-1-3) %
      edge[labelbl={𝐔^{Σ\Mon}}] (m-2-2) %
      (m-1-3) edge[labelbr={𝐔^{𝐌𝐨𝐧(𝐂)}}] (m-2-2) %
    }
  \end{center}
  where, by~\cite[§7.2]{FioreHurEquational}, $𝐔^{Σ\Mon}$ and
  $𝐔^{𝐌𝐨𝐧(𝐂)}$ are both monadic, the corresponding monads being
  finitary (noting that $𝐌𝐨𝐧(𝐂)≅ 0\Mon$).
  By~\cite[Remark~2.78]{Adamek} with $λ = ω$, it follows that all
  three categories are locally finitely presentable, hence in
  particular cocomplete.  The result thus follows by
  Proposition~\ref{prop:monadic} and Lemma~\ref{lem:finitary}.
\end{proof}

We now want to recall a standard characterisation of the monad induced
by the monadic functor $U^{Σ\Mon}$, but before that let us fix some
notation.

\begin{nota}\label{not:flocon}
  Let us refine Notation~\ref{not:star} and Remark~\ref{rk:ambiguity}.
  For a pointed strong endofunctor $Σ$ on a convenient monoidal
  category $𝐂$, $Σ^⋆(0)$ might be read as different objects of $𝐂$,
  according to whether $Σ$ is viewed as a finitary endofunctor or a
  finitary pointed strong endofunctor.  Since $Σ$-monoids are also
  monadic over monoids, it might even be understood as a monoid in
  $𝐂$.  We choose the following convention:
  \begin{itemize}
  \item  $Σ^⋆$ denotes the ``free $Σ$-monoid'' monad $𝐂$,
  \item  $Σ^⊛$ denotes the ``free $Σ$-monoid'' monad on
    $𝐌𝐨𝐧(𝐂)$, 
   \item $Σ^*$ denotes the ``free $Σ$-algebra'' monad on $𝐂$.
  \end{itemize}
\end{nota}
\begin{prop}\label{prop:sigmamonchar}
  For any pointed strong endofunctor $Σ$ on a convenient monoidal
  category $𝐂$, the monad $Σ^⋆$ induced by the right adjoint functor
  $Σ\Mon → 𝐌𝐨𝐧(𝐂) → 𝐂$ maps any object $C ∈ 𝐂$ to
  $μA.(I + Σ(A) + C⊗A)$.
  As a consequence of Proposition~\ref{prop:reiterman},
  the carrier $Σ^⋆(0)$ of the initial $Σ$-monoid is then
  isomorphic to $Σ^*(I)$.
\end{prop}
\begin{proof}
  This is more or less known
  since~\cite{fiore:presheaf,FioreHurEquational,DBLP:conf/rta/FioreS17},
  see~\cite[Theorem~2.15]{BHL} for an explicit, complete, yet slightly
  more general statement.
\end{proof}

Let us finally define our register for monoids.
\begin{defi}\label{def:fpt}
  For any convenient monoidal category $𝐂$, the
  monadic 
  register $𝐏𝐒𝐄𝐅_𝐂$ for the category $𝐌𝐨𝐧(𝐂)$ of monoids in $𝐂$ is
  defined as follows.
  \begin{description}
  \item[Signatures] A signature is a finitary, pointed strong
    endofunctor on $𝐂$.
  \item[Semantics] The \semsig associated to any signature $Σ$ is the
    forgetful functor $Σ\Mon → 𝐌𝐨𝐧(𝐂)$.
  \end{description}
\end{defi}
\begin{proof}[{\rm \bf Validity proof:}\nopunct]
  By Theorem~\ref{thm:sigmamonadic}.
\end{proof}

\subsection{Monoidal equational systems}\label{ss:MES}
In this section, we combine the registers $𝐄𝐒_𝐂$ and $𝐏𝐒𝐄𝐅_𝐂$ of the
two previous sections, yielding a new register $𝐌𝐄𝐒_𝐂$ obtained 
from $𝐏𝐒𝐄𝐅_𝐂$ by `adding monoidal equations'.

\begin{defi}
  Given finitary, pointed strong endofunctors $Σ$ and $Γ$ on a
  convenient monoidal category $𝐂$, a \alert{monoidal functorial
    $Σ^⊛$-term} of arity $Γ$ is a parametric module morphism
  $Γ\pmod → (Σ^⊛)\pmodmon$, where we recall $Γ\pmod$ from
  Definition~\ref{def:sigmas} and $(Σ^⊛)\pmodmon$ from
  Notation~\ref{not:flocon} and Example~\ref{ex:Tmod}.
\end{defi}
\begin{defi}
  A \alert{monoidal equational system} $𝔼 = (𝐂∶ Γ ⊢ L =_⊗ R : Σ)$ consists of
  \begin{itemize}
  \item a convenient monoidal category $𝐂$,
  \item two pointed strong endofunctors $Σ$ and $Γ$ on $𝐂$, and
  \item a parallel pair $L,R∶ Γ\pmod ⇉ (Σ^⊛)\pmodmon$ of monoidal
    functorial $Σ^⊛$-terms, which we call a \alert{monoidal equation}.
  \end{itemize}
\end{defi}

\begin{prop}\label{prop:mon:alteqsys}
  Given any locally finitely presentable category $𝐂$, there are
  natural isomorphisms (in $Σ$ and $Γ$) between
  \begin{enumerati}
  \item \label{item:mon:GS} monoidal functorial $Σ^⊛$-terms of arity
    $Γ$, or in other words parametric module morphisms $Γ\pmod → (Σ^⊛)\pmodmon$
    (Definition~\ref{def:paramodule}),
  \item \label{item:mon:monads} morphisms $Γ^⊛ → Σ^⊛$ of monads on $𝐌𝐨𝐧(𝐂)$,
  \item \label{item:mon:monadic} functors $Σ\Mon → Γ\Mon$ over $𝐌𝐨𝐧(𝐂)$,
  \item \label{item:mon:reflmonads} morphisms $(Σ+Γ)^⊛ → Σ^⊛$ of monads on
    $𝐌𝐨𝐧(𝐂)$, with section the canonical morphism $Σ^⊛ → (Σ+Γ)^⊛$,
  \item \label{item:mon:reflmonadic} functors $Σ\Mon → (Σ+Γ)\Mon$ over
    $𝐌𝐨𝐧(𝐂)$ which are sections of the forgetful functor
    $(Σ+Γ)\Mon → Σ\Mon$, and
  \item \label{item:mon:ded} natural transformations
    \begin{center}
      \diag(.3,1){%
         \& 𝐌𝐨𝐧 \\
         Σ\Mon \& \& 𝐌𝐨𝐝\rlap{.} \\
         \& 𝐌𝐨𝐧 %
       }{%
         (m-2-1) edge[labelal={U_Σ}, bend left=10] (m-1-2) %
         (m-2-1) edge[labelbl={U_Σ}, bend right=10] (m-3-2) %
         (m-1-2) edge[labelar={Γ\pmod}, bend left=10] (m-2-3) %
         (m-3-2) edge[labelbr={(-)\pmodmon}, bend right=10] (m-2-3) %
         (m-1-2) edge[labell={},cell=.3] (m-3-2) 
      }
    \end{center}
  \end{enumerati}
\end{prop}
\begin{proof}
  Both equivalences \ref{item:mon:monads} $⇔$ \ref{item:mon:monadic}
  and \ref{item:mon:reflmonads} $⇔$ \ref{item:mon:reflmonadic} follow
  directly from Proposition~\ref{prop:street} because by definition we have
  $Σ^⊛\alg ≅ Σ\Mon$, and similarly for $Γ$ and $Σ+Γ$.  Furthermore,
  \ref{item:mon:ded} $⇔$ \ref{item:mon:monadic} and
  \ref{item:mon:monadic} $⇔$ \ref{item:mon:reflmonadic} both follow by
  mere definition unfolding.  This leaves us with proving that the
  first point agrees with one of the others.

  For this, we exhibit a natural isomorphism
  \begin{mathpar}
    \mprset{fraction={===}} \inferrule*{ T\alg → Γ\Mon \mbox{ (over
        $𝐌𝐨𝐧(𝐂)$)}}{ Γ\pmod → T\pmodmon }
\end{mathpar}
for any monad $T$ on $𝐌𝐨𝐧(𝐂)$. The result thus follows by taking
$T = Σ^⊛$. 
\begin{itemize}
\item Given any $L∶ T\alg → Γ\Mon$ over $𝐌𝐨𝐧(𝐂)$, consider the
  morphism $L^↓∶ Γ\pmod → T\pmodmon$ defined at any monoid $M$ by
    \[Γ(M) \xto{Γ(η_M)} Γ (T (M)) \xto{L(T(M))} T(M)\rlap{.}\]
  This indeed defines a morphism of the claimed type, by a simple
  diagram chasing, and furthermore this assignment is clearly natural
  in $Γ$ and $T$.
\item Conversely, given any $α∶ Γ\pmod → T\pmodmon$, let
  $α^↑∶ T\alg → Γ\Mon$ map any $T$-algebra structure $a∶ T(M) → M$ on a monoid $M$
  to the $Γ$-algebra structure
    \[Γ(M) \xto{α_M} T(M) \xto{a} M\rlap{.}\]
  Again, a simple diagram chase shows that this $Γ$-algebra structure
  satisfies the coherence law of $Γ$-monoids.
\end{itemize}
These two maps are easily checked to be mutually inverse, thus proving
the claim.
\end{proof}

\begin{exa}\label{ex:betaeq}
  Recall from Example~\ref{ex:lam} the pointed strong endofunctor
  $Σ(X)(n) = X(n)² + X(n+1)$ for pure $λ$-calculus on $[𝐒𝐞𝐭,𝐒𝐞𝐭]_f$.
  As promised, let us now use this to complete the aborted treatment
  of the $β$-equation in Example~\ref{ex:betaeq}.  This is made
  possible by working directly at the level of monoids (which we think
  of as objects equipped with substitution). We again take
  $Γ(X)(n) = X(n+1) × X(n)$, and define $L$ and $R$ at any $Σ$-monoid
  $T$ to map any $(f,e) ∈ T(n+1)×T(n)$ to $λ(f)\ e$ and $f[e]$
  respectively, where $f[e]$ denotes the result of substituting $e$
  for the $(n+1)$th variable of $f$.  More precisely, $f[e]=μ∘Tuₑ(f)$, where
  $uₑ∶n+1 → Tn$ is $[η,e]$.
\end{exa}

Let us now turn to defining algebras for a monoidal equational system.
\begin{defi}
  For any monoidal equational system $𝔼 = (𝐂∶ Γ ⊢ L =_⊗ R : Σ)$, an
  \alert{$𝔼$-algebra} is a $Σ$-monoid $X$ whose $Σ^⊛(X)$-algebra structure
  coequalises $L_X,R_X∶ Γ\pmod(X) ⇉ (Σ^⊛)\pmodmon(X)$. An $𝔼$-algebra
  morphism is a morphism between underlying $Σ^⊛$-algebras (a.k.a.\
  $Σ$-monoids).  We let $U_𝔼∶ 𝔼\alg → 𝐌𝐨𝐧(𝐂)$ denote the forgetful
  functor.
\end{defi}

Mimicking Proposition~\ref{prop:alteqsysalg}, we now transfer this
definition across the various correspondences of
Proposition~\ref{prop:mon:alteqsys}.
\begin{prop}
  For any monoidal equational system $𝔼 = (𝐂∶ Γ ⊢ L =_⊗ R : Σ)$ and 
  $Σ^⊛$-algebra $a∶ Σ^⊛(X) → X$, the following are equivalent:
  \begin{enumerata}
  \item \label{item:mon:alg:GS} $a$ coequalises $L_X,R_X∶ Γ\pmod(X) → (Σ^⊛)\pmodmon(X)$,
  \item \label{item:mon:alg:monads} $a$ coequalises the corresponding morphisms $Γ^⊛(X) → Σ^⊛(X)$,
  \item \label{item:mon:alg:monadic} the induced $Σ$-monoid structure
    $Σ\pmod(X) → X$ belongs to the equaliser of the corresponding functors
    $Σ\Mon → Γ\Mon$,
  \item \label{item:mon:alg:monadicfurther} the induced $Σ$-monoid structure
    $Σ\pmod(X) → X$ belongs to the equaliser of the induced functors
    $Σ\Mon → Γ\alg$,
  \item \label{item:mon:alg:reflmonads} $a$ coequalises the corresponding
    morphisms $(Σ+Γ)^⊛(X) → Σ^⊛(X)$,
  \item \label{item:mon:alg:reflmonadic} the induced $Σ\pmod$-algebra structure
    $Σ\pmod(X) → X$ belongs to the equaliser of the corresponding functors
    $Σ\Mon → (Σ+Γ)\Mon$,
  \item \label{item:mon:alg:ded} the corresponding natural transformations
    $Γ\pmod ∘ U_Σ → (-)\pmodmon ∘ U_Σ$ have the same components at the induced
    $Σ$-algebra structure $Σ\pmod(X) → X$.
  \end{enumerata}
\end{prop}
\begin{proof} \ \hfill
  \begin{itemize}
  \item \ref{item:mon:alg:GS} $⇔$ \ref{item:mon:alg:monadic} The
    corresponding functors $Σ\Mon → Γ\Mon$ map the induced $Σ\pmod$-algebra
    to $Γ\pmod(X) \xto{K} (Σ^⊛)\pmodmon(X) \xto{a} X$, for $K = L,R$, so both sides
    unfold to the same thing.
  \item \ref{item:mon:alg:monadic} $⇔$ \ref{item:mon:alg:monadicfurther} is clear.
  \item \ref{item:mon:alg:monads} $⇒$ \ref{item:mon:alg:GS} Follows
    from the fact that precomposing the induced morphisms \[
      (Γ^⊛)\pmodmon → (Σ^⊛)\pmodmon
    \]
    by $Γ\pmod → (Γ^⊛)\pmodmon$ yields the original $L$ and $R$ by construction.
  \item \ref{item:mon:alg:monadicfurther} $⇒$
    \ref{item:mon:alg:monads} This holds
    because~\ref{item:mon:alg:monads} is equivalent to equality of
    induced $Γ^⊛$-algebra structures, which is further equivalent to
    equality of induced $Γ$-algebra structures.
  \end{itemize}
  The rest follows easily.
\end{proof}

\begin{defi}
  \label{def:moneqsysreg}
  For a given convenient monoidal category $𝐂$, we define the monadic
  register $𝐌𝐄𝐒_𝐂$, called the \alert{monoidal equational register},
  as follows.
  \begin{description}
  \item[Signatures] A signature is a monoidal equational system.
  \item[Semantics] The \semsig associated to a signature $𝔼 $ is the
    forgetful functor $𝔼\alg → 𝐌𝐨𝐧(𝐂)$.
  \end{description}
\end{defi}
\begin{proof}[{\rm \bf Validity proof:}\nopunct]
  By Theorem~\ref{thm:MES} below.
\end{proof}

\section{Computing initial algebras in the presence of
equations}\label{s:friendliness}
In this section, we establish the announced explicit descriptions of
initial algebras, thereby proving that the registers introduced
in~§\ref{ss:regmonads} and~§\ref{s:regfun} are valid.  For this, we in
passing also prove the validity of the registers
from~§\ref{sec:gen-signatures}, and establish useful explicit
descriptions of initial algebras for them too.  In order, we
characterise the underlying monad and initial algebra for suitable
signatures in our registers featuring equations, namely the registers
\begin{itemize}
\item $𝐄𝐒_𝐂$ of equational systems (Definition~\ref{def:eqsysreg}),
\item $𝐌𝐄𝐒_𝐂$ of monoidal equational systems
  (Definition~\ref{def:moneqsysreg}),
\item $𝐑𝐞𝐠𝐌𝐧𝐝_f(𝐒𝐞𝐭^ℙ)$ of equational modular signatures
  (Definition~\ref{def:modmnd}), and
\item $𝐑𝐞𝐠⁰[𝐒𝐞𝐭^ℙ,𝐒𝐞𝐭^𝕊]_f$ of equational facet-based signatures
  (Definition~\ref{def:elemfunc}).
\end{itemize}
But before doing this, in~§\ref{ss:friendliness}, we study a
well-known~\cite{algebraictheories} refinement of finitariness, which
we call \alert{friendliness}, and prove the main foundational result
about it. 
In~§\ref{ss:friendlyES}--\ref{ss:friendlyPModMnd}, we then exploit
this to characterise underlying monads and initial algebras for the
announced registers.

\subsection{Reflexive coequalisers of friendly monoids}\label{ss:friendliness}
Let us start from the announced result on monads
(Theorem~\ref{thm:modmnd}).  The monads generated by our register
$𝐑𝐞𝐠𝐌𝐧𝐝_f(𝐒𝐞𝐭^ℙ)$ with equations are, almost by definition,
coequalisers in $𝐌𝐧𝐝_f(𝐒𝐞𝐭^ℙ)$.  We have announced in
Theorem~\ref{thm:modmnd} that their underlying functors are
coequalisers in $[𝐒𝐞𝐭^ℙ,𝐒𝐞𝐭^ℙ]_f$.  Technically, the goal is thus to
delineate sufficient conditions for monad coequalisers
$T_𝐄 ⇉ T_𝐃 ↠ T'$ to be computed pointwise, in the sense that each
$T_𝐄(X) ⇉ T_𝐃(X) ↠ T'(X)$ is a coequaliser, and the monad structure is
entirely determined by the family $(T'(X))_{X ∈ 𝐒𝐞𝐭^ℙ}$
(see~\cite[§V.3]{MacLane:cwm}).  Roughly, this will work if both
monads $T_𝐄$ and $T_𝐃$ preserve reflexive coequalisers. Because
(finitary) monads are monoids in the category of (finitary)
endofunctors, we can in fact generalise the result to monoids in a
suitable category (see Proposition~\ref{prop:algebraicptstr} below).
Preservation of reflexive coequalisers plays a fundamental role in the
study of algebraic theories~\cite{algebraictheories}. As we will use
it a lot, let us give it a name.

\begin{defi}\ \hfill
  \begin{itemize}
  \item A \alert{reflexive pair} of morphisms is a pair $f,g∶ X → Y$
    of parallel morphisms, sharing a common section $s$, i.e.,
    $s∶ Y → X$ such that $fs = \id_Y = gs$.
  \item A \alert{reflexive coequaliser} is a coequaliser of a
    reflexive pair.
  \item A functor is \alert{friendly} when it is finitary and
    preserves reflexive coequalisers.
\end{itemize}
\end{defi}
\begin{rem}
  By~\cite[Theorem~7.7]{algebraictheories}, when the considered
  categories are cocomplete, this is equivalent to preserving all
  \alert{sifted} colimits.
\end{rem}

\begin{defi}
  An object $X$ of a monoidal category is \alert{$⊗$-friendly}
  (pronounced ``tensor-friendly'') when the functor $X ⊗ {-}$ is
  friendly.
\end{defi}

\begin{prop}\label{prop:fiendlyendos}
  For any locally finitely presentable category $𝐂$, a finitary
  endofunctor on $𝐂$ is friendly iff it is $⊗$-friendly as an object
  of $[𝐂,𝐂]_f$ (viewed as monoidal for the composition tensor
  product).
\end{prop}
\begin{proof}
  Colimits of functors are computed
  pointwise~\cite[§V.4]{MacLane:cwm}.
\end{proof}
\begin{nota}
  By the proposition, in all of our use cases, $⊗$-friendliness and
  friendliness are synonymous, hence we use the latter for simplicity.
\end{nota}

\begin{defi} A monoidal category is \alert{friendly} when it is
  convenient and all objects are friendly.
\end{defi}

\begin{exa}\label{ex:friendly}\hfill
  \begin{itemize}
  \item The composition monoidal structure on finitary endofunctors on
    any locally finitely presentable category is convenient, though not
    friendly in general. By Proposition~\ref{prop:fiendlyendos}, the
    friendly objects are precisely the endofunctors preserving
    reflexive coequalisers.
  \item On categories of the form $𝐒𝐞𝐭^ℙ$ for some set $ℙ$, though,
    by~\cite[Corollary~6.30]{algebraictheories}, all finitary
    endofunctors are friendly, hence $[𝐒𝐞𝐭^ℙ,𝐒𝐞𝐭^ℙ]_f$ is friendly.
\end{itemize}
\end{exa}

\begin{prop}
  \label{prop:sifted-colim-mon}
  For any convenient monoidal category $𝐂$, the forgetful functor
  \[𝐌𝐨𝐧(𝐂) → 𝐂 \]creates reflexive coequalisers of friendly objects.
  More concretely, given a reflexive pair $X ⇉ Y$ of monoid morphisms,
  if $X$ and $Y$ are friendly, then the coequaliser in $𝐂$ lifts
  uniquely to a cocone of monoids, which is again a coequaliser.
\end{prop}
\begin{proof}
  Monoids are the algebras of the ``free monoid'' monad, which we
  denote by $(-)^⋆$ in this proof. Furthermore, a forgetful functor
  from monad algebras always creates those colimits that the monad
  preserves (Proposition~\ref{prop:algprescolim}).  It thus suffices
  to show that $(-)^⋆$ preserves reflexive coequalisers of friendly
  objects.

  Let thus $X₁ ⇉ X₂ ↠ Z$ denote a reflexive coequaliser, with $X₁$ and
  $X₂$ friendly.

  Let us first show that $Z$ is again friendly.
  For this, we consider any reflexive coequaliser
  $A₁ ⇉ A₂ ↠ C$. Then, we have:
  \[
  \begin{array}{rcll}
    Z⊗C & = & (\colimᵢ Xᵢ) ⊗ (\colimⱼ Aⱼ) \\
        & ≅ & \colimᵢ (Xᵢ ⊗ (\colimⱼ Aⱼ)) & \mbox{($𝐂$ is convenient)} \\
        & ≅ & \colimᵢ \colimⱼ (Xᵢ ⊗ Aⱼ) & \mbox{(each $Xᵢ$ is
                                          friendly)}     \\
        & ≅ & \colimⱼ \colimᵢ (Xᵢ ⊗ Aⱼ) & \mbox{(by interchange of
                                          colimits)}     \\
        & ≅ & \colimⱼ  ((\colimᵢ Xᵢ) ⊗ Aⱼ) & \mbox{($𝐂$ is
                                             convenient)}     \\
        & = & \colimⱼ  (Z ⊗ Aⱼ)\rlap{,}
  \end{array}\]
  as desired.

  Furthermore, by Proposition~\ref{prop:sigmamonchar} (with $Σ = 0$),
  $Z^⋆$ is the initial algebra of the endofunctor $H_Z = I + Z⊗{-}$.
  Morevoer, $H_Z$ is the coequaliser of $H_{X₁} ⇉ H_{X₂}$, i.e.,
  $(I + X₁⊗{-}) ⇉ (I + X₂⊗{-})$. 

  We then prove by induction that $H_{X₁}ⁿ(0) ⇉ H_{X₂}ⁿ(0) ↠ H_Zⁿ(0)$
  is again a (reflexive) coequaliser, for all $n ∈ ℕ$.
  \begin{itemize}
  \item The base case is trivial.
  \item Assuming that  $H_{X₁}ⁿ(0) ⇉ H_{X₂}ⁿ(0) ↠ H_Zⁿ(0)$
    is a coequaliser, we consider the following diagram.
    \begin{center}
      \diag{%
        H_{X₁}(H_{X₁}ⁿ(0)) \& H_{X₁}(H_{X₂}ⁿ(0)) \& H_{X₁}(H_{𝐙}ⁿ(0)) \\
        H_{X₂}(H_{X₁}ⁿ(0)) \& H_{X₂}(H_{X₂}ⁿ(0)) \& H_{X₂}(H_{𝐙}ⁿ(0)) \\
        H_{𝐙}(H_{X₁}ⁿ(0)) \& H_{𝐙}(H_{X₂}ⁿ(0)) \& H_{𝐙}(H_{𝐙}ⁿ(0))
      }{%
        (m-1-1.5) edge[labela={}] (m-1-2.175) %
        (m-1-1.-5) edge[labela={}] (m-1-2.185) %
        (m-1-1.-70) edge[labell={}] (m-2-1.70) %
        (m-1-1.-110) edge[labell={}] (m-2-1.110) %
        (m-2-1.5) edge[labela={}] (m-2-2.175) %
        (m-2-1.-5) edge[labela={}] (m-2-2.185) %
        (m-1-2.-70) edge[labell={}] (m-2-2.70) %
        (m-1-2.-110) edge[labell={}] (m-2-2.110) %
        (m-1-2) edge[labela={}] (m-1-3) %
        (m-2-2) edge[labelb={}] (m-2-3) %
        (m-1-3.-70) edge[labell={}] (m-2-3.70) %
        (m-1-3.-110) edge[labell={}] (m-2-3.110) %
        (m-3-1.5) edge[labela={}] (m-3-2.175) %
        (m-3-1.-5) edge[labela={}] (m-3-2.185) %
        (m-2-1) edge[labell={}] (m-3-1) %
        (m-2-2) edge[labelr={}] (m-3-2) %
        (m-3-2) edge[labelb={}] (m-3-3) %
        (m-2-3) edge[labelr={}] (m-3-3) %
      }
    \end{center}
    By construction all columns are reflexive coequalisers, and by
    friendliness so are all rows.
    By~\cite[Lemma~8.4.2]{BarrWells:ttt}, the diagonal is thus again a
    (reflexive) coequaliser.
  \end{itemize}
  Finally, by interchange of colimits, we obtain that
  $X₁^⋆ ⇉ X₂^⋆ ↠ Z^⋆$ is also a coequaliser, as desired.
\end{proof}
\begin{cor}
\label{cor:equaliser-monads}
  Reflexive coequalisers of friendly monads on a finitely presentable category
  $𝐂$ are computed pointwise.

  Similarly, reflexive coequalisers of finitary monads on a category
  of the form $𝐒𝐞𝐭^ℙ$ for some set $ℙ$, are computed pointwise.
\end{cor}
\begin{proof}
  The first point follows directly from the proposition.  For the
  second, we additionally use the fact that $[𝐒𝐞𝐭^ℙ,𝐒𝐞𝐭^ℙ]_f$ is
  friendly (Example~\ref{ex:friendly}).
\end{proof}

\subsection{Initial algebras for equational systems}\label{ss:friendlyES}
In this section, we want to apply the previous corollary to
characterise the induced monad and initial algebra for equational
systems.  For this, we should show that the monads $T_𝐃$ and $T_𝐄$
underlying the relevant parallel pair $𝐃 ⇉ 𝐄$ in $𝐌𝐨𝐧𝐚𝐝𝐢𝐜_f/𝐂$ are
indeed friendly. What helps us here is that these monads are free on a
friendly endofunctor, as we now show. We again state this in the
abstract context of a convenient monoidal category.
\begin{prop}\label{prop:freemon:algebraic}
  In any convenient monoidal category, the free monoid on a friendly object
  is again friendly.
\end{prop}
\begin{proof}
  We in fact prove the more general result that if $X∈ 𝐂$ preserves
  $𝐃$-colimits for a given category $𝐃$, in the sense that $X⊗ -$
  preserves $𝐃$-colimits, then so does $X^⋆ ⊗ {-}$. By
  Corollary~\ref{cor:Ucreates:Tpreserves}, it is enough to show that
  the forgetful functor from the category of algebras for the monad
  $X^⋆⊗-$ creates $𝐃$-colimits.  But, by~\cite[Proposition
  23.2]{KellyUnified}, this category of algebras is isomorphic (over
  $𝐂$) to the category of algebras for the endofunctor $X⊗-$. Thus, we
  are left with showing that the forgetful functor from this latter
  category creates $𝐃$-colimits, which follows from the next lemma.
\end{proof}

The following is analogous to Proposition~\ref{prop:algprescolim},
with an endofunctor instead of a monad.
\begin{lem}
  \label{lem:falg-creates-colim}
  Let $F$ be an endofunctor on a category $𝐂$. Then, the forgetful functor
  $F\alg → 𝐂$ creates any colimit that $F$ preserves.
  More specifically, given a category $𝐃$ such that $F$ preserves colimits
  of all diagrams $J∶ 𝐃 → 𝐂$, then the forgetful functor $F\alg → 𝐂$ creates 
  colimits of all diagrams $J∶ 𝐃 → F\alg$.
\end{lem}
\begin{proof}
  Straightforward.
\end{proof}

\begin{cor}\label{cor:freemonad:algebraic}
  The free monad on a friendly endofunctor on a finitely presentable category
  is friendly.
\end{cor}

We now want to apply Corollary~\ref{cor:equaliser-monads} to prove a
first free+quotient explicit description of initial algebras for
equational systems.  The exact same technique will then be applied to
other registers in the following subsections, namely to monoidal
equational systems, equational modular signatures, and equational
facet-based signatures.

Before giving the explicit description, we need to introduce the
following.
\begin{defi}
  For any functorial term $K∶ Γ → Σ^*$, let $\tilde{K}$ denote the monad
  morphism \[(Σ+Γ)^* ≅ Σ^* + Γ^* \xto{[\id_{Σ^*},K']} Σ^*\rlap{,}\]
  where $K'∶ Γ^* → Σ^*$ denotes the monad morphism induced by $K$ by
  freeness of $Γ^*$.
\end{defi}
We may now state:
\begin{thm}\label{thm:ES}
Let  $𝔼 = (𝐂 ∶ Γ ⊢ L = R : Σ)$ be any equational system. Then:
\begin{enumerati}
\item The forgetful functor $𝔼\alg → 𝐂$ is finitary monadic.
\item \label{item:coeqtilde} The finitary monad $𝔼^*$ underlying the forgetful functor
  $𝔼\alg → 𝐂$ is the coequaliser (in $𝐌𝐧𝐝_f(𝐂)$) of
  $\tilde{L},\tilde{R}∶ (Σ+Γ)^* ⇉ Σ^*$.
\end{enumerati}
Furthermore, if $Σ$ and $Γ$ are friendly (which is in particular the
case when $𝐂$ is $[𝐒𝐞𝐭^ℙ,𝐒𝐞𝐭^ℙ]_f$ for some set $ℙ$), we have:
\begin{enumerati}[resume]
\item 
  The above coequaliser $𝔼^*$ is created by the forgetful functor
  $𝐌𝐧𝐝_f(𝐂) → [𝐂,𝐂]_f$, hence computed pointwise.
\item 
  The initial $𝔼$-algebra is the
  coequaliser of
  \[\tilde{L}₀,\tilde{R}₀∶ (Σ+Γ)^*(0) ⇉ Σ^*(0)\rlap{,}\]
  with its canonical $Σ$-algebra structure.
\end{enumerati}
\end{thm}
\begin{proof}
  We start by expressing the category $𝔼\alg$ as an equaliser in $𝐂𝐀𝐓$.
  
  For any functorial term $K∶ Γ → Σ^*$, let $K'∶ Σ\alg → Γ\alg$ map any
  $Σ$-algbra $a∶ Σ(X) → X$ to the composite
  \[Γ(X) \xto{K_X} Σ^*(X) \xto{\tilde{a}} X\rlap{,}\]
  where $\tilde{a}$ denotes the induced $Σ^*$-algebra structure on
  $X$.

  The category $𝔼\alg$ is then the equaliser of the (generally non-reflexive) pair
  below left,
  \begin{center}
    \diag{%
      Σ\alg \& Γ\alg
    }{%
      (m-1-1.north east) edge[bend left=20,labela={L'}] (m-1-2.north west) %
      (m-1-1.-10) edge[bend right=20,labelb={R'}] (m-1-2.190) %
    }
    \hfil
    \diag{%
      Σ\alg \& (Σ+Γ)\alg
    }{%
      (m-1-1.north east) edge[bend left=20,labela={{\tilde{L}}\alg}] (m-1-2.north west) %
      (m-1-1.-10) edge[bend right=20,labelb={{\tilde{R}}\alg}] (m-1-2.190) %
    }
  \end{center}
  which already entails monadicity by
  Corollary~\ref{cor:monadicequaliser}.  But it is also an equaliser
  of the reflexive pair above right, which
  entails~\ref{item:coeqtilde} by Corollary~\ref{cor:street}.  The
  rest then follows from Corollary~\ref{cor:equaliser-monads} using
  Corollary~\ref{cor:freemonad:algebraic}.
\end{proof}
\begin{exa}
  Let us recall Example~\ref{ex:assocES}, where we introduced an
  equational system whose algebras are sets equipped with an
  associative binary operation.  Because we know how to compute
  coequalisers in sets, the theorem says that the free algebra on any
  $X ∈ 𝐒𝐞𝐭$ is obtained by quotienting the free $Σ$-algebra $Σ^*(X)$ by
  following relation $∼$.  We first construct the free $(Σ+Γ)$-algebra
  $(Σ+Γ)^*(X)$ on $X$, obtained by freely adding a binary operation
  and a ternary operation, say $f$, to $X$.  We then define two maps
  $L,R∶ (Σ+Γ)^*(X) → Σ^*(X)$, by recursively interpreting $f(x,y,z)$
  as $(x+y)+z$, resp.\ $x + (y+z)$.  We finally define $∼$ to be
  the smallest equivalence relation such that $x∼y$ whenever $x = L(z)$ and
  $y = R(z)$ for some $z ∈ (Σ+Γ)^*(X)$.
\end{exa}

\subsection{Initial algebras for monoidal equational systems}\label{ss:friendlyMES}
In this section, we characterise the induced monad and initial
algebra of monoidal equational systems, under mild additional
hypotheses.

\begin{thm}\label{thm:MES}
Let  $𝔼 = (𝐂 ∶ Γ ⊢ L =_⊗ R : Σ)$ be any monoidal equational system. Then:
\begin{enumerati}
\item \label{item:meqs:finitary:monadic} The forgetful functor
  $𝔼\alg → 𝐌𝐨𝐧(𝐂)$ is finitary monadic.
\item \label{item:meqs:rcoeq} The finitary monad $𝔼^⋆$ induced by the forgetful functor
  $𝔼\alg → 𝐂$ is the coequaliser of $L,R∶ (Σ+Γ)^⋆ ⇉ Σ^⋆$ in
  $𝐌𝐧𝐝_f(𝐂)$, where we recall that, for any finitary, pointed strong
  endofunctor $F$, $F^⋆$ denotes the ``free $F$-monoid'' monad on $𝐂$.
\end{enumerati}
If $Σ$ and $Γ$ are friendly, which is in particular the case when $𝐂$
is $[𝐒𝐞𝐭^ℙ,𝐒𝐞𝐭^ℙ]_f$ for some set $ℙ$, then:
\begin{enumerati}[resume]
\item The above coequaliser $𝔼^⋆$ is created by the forgetful functor
  $𝐌𝐧𝐝_f(𝐂) → [𝐂,𝐂]_f$, hence computed pointwise.
\item The initial $𝔼$-algebra is the coequaliser of
  $L₀,R₀∶ (Σ+Γ)^*(\id) ⇉ Σ^*(\id)$, equipped with its canonical
  $Σ$-monoid structure.
\end{enumerati}
\end{thm}
\begin{proof}
  By Theorem~\ref{thm:sigmamonadic}, both forgetful functors
  $Σ\Mon → 𝐂$ and $Γ\Mon → 𝐂$ are finitary
  monadic. Corollaries~\ref{cor:street} and~\ref{cor:monadicequaliser}
  thus directly entail~\ref{item:meqs:finitary:monadic}
  and~\ref{item:meqs:rcoeq}. The rest will follow from
  Corollary~\ref{cor:equaliser-monads} if we prove that both monads
  $(Σ+Γ)^⋆$ and $Σ^⋆$ are friendly. This is proved in the next lemma.
\end{proof}

\begin{prop}\label{prop:algebraicptstr}
  For any finitary, pointed strong endofunctor $F$ on a friendly
  monoidal category $𝐂$, if $F$ is friendly, then so is the ``free
  $F$-monoid'' monad $F^⋆$.
\end{prop}
\begin{proof}
  This is a direct generalisation of the proof of
  Proposition~\ref{prop:sifted-colim-mon} (which corresponds to the case
  $F = 0$), using the fact that the free $F$-monoid monad maps $X$ to
  the initial algebra of $A ↦ I + X ⊗ A + F(A)$, as recalled in
  Proposition~\ref{prop:sigmamonchar}.
\end{proof}

\subsection{Initial algebras for \texorpdfstring{$𝐑𝐞𝐠𝐌𝐧𝐝_f(𝐒𝐞𝐭^ℙ)$}{RegMndfSetP}}
\label{ss:friendlyModMnd}
In this section, we characterise the underlying monad and initial
algebra of equational modular signatures, i.e., signatures of the
register $𝐑𝐞𝐠𝐌𝐧𝐝_f(𝐒𝐞𝐭^ℙ)$, filling in the missing bits
from~§\ref{ss:ia:mnd}.  We first deal with the register
$𝐑𝐞𝐠𝐌𝐧𝐝⁰_f(𝐒𝐞𝐭^ℙ)$ without equations, by compiling (in the sense of
Definition~\ref{def:compilation}) to the register
$𝐏𝐒𝐄𝐅_{[𝐒𝐞𝐭^ℙ,𝐒𝐞𝐭^ℙ]_f}$ of pointed strong endofunctors on
$[𝐒𝐞𝐭^ℙ,𝐒𝐞𝐭^ℙ]_f$. We then tackle the whole register
$𝐑𝐞𝐠𝐌𝐧𝐝_f(𝐒𝐞𝐭^ℙ)$, using friendliness.

Let us first deal with the case without equations.  Recalling from
Definition~\ref{def:SigmaS} the endofunctor $Σ_S$ on $[𝐒𝐞𝐭^ℙ,𝐒𝐞𝐭^ℙ]_f$
induced by a modular signature $S$, the idea in this case is that the
assignment $S ↦ Σ_S$ may be viewed as mapping signatures of
$𝐑𝐞𝐠𝐌𝐧𝐝⁰_f(𝐒𝐞𝐭^ℙ)$ to signatures of $𝐏𝐒𝐄𝐅_{[𝐒𝐞𝐭^ℙ,𝐒𝐞𝐭^ℙ]_f}$, i.e.,
pointed strong endofunctors.  We first establish this by equipping
$Σ_S$ with a pointed strength, then use compilation to transport the
problem, and conclude.
\begin{prop}
  For any modular signature $S$, the endofunctor $Σ_S$ admits a
  pointed strength given at any $P ∈ [𝐒𝐞𝐭^ℙ,𝐒𝐞𝐭^ℙ]_f$,
  $Q ∈ \id/[𝐒𝐞𝐭^ℙ,𝐒𝐞𝐭^ℙ]_f$, $X ∈ 𝐒𝐞𝐭^ℙ$, and operation of arity
  $(d,p)$, say with $d = (F ∘ Θ)^{(p₁,…,pₙ)}$, by applying
  $F(P(-))·𝐲ᵣ$ to the obvious morphism
  $Q(X) + ∑ᵢ 𝐲_{pᵢ} → Q (X + ∑ᵢ 𝐲_{pᵢ} ).$
\end{prop}
\begin{proof}
  Straightforward.
\end{proof}
\begin{rem}
  In a bit more detail, letting $S = (dᵢ,pᵢ)_{i ∈ I}$,
  we have $Σ_S = ∑_{i ∈ I} Σ_{dᵢ,pᵢ}$.
  The pointed strength is defined as the coproduct (of morphisms)
  \[(∑_{i ∈ I} Σ_{dᵢ,pᵢ})(P) ∘ Q = ∑_{i ∈ I} (Σ_{dᵢ,pᵢ}(P) ∘ Q)
  \xto{∑_{i ∈ I} stⁱ_{P,Q}} ∑_{i ∈ I} Σ_{dᵢ,pᵢ}(P ∘ Q)\rlap{,} \]where
  for all $i ∈ I$, say with $dᵢ ≅ (F ∘ Θ)^{(qⁱ₁,…,qⁱ_{nᵢ})}$ and $pᵢ ∈ ℙ$,
  $stⁱ_{P,Q,X}$ is the obvious morphism
  \[F (P (Q (X) + ∑_{j ∈ nᵢ} 𝐲_{qⁱⱼ})) · 𝐲_{pᵢ} →
  F (P (Q (X + ∑_{j ∈ nᵢ} 𝐲_{qⁱⱼ}))) · 𝐲_{pᵢ}.\]
\end{rem}

  \begin{prop}\label{prop:SigmaSMon:Salg}
    The assignment $S ↦ Σ_S$ defines a compilation
    \[𝐑𝐞𝐠𝐌𝐧𝐝⁰_f(𝐒𝐞𝐭^ℙ) → 𝐏𝐒𝐄𝐅_{[𝐒𝐞𝐭^ℙ,𝐒𝐞𝐭^ℙ]_f}.\]
    More concretely, for any modular signature $S = (dᵢ,rᵢ)_{i ∈ I}$,
    there exists an isomorphism \[S\alg ≅ Σ_S\Mon \]of categories over
    $𝐌𝐧𝐝_f(𝐒𝐞𝐭^ℙ)$.
  \end{prop}
  \begin{proof}
Any family of $𝐒𝐞𝐭$-valued module morphisms $ρᵢ∶ dᵢ(X) → X_{rᵢ}$
corresponds by the adjunction of Proposition~\ref{prop:adjyev} to a
family $\tilde{ρᵢ}∶ dᵢ(X) · 𝐲_{rᵢ} → X$ of $𝐒𝐞𝐭^ℙ$-valued module
morphisms, hence by copairing to one compatible module morphism
$Σ_S\pmod(X) → X$ (recalling Definition~\ref{def:sigmas}), and thus to
$Σ_S$-monoid structure on $X$.  This correspondence extends
straightforwardly to morphisms, hence defining the desired functor
$S\alg → Σ_S\Mon$ over $𝐌𝐧𝐝_f(𝐒𝐞𝐭^ℙ)$.  Since it is bijective, the
functor is an isomorphism.
  \end{proof}

  \begin{cor}\label{cor:SModMnd}
    For any modular signature $S$, the forgetful functor
    $S\alg → 𝐌𝐧𝐝_f(𝐒𝐞𝐭^ℙ)$ is monadic, and furthermore the free
    $S$-algebra on an endofunctor $X ∈ [𝐒𝐞𝐭^ℙ,𝐒𝐞𝐭^ℙ]_f$ is the free
    $Σ_S$-monoid $Σ_S^⋆(X)$, as characterised in
    Proposition~\ref{prop:sigmamonchar}.

    In particular, $Σ_S^*(\id)$ has a canonical $S$-algebra structure,
    which is initial.
  \end{cor}

  We now want to characterise the initial $E$-algebra, for any
  equational modular signature $E = (S,S',L,R)$, 
  where we recall that $L$ and $R$ are functors $S\alg → S'\alg$
  over $𝐌𝐧𝐝_f(𝐒𝐞𝐭^ℙ)$.
Clearly:
\begin{prop}
  For any equational modular signature $E= (S,S',L,R)$, $E\alg$ is
  the equaliser in $𝐂𝐀𝐓$ of $L$ and $R$.
\end{prop}

By Proposition~\ref{prop:SigmaSMon:Salg}, $L$ and
$R$ induce functors $L',R'∶ Σ_S\Mon → Σ_{S'}\Mon$ over
$𝐌𝐧𝐝_f(𝐒𝐞𝐭^ℙ)$ making the following square commute serially.
\begin{center}
  \diag{%
    S\alg \& S'\alg \\
    Σ_S\Mon \& Σ_{S'}\Mon
  }{%
    (m-1-1) edge[labell={≅}] (m-2-1)
    (m-1-2) edge[labelr={≅}] (m-2-2)
    (m-1-1.10) edge[bend left=10,labela={L}] (m-1-2.170) %
    (m-1-1.0) edge[bend right=10,labelb={R}] (m-1-2.180) %
    (m-2-1.10) edge[bend left=10,labela={L'}] (m-2-2.170) %
    (m-2-1.0) edge[bend right=10,labelb={R'}] (m-2-2.180) %
  }
\end{center}
By Corollary~\ref{cor:street}, $L'$ and $R'$ induce a
reflexive parallel pair of monad morphisms
\[L'',R''∶ (Σ_S+Σ_{S'})^⋆ ⇉ Σ_S^⋆\rlap{.}\]
The tuple $(Σ_S,Σ_{S+S'},L'',R'')$ forms a
monoidal equational system $𝔼_E$, over the category $[𝐒𝐞𝐭^ℙ,𝐒𝐞𝐭^ℙ]_f$, whose
algebras are by construction isomorphic to $E$-algebras.  This
readily entails:
\begin{cor}
  The assignment
  \[E= (S,S',L,R) ↦ 𝔼_E ≔ (Σ_{S'} ⊢ L'' =_⊗ R'' : Σ_S)\]
  induces a compilation
  \[𝐑𝐞𝐠𝐌𝐧𝐝_f(𝐒𝐞𝐭^ℙ) → 𝐌𝐄𝐒_{[𝐒𝐞𝐭^ℙ,𝐒𝐞𝐭^ℙ]_f}.\]
  More concretely, for any equational modular signature $E$, we have
  an isomorphism
  \[E\alg ≅ 𝔼_E\Mon\]
  of categories over $𝐌𝐧𝐝_f(𝐒𝐞𝐭^ℙ)$.
\end{cor}

As a result, we readily obtain by Theorem~\ref{thm:MES}:
\begin{cor}\label{cor:modmnd}
  Let $E = (S,S',L,R)$ be any equational modular signature. Then:
  \begin{enumerati}
  \item \label{item:ems:monadic} The forgetful functor
    $E\alg → 𝐌𝐧𝐝_f(𝐒𝐞𝐭^ℙ)$ is finitary monadic.
  \item The finitary monad $E^⋆$ underlying the forgetful functor
    \[E\alg → [𝐒𝐞𝐭^ℙ,𝐒𝐞𝐭^ℙ]_f \]is the reflexive coequaliser in
    $𝐌𝐧𝐝_f([𝐒𝐞𝐭^ℙ,𝐒𝐞𝐭^ℙ]_f)$ of $L'',R''∶ (Σ_S+Σ_{S'})^⋆ ⇉ Σ_S^⋆$.
  \item The above coequaliser is created by the forgetful functor
    \[𝐌𝐧𝐝_f([𝐒𝐞𝐭^ℙ,𝐒𝐞𝐭^ℙ]_f) → [[𝐒𝐞𝐭^ℙ,𝐒𝐞𝐭^ℙ]_f,[𝐒𝐞𝐭^ℙ,𝐒𝐞𝐭^ℙ]_f]_f\rlap{,}\]
  hence computed pointwise.
\item The initial $E$-algebra is the coequaliser of
  $(L'')₀,(R'')₀∶ (Σ_S+Σ_{S'})^*(\id) ⇉ Σ_S^*(\id)$, equipped with
  its canonical $S$-algebra structure.
\end{enumerati}
\end{cor}

\subsection{Initial algebras for \texorpdfstring{$𝐑𝐞𝐠[𝐒𝐞𝐭^ℙ,𝐒𝐞𝐭^𝕊]_f$}{𝐑𝐞𝐠[SetP,SetS]f}}
\label{ss:friendlyPModMnd}
In this section, we prove monadicity of the register $𝐑𝐞𝐠[𝐒𝐞𝐭^ℙ,𝐒𝐞𝐭^𝕊]_f$ of
equational facet-based signatures, and characterise underlying monads
and initial algebras.  For this, we proceed essentially as in the
previous section, but more simply since the intricacies related to
$Σ$-monoids do not arise. We are thus able to compile
\begin{itemize}
\item $𝐑𝐞𝐠⁰[𝐒𝐞𝐭^ℙ,𝐒𝐞𝐭^𝕊]_f$ to the register  $𝐄𝐅_{[𝐒𝐞𝐭^ℙ,𝐒𝐞𝐭^𝕊]_f}$ of
  endofunctors on $[𝐒𝐞𝐭^ℙ,𝐒𝐞𝐭^𝕊]_f$, and 
\item $𝐑𝐞𝐠[𝐒𝐞𝐭^ℙ,𝐒𝐞𝐭^𝕊]_f$ to the register $𝐄𝐒_{[𝐒𝐞𝐭^ℙ,𝐒𝐞𝐭^𝕊]_f}$ of
  equational systems over $[𝐒𝐞𝐭^ℙ,𝐒𝐞𝐭^𝕊]_f$.
\end{itemize}

Recalling from Definition~\ref{def:SigmaSpre} the endofunctor
$Σ_S$ associated to any facet-based signature $S$, we have:
  \begin{prop}\label{prop:SigmaSMon:Salgpre}
    The assignment $S ↦ Σ_S$ defines a compilation
    \[𝐑𝐞𝐠⁰[𝐒𝐞𝐭^ℙ,𝐒𝐞𝐭^𝕊]_f → 𝐄𝐅_{[𝐒𝐞𝐭^ℙ,𝐒𝐞𝐭^𝕊]_f}.\]
    More concretely, for any facet-based signature
    $S = (dᵢ,sᵢ)_{i ∈ I}$, there exists an isomorphism
    \[S\alg ≅ Σ_S\alg \]of categories over $[𝐒𝐞𝐭^ℙ,𝐒𝐞𝐭^𝕊]_f$.
  \end{prop}
  \begin{proof}
    As in the proof of Proposition~\ref{prop:SigmaSMon:Salg},
    by the adjunction 
    \begin{center}
      \adj{[𝐒𝐞𝐭^ℙ,𝐒𝐞𝐭^𝕊]_f}{[𝐒𝐞𝐭^ℙ,𝐒𝐞𝐭]_f\rlap{,}}{(-)·𝐲ₛ}{(-)ₛ}
    \end{center}
    the facet morphisms
    $ρᵢ∶ dᵢ(X) → X_{sᵢ}$ correspond to natural transformations
    $\tilde{ρᵢ}∶ dᵢ(X) · 𝐲_{sᵢ} → X$, hence by copairing to
    $Σ_S$-algebra structure $Σ_S(X) → X$.  This correspondence extends
    straightforwardly to morphisms, hence defining the desired functor
    $S\alg → Σ_S\alg$ over $[𝐒𝐞𝐭^ℙ,𝐒𝐞𝐭^𝕊]_f$.  Since it is bijective,
    the functor is an isomorphism.
  \end{proof}

  We now want to show that, for any equational facet-based signature
  $E$, the forgetful functor $E\alg → [𝐒𝐞𝐭^ℙ,𝐒𝐞𝐭^𝕊]_f$ is monadic, and
  to explicitly characterise the corresponding monad and initial
  algebra.  For this, we can exhibit $E\alg$ as an equaliser of
  finitary monadic functors over $[𝐒𝐞𝐭^ℙ,𝐒𝐞𝐭^𝕊]_f$ and apply
  Corollary~\ref{cor:monadicequaliser}:
  \begin{prop}
    For any equational facet-based signature $E= (S,S',L,R)$, $E\alg$
    is the equaliser in $𝐂𝐀𝐓$ of $L$ and $R$.
  \end{prop}
  \begin{proof}
    Straightforward.
  \end{proof}

  But in fact, perhaps more conveniently, we can also compile into
  equational systems.  By Proposition~\ref{prop:SigmaSMon:Salgpre},
  $L$ and $R$ induce functors $L',R'∶ Σ_S\alg → Σ_{S'}\alg$ over
  $[𝐒𝐞𝐭^ℙ,𝐒𝐞𝐭^𝕊]_f$ making the following square commute serially.
\begin{center}
  \diag{%
    S\alg \& S'\alg \\
    Σ_S\alg \& Σ_{S'}\alg
  }{%
    (m-1-1) edge[labell={≅}] (m-2-1)
    (m-1-2) edge[labelr={≅}] (m-2-2)
    (m-1-1.10) edge[bend left=10,labela={L}] (m-1-2.170) %
    (m-1-1.0) edge[bend right=10,labelb={R}] (m-1-2.180) %
    (m-2-1.10) edge[bend left=10,labela={L'}] (m-2-2.170) %
    (m-2-1.0) edge[bend right=10,labelb={R'}] (m-2-2.180) %
  }
\end{center}
This readily entails:
\begin{cor}
  The assignment
  \[E= (S,S',L,R) ↦ 𝔼_E ≔ (Σ_{S'} ⊢ L' = R' : Σ_S)\]
  induces a compilation
  \[𝐑𝐞𝐠[𝐒𝐞𝐭^ℙ,𝐒𝐞𝐭^𝕊]_f → 𝐄𝐒_{[𝐒𝐞𝐭^ℙ,𝐒𝐞𝐭^𝕊]_f}.\]
  More concretely, for any equational facet-based signature
  $E= (S,S',L,R)$, we have isomorphisms
  \[E\alg ≅ 𝔼_E\alg\]
  of categories over $[𝐒𝐞𝐭^ℙ,𝐒𝐞𝐭^𝕊]_f$, and
  \[E^⋆ ≅ 𝔼_E^⋆\]
  of finitary monads thereupon.
\end{cor}

As a result, we readily obtain by Theorem~\ref{thm:ES}:
\begin{cor}\label{cor:PSEmonadic}
  Let $E= (S,S',L,R)$ be any equational facet-based signature. Then:
  \begin{enumerati}
  \item The forgetful functor $E\alg → [𝐒𝐞𝐭^ℙ,𝐒𝐞𝐭^𝕊]_f$ is finitary
    monadic.
  \item The finitary monad $E^⋆$ which underlies the forgetful functor
    $E\alg → [𝐒𝐞𝐭^ℙ,𝐒𝐞𝐭^𝕊]_f$ is the reflexive coequaliser in
    $𝐌𝐧𝐝_f([𝐒𝐞𝐭^ℙ,𝐒𝐞𝐭^𝕊]_f)$ of
    $L',R'∶ (Σ_S+Σ_{S'})^* ⇉ Σ_S^*$.
  \item The above coequaliser is created by the forgetful functor
    \[𝐌𝐧𝐝_f([𝐒𝐞𝐭^ℙ,𝐒𝐞𝐭^𝕊]_f) → [[𝐒𝐞𝐭^ℙ,𝐒𝐞𝐭^𝕊]_f,[𝐒𝐞𝐭^ℙ,𝐒𝐞𝐭^𝕊]_f]_f\rlap{,}\]
    hence computed pointwise.
  \item The initial $E$-algebra is the coequaliser of
    $(L')₀,(R')₀∶ (Σ_S+Σ_{S'})^*(0) ⇉ Σ_S^*(0)$, equipped with its
    canonical $S$-algebra structure.
\end{enumerati}
\end{cor}
  
\section{Applications}\label{s:applications}

In this section, we design a signature in $𝐑𝐞𝐠𝐓𝐫𝐚𝐧𝐬𝐌𝐧𝐝_{ℙ,𝕊}$ for each
of the announced languages.  One exception is Positive GSOS systems:
for them, we go further, by recasting them as the signatures of a
subregister of $𝐑𝐞𝐠𝐓𝐫𝐚𝐧𝐬𝐌𝐧𝐝_{ℙ,𝕊}$.

  \subsection{The call-by-value, simply-typed, big-step \lam-calculus}\label{ss:cbvspec}
  Recall from~§\ref{ss:stlambda}, that the simply-typed,
  call-by-value, big-step $λ$-calculus forms \atransition monad, where
  we take $ℙ = 𝕊$ to be the set of types (generated from some fixed
  set of type constants).  The monad $T$ over $𝐒𝐞𝐭^ℙ$ is given by
  values, the source state functor $S₁$ is given by application binary
  trees, and the second state functor $S₂$ is the identity.

  Let us now design a signature for this \transition monad.  Let
  $F∶ 𝐒𝐞𝐭^ℙ → 𝐒𝐞𝐭^𝕊$ be specified by the signature of $𝐑𝐞𝐠[𝐒𝐞𝐭^ℙ,𝐒𝐞𝐭^𝕊]_f$
  consisting of two families of operations
  $\app_{A,B}∶ Θ_{A→B} × Θ_A → Θ_{B}$ and $\val_A∶ 𝐈_A → Θ_A$.  Our
  signature for call-by-value, simply-typed, big-step $λ$-calculus is
  presented in the following table \vspace{-1em}
\begin{center}
$\begin{array}[t]{|p{.18\linewidth}|c|c|c|}
   \hline
   \multirow{2}{2.3cm}{Monad and state functors} & \Tstrutpetit \Bstrutbas T & S₁ & S₂  \\ \cline{2-4}
   & \Tstrut \Bstrut
   λ_{A,B}: (F_{B}∘Θ)^{(A)} → Θ_{A→B} & F & \Id \\ \hline\hline
   Rules & \multicolumn{3}{|c|}{   \Tstruthaut
   \inferrule{ }{\val_A(v) \leadsto v}
  \qquad \qquad
  \inferrule{ e₁ \leadsto λ_{A, B}(e₃) \\ e₂ \leadsto w \\ e₃[w] \leadsto v }{ \app_{A,B}(e₁, e₂) \leadsto v }} \\ \hline
  \end{array}$
\end{center}
where 
\begin{itemize}
\item $-[-]∶ (S₁T)^{(A)}_{B}×T_A → (S₁T)_{B}$ denotes the substitution morphism
  \[\begin{array}{rc}
S₁(T(X+𝐲_A))_B × T(X)_A ≅ & ∑_{e' ∈ T(X)_A} S₁(T(X+𝐲_A))_B \\
& ↓ \rlap{$\scriptstyle [S₁(T[ηᵀ_X,\tilde{e}'])_B]_{e ∈ T(X)_A}$} \\
  & S₁(T (T (X)))_B \\
  & ↓ \rlap{$\scriptstyle S₁(μᵀ_X)_B$} \\
  & S₁ (T (X))_B\rlap{,}
  \end{array}\]
  with $\tilde{e}'∶ 𝐲_A → T(X)$ corresponding to $e' ∈ T(X)_A$ by Yoneda;
\item $S₁=F$ and $S₂=\Id$ are specified by easy signatures as in
  Remark~\ref{rk:statefuns};
\item the rules should be understood as families of rules indexed by
  suitable types.
\end{itemize}
In a bit more detail, the first rule is indexed by the type $A$ of
$v$.  The second one is indexed by two types $A$ and $B$.  There are
five metavariables, $e₁$, $e₂$, $e₃$, $v$, and $w$.  We thus take
$V ≔ (S₁T)_{A→B} × (S₁T)_A × (S₁T)^{(A)}_{B} × T_{B} × T_A.$

\subsection{The \lbmu-calculus}\label{ss:lambdamu}
Recall from~§\ref{ss:lbmu} that the
$\overline{λ}μ$-calculus~\cite{DBLP:phd/hal/Herbelin95,Iouri} forms
\atransition monad with $ℙ = 2 = \ens{𝐩,𝐬}$, where
$𝐩$ stands for ``programs'' and $𝐬$ for ``stacks''.
The placetaker monad $T$ on
$𝐒𝐞𝐭^ℙ$ is given by programs and stacks.
The set of transition types is $𝕊 = 3 = \ens{𝐜,𝐩,𝐬}$, where $𝐜$ stands for
``commands'', and both state functors
$S₁,S₂∶ 𝐒𝐞𝐭^ℙ → 𝐒𝐞𝐭^𝕊$ are given by 
$S₁(A) = S₂(A) = (A_𝐩 × A_𝐬, A_𝐩, A_𝐬)$.
Let us repeat the grammar for the reader's convenience.
\begin{center}
  \vspace{-0.25em}
  $\begin{array}[t]{c}
    \mbox{Commands} \\
    c  \Coloneqq  ⟨e|π⟩
  \end{array}$
\hfil
  $\begin{array}[t]{c}
    \mbox{Programs} \\
    e  \Coloneqq  x ｜ μα.c ｜ λx.e
  \end{array}$
    \hfil
    $\begin{array}[t]{c}
    \mbox{Stacks} \\
    π  \Coloneqq  α ｜ e⋅π
      \end{array}$
\end{center}

Transitions are generated by the congruence rules and the following two rules
\begin{center}
  ${⟨μα.c|π⟩ → c[α↦π]}$ \hfil ${⟨λx.e|e'⋅π⟩ → ⟨e[x↦e']|π⟩}$.
\end{center}

Let us see how to specify this \transition monad using our register.

We saw in Remark~\ref{rk:statefuns} that both state functors may be specified by
operations
\begin{mathpar}
     ⟨-|-⟩∶  𝐈_𝐩 × 𝐈_𝐬  → Θ_𝐜 \and
     η_𝐩 ∶  𝐈_𝐩  → Θ_𝐩 \and
     η_𝐬 ∶  𝐈_𝐬  → Θ_𝐬.
   \end{mathpar}

   The monad, say $T$, is specified by operations
\begin{mathpar}
  μ: Θ^{(𝐬)}_𝐩 × Θ^{(𝐬)}_𝐬 → Θ_𝐩 \and
  λ: Θ^{(𝐩)}_𝐩 → Θ_𝐩 \and
  {·}: Θ_𝐩 × Θ_𝐬 → Θ_𝐬,
\end{mathpar}
with no equation.
 
The basic transition rules are almost as usual:
\begin{mathpar}
  \inferrule{ }{⟨μ⟨e|π'⟩|π⟩ ↝ ⟨e[π],π'[π]⟩}
  \and
 \inferrule{ }{⟨λ(e)|e'⋅π⟩ ↝ ⟨e[e']|π⟩}~·
\end{mathpar}

\begin{rem}\hfill
  \begin{itemize}
  \item The various substitution morphisms ${-}[-]$ are defined
    analogously as in~§\ref{ss:cbvspec}.
  \item The first rule has metavariable module
    $V ≔ T^{(𝐬)}_𝐩×T^{(𝐬)}_𝐬×T_𝐬$, the argument being $(e,π',π)$.
  \item The second rule has $V ≔ T_𝐩^{(𝐩)}×T_𝐩×T_𝐬$.
\end{itemize}
\end{rem}

Let us finish by listing the various congruence rules.
\begin{center}
\begin{tabular}{rcc}
  Commands:
  &
    $ 
      \inferrule{e ↝ e' }{⟨e|π⟩ ↝ ⟨e'|π⟩}
      $ 
    &
      $
        \inferrule{π ↝ π'}{⟨e|π⟩ ↝ ⟨e|π'⟩}
      $ 
      \\
  \\
  Programs:
  &
    $
    \inferrule{⟨e|π⟩ ↝ ⟨e'|π'⟩ }{μ⟨e|π⟩ ↝ μ⟨e'|π'⟩}
    $
    &
      $
        \inferrule{e ↝ e' }{λ(e) ↝ λ(e')}
      $
      \\
  \\
      Stacks:
  &
    $
      \inferrule{e ↝ e' }{e·π ↝ e'·π}
    $
    &
      $
      \inferrule{π ↝ π' }{e·π ↝ e·π'}
      $
\end{tabular}
\end{center}

\subsection{The \pii-calculus}\label{ss:signature:pi}
Recall from~§\ref{ss:expi} that
the $π$-calculus~\cite{DBLP:books/daglib/0004377} also forms \atransition
monad with $ℙ = 2 = \ens{𝐜,𝐩}$ ($𝐜$ for ``channels'', $𝐩$ for
``processes'').  The placetaker monad is the identity on channels, and is
defined on processes by the  grammar
\[P,Q \Coloneqq x ｜ 0 ｜ (P|Q) ｜ νa.P ｜ \abar⟨b⟩.P ｜ a(b).P\rlap{,}\]
where $a$ and $b$ range over channel names, and $b$ is bound in
$a(b).P$ and in $νb.P$.  Processes are identified when related by the
smallest context-closed equivalence relation $≡$ satisfying
\begin{center}
 $0|P ≡ P$
\hfil $P|Q ≡ Q|P$ \hfil $P|(Q|R) ≡ (P|Q)|R$ \hfil
$(νa.P)|Q ≡ νa.(P|Q)$,
\end{center}
where in the last equation $a$ should not occur free in $Q$.
Transitions are then given by the following rules.
\begin{mathpar}
\inferrule{ }{\abar⟨b⟩.P | a(c).Q ⟶ P|(Q[c↦b])} \and
\inferrule{P ⟶ Q}{P|R ⟶ Q|R} \and \inferrule{P ⟶ Q}{νa.P ⟶ νa.Q} 
\end{mathpar}
Let us recall that we denote any object $X ∈ 𝐒𝐞𝐭^ℙ$ by
$X = (X_𝐜,X_𝐩)$, so that we have $T(X)=(X_𝐜, T(X)_𝐩) ∈
𝐒𝐞𝐭^ℙ$. Furthermore, transitions relate processes, so we have $𝕊 = 1$
and $S₁(X) = S₂(X) = X_𝐩$.

Let us see how to specify this \transition monad using our
register. The state functor has been specified in
Remark~\ref{rk:statefuns}, by a single operation $η∶ 𝐈_𝐩 → Θ$.  The
placetaker monad $T$ is specified by operations

\begin{center}
  $0∶ 1 → Θ_𝐩 \quad {|}∶ Θ_𝐩 × Θ_𝐩 → Θ_𝐩 \quad ν∶ Θ^{(𝐜)}_𝐩 → Θ_𝐩
  \quad out∶ Θ_𝐜² × Θ_𝐩 → Θ_𝐩 \quad in∶ Θ_𝐜 × Θ^{(𝐜)}_𝐩 → Θ_𝐩$
\end{center}
  with equations 
\[
0|P ≡ P \quad P|Q ≡ Q|P \quad P|(Q|R) ≡ (P|Q)|R \quad ν(P)|Q ≡ ν(P|𝐰_𝐜(Q)),
\]
 almost copied verbatim from the standard presentation above,
 where $𝐰_𝐜(Q)$ denotes the action of $T(X) → T(X + 𝐲_𝐜)$ on $Q$.
 Finally, the transition rules are as follows,
 \begin{mathpar}
     \inferrule{ }{out(a,b,P) | in(a,Q) ↝ P|(Q[b])}
    \and
    \inferrule{P ↝ Q}{P|R ↝ Q|R} \and
    \inferrule{P ↝ Q}{ν(P) ↝ ν(Q)}
  \end{mathpar}
  where $Q[b]$ denotes the action
  of
  \[T(X+𝐲_𝐜)_𝐩 × T(X)_𝐜 ≅ ∑_{b ∈ T(X)_𝐜} T(X+𝐲_𝐜)_𝐩 \xto{[T[ηᵀ_X,\tilde{b}]_𝐩]_{b ∈ T(X)_𝐜}}
  T (T(X))_𝐩 \xto{(μᵀ_X)_𝐩} T(X)_𝐩\] on $(Q,b)$, with
  $\tilde{b}∶ 𝐲_𝐜 → T(X)$ corresponding to $b ∈ T(X)_𝐜$ by Yoneda.
  
  \begin{rem}
    The third rule has as metavariable module $V ≔ (Θ^{(𝐜)}_𝐩)²$.
  \end{rem}

  \begin{rem}
    The alternative presentation mentioned in~§\ref{ss:expi} may be
    specified in a similar way.
  \end{rem}

  \begin{rem}\label{rk:pilts}
    We have accounted for the presentation of the $π$-calculus through
    a reduction relation, but there is an important, alternative
    presentation based on a labelled transition system~\cite[Table~1.5
    p38]{DBLP:books/daglib/0004377}.  Let us explain why our framework
    cannot cover such a presentation. The problem is that it includes
    an input transition $a(x).P \xto{a(y)} P[x ↦ y]$ in which $y$ may
    be fresh. Letting $T$ denote the monad for syntax as
    in~§\ref{ss:expi}, the correct way to model the fresh case is to
    take as corresponding component of the target state functor $S₂$
    the $T$-module $T_𝐩^{(𝐜)}$, whose elements are processes $P$ with
    a fresh variable. Because the module $T_𝐩^{(𝐜)}$ is not of the
    shape $S₂∘T$, this goes beyond the framework of \transition
    monads.
\end{rem}

  \subsection{A register for Positive GSOS systems}
  Finally, let us recall that Positive GSOS rules have the shape
  \[\inferrule{
    … \\ \labredrule{a_{i,j}}{xᵢ}{y_{i,j}} \\ … \\ (i ∈ n, j ∈ nᵢ)}{\labredrule{c}{\ope(x₁,…,xₙ)}{e}}\rlap{\ ,}\]
  where the variables $xᵢ$ and $y_{i,j}$ are all distinct, $\ope ∈ O$ is
  an operation with arity $n$, and $e$ is an expression potentially
  depending on all the variables.

  We saw in~§\ref{ss:gsos} that each family of operations and rules yields
  \atransition monad where
  \begin{itemize}
  \item $ℙ=1$, because we are in an untyped setting,
  \item  $𝕊 = 1$ because states are terms,
  \item $T$ denotes the term monad,
  \item  $S₁(X) = X$, and
  \item $S₂(X) = 𝔸×X$, where $𝔸$ denotes the set of labels.
\end{itemize}

Let us now define a register $𝐆𝐒𝐎𝐒⁺$ for specifying positive
GSOS systems~\cite{GSOS}. This is a subregister of our register
$𝐑𝐞𝐠𝐓𝐫𝐚𝐧𝐬𝐌𝐧𝐝_{ℙ,𝕊}$, for untyped ($ℙ = 𝕊 = 1$) transition monads.
Let us recall that signatures in this register consist of tuples $(Σ_T,Σ₁,Σ₂,Σ_\Red)$,
where $T = \spec{Σ_T}$, $S₁ = \spec{Σ₁}$, $S₂ = \spec{Σ₂}$, and $Σ_\Red$
is a signature in the register $𝐑𝐞𝐠𝐓𝐫𝐚𝐧𝐬𝐒𝐭𝐫𝐮𝐜𝐭_{ℙ,𝕊}(T,S₁,S₂)$.

In order to describe this subregister, we have to describe its class
of signatures, and then assign to each such signature a tuple $(Σ_T,Σ₁,Σ₂,Σ_\Red)$
as above.

A signature of the register $𝐆𝐒𝐎𝐒⁺$ consists of
\begin{itemize}
\item 
 three sets $O$ (for operations), $𝔸$
 (for labels), and $R$ (for rules),
\item 
 for each element $o$ of $O$, a number $mₒ$ (the arity),
\item 
  for each rule,
  \begin{itemize}
  \item 
    an operation $o ∈ O$ (for the source of the conclusion),
    \item a label $c ∈ 𝔸$ (the label of the conclusion),
    \item  for each $1 ≤ i ≤ mₒ$,
      \begin{itemize}
      \item 
        a number $n_i$ (the number of premises for this argument),
        \item for each $1 ≤ j ≤ nᵢ$ , an element $a_{ij}$ of $𝔸$
        (for the label of the premise),
      \end{itemize}
\item a term $e$ in the syntax generated by $O$, potentially depending on
        $mₒ + ∑_i n_i$ variables.
  \end{itemize}
          \end{itemize}

We now describe the tuple $(Σ_T,Σ₁,Σ₂,Σ_\Red)$ associated to a signature as above:
\begin{itemize}
\item the signatures for both state functors have been given in
  Remark~\ref{rk:statefuns}: $S₁$ is specified by a single operation
  $𝐈 → Θ$, while $S₂$ is specified by a single operation $𝔸×𝐈 → Θ$;
\item the signature for the underlying monad is the family
  $(Θ^{m_o} → Θ)_{o ∈ O}$ (following §\ref{ss:regmon});
\item finally, each Positive GSOS rule yields a rule
  \[\inferrule{
  … \\ eᵢ \leadsto (a_{i,j},e_{i,j}) \\ … \\ (i ∈ mₒ, j ∈ nᵢ) }{ o(e₁,…,e_{mₒ})
\leadsto (c,E)}\] in the corresponding signature $Σ_\Red$ (in the register
$𝐑𝐞𝐠𝐓𝐫𝐚𝐧𝐬𝐒𝐭𝐫𝐮𝐜𝐭_{ℙ,𝕊}(T,S₁,S₂)$).
\end{itemize}

\begin{rem}
  In a bit more detail:
  \begin{itemize}
  \item The metavariable module is $V = T^{mₒ + ∑_{i ∈ mₒ} nᵢ}$, of
    which a typical element is a tuple
    $e ≔ ((eₖ)_{k ∈ mₒ}, (e_{i,j})_{i ∈ mₒ, j ∈ nᵢ}) ∈ T^{mₒ + ∑ᵢ
      nᵢ}(X)$.
\item 
There are $∑_{i ∈ mₒ} nᵢ$ premises.
\item The $(i,j)$th premise maps any tuple $e$ to
  $(eᵢ,(a_{i,j},e_{i,j}))$.
\item The conclusion maps it to $(o(e₁,…,e_{mₒ}),(c,E(e)))$, where
  $E∶ V → T$ is the target expression viewed as a $T$-module morphism.
\end{itemize}
\end{rem}

\subsection{The differential \lam-calculus}\label{ss:difftrans}
Recall from~§\ref{ss:lambdadiff} that the differential $λ$-calculus syntax
is defined by
  \begin{displaymath}
    \begin{array}[t]{rcll}
      e,f,g & \Coloneqq & x ｜ λx.e ｜ e⟨U⟩ ｜ D e ⋅ f & \mbox{(terms)} \\
      U,V & \Coloneqq & 0 ｜ e + U & \mbox{(multiterms)}\rlap{,}
    \end{array}
  \end{displaymath}
  where terms and multiterms are considered equivalent up to the
  following equations.
  \begin{mathpar}
    e + e' + U = e' + e + U \and D (D e · f) · g = D (D e · g) · f
  \end{mathpar}

  Based on unary multiterm substitution $e[x↦U]$ and partial
  derivation $\frac{∂e}{∂x} ⋅ U$, the transition relation is defined
  as the smallest context-closed relation satisfying the rules below,
    \begin{center}
      \hfil $\inferrule{}{(λx.e)⟨U⟩ → e[x ↦ U]}$ \hfil
      $\inferrule{}{D(λx.e)⋅f → λx.\left (\frac{∂e}{∂x}⋅f \right )}$
  \end{center}
  where $λ$ is linearly extended to multiterms: $ λx.(e₁+…+eₙ)$ is
  notation for $λx.e₁ + … + λx.eₙ$.

  We saw that this all forms \atransition monad with one placetaker
  type and one transition type ($ℙ = 𝕊 = 1$).
  
  Let us now design a signature specifying this \transition monad.
  Such a signature consists of signatures for the state functors and
  monad, together with a signature for the transition module.

  \mypar{Signature for the placetaker monad} Recalling that $\oc$
  denote the finite multiset functor, the monad $T$ of differential
  $λ$-calculus is specified by the signature $S$ with operations 
  \begin{mathpar}
    λ∶ Θ^{(1)} → Θ \and {-}⟨{-}⟩∶ Θ×\oc Θ → Θ \and D{-}·{-}∶ Θ×Θ → Θ
  \end{mathpar}
  and one equation:
  \begin{equation}
    D (D e · f) · g ≡ D (D e · g) · f\rlap{.}\label{eq:DD}
  \end{equation}
  By taking $\oc Θ$ as type for the second argument of application, we
  directly identify multiterms as finite multisets, which explains why
  we do not need any further equation for enforcing order irrelevance
  of $+$.

  \mypar{Signature for the state functors} Our state functors are
  $S₁ = \Id$ and $S₂ = {\oc}$, which are specified by easy signatures
  in the sense of Remark~\ref{rk:statefuns}.

  \mypar{Signature for transitions} Specifying the transition rules
  requires the following lemma.
  \begin{lem}\label{lem:deriv}
    There exist $T$-module morphisms
    \begin{center}
      $σ∶ T^{(1)}×\oc T → \oc T$ \hfil and \hfil
      $δ∶ T^{(1)}×\oc T → \oc T ^{(1)}$,
  \end{center}
  respectively called \alert{unary multiterm substitution} and
  \alert{partial derivation}, satisfying the equations
  of~\cite[Definition~6.3 and~6.4]{Iouri}.
  \end{lem}
  \begin{proof}[Proof sketch]
    By the explicit description of Corollary~\ref{cor:modmnd}, $T$ is a
    quotient of the initial $Σ_S$-monoid $Z$ by
    Equation~\eqref{eq:DD}. We thus first define both operations from
    $Z^{(1)}$ by induction following~\cite[Definition~6.3
    and~6.4]{Iouri}, then check that both obtained morphisms
    coequalise the relevant parallel pair to extend them to morphisms
    from $T^{(1)}$. Finally, we check that both obtained morphisms are
    indeed module morphisms, which follows
    from~\cite[Lemma~6.10]{Iouri}.  
\end{proof}

  Now that both auxiliary operations $σ$ and $δ$ are defined, the main
  transition rules are
\begin{mathpar}
 \inferrule{ }{λ(t)⟨U⟩ ↝ σ(t,U)}
  \and
  \inferrule{ }{D(λ(t))·u ↝ λ(δ(t,u))}\ ,
\end{mathpar}
where we implicitly coerce $λ∶ T^{(1)} → T$ into a morphism
$\oc T^{(1)} → \oc T$.

Furthermore, because reduction in the differential $λ$-calculus is
context-closed, we need to include the following congruence
rules, detailed in~\cite[Definition 6.18]{Iouri}:
  \begin{mathpar}
    \inferrule{t ↝ U}{λ(t) ↝ λ(U)}
    \and
    \inferrule{t ↝ U}{D t·s ↝ D U·s}
    \and
    \inferrule{s ↝ U}{D t·s ↝ D t·U}
    \\
    \inferrule{t ↝ U}{t⟨S⟩ ↝ U⟨S⟩}
    \and
    \inferrule{t ↝ U}{s⟨t + V⟩ ↝ s⟨U + V⟩}
  \end{mathpar}
  where, as for $λ$ above, we implicitly coerce some module morphisms
  to $T$ into module morphisms to $\oc T$, also lifting some of their
  arguments from $T$ to $\oc T$, namely we use
  \begin{mathpar}
    {\begin{array}[t]{rcl}
        D(-)·{-}∶ \oc T × \oc T & → & \oc T \\
        (∑ᵢ eᵢ, ∑ⱼ fⱼ) & ↦ & ∑_{i,j} Deᵢ·fⱼ
      \end{array}}
    \and
    {\begin{array}[t]{rcl}
       {-}⟨{-}⟩∶ \oc T × \oc T & → & \oc T \\
       (∑ᵢ eᵢ, ∑ⱼ fⱼ) & → & ∑ᵢ eᵢ⟨∑ⱼfⱼ⟩       
      \end{array}}
    \and
    {\begin{array}[t]{rcl}
    {+}∶ \oc T × \oc T & → & \oc T \\
  (∑ᵢ eᵢ, ∑ⱼ fⱼ) & → & ∑ᵢ eᵢ + ∑ⱼfⱼ.
    \end{array}}
  \end{mathpar}
  \begin{rem}\ \hfill
    \begin{itemize}
    \item The first lifting is only used in cases where one of its
      arguments is a singleton multiset, i.e.,
      we need $DU·t$ and $Dt·u$, but not $DU·T$.
    \item In the second case, only the first argument needs lifting,
      i.e., we have $(e₁ + … + eₙ)⟨U⟩ = e₁⟨U⟩ + … + eₙ⟨U⟩$.
    \item The last lifting is in fact mere multiset union.
  \end{itemize}
\end{rem}

\section{Conclusion and perspectives}\label{s:conclu}
We have introduced \transition monads as a generalisation of reduction
monads, and demonstrated that they cover relevant new examples.  We
have introduced a register of signatures for specifying them.  Let us
briefly assess the scope of our register for \transition monads, or
more generally of the notion of \transition monad itself.

All combinations of call-by-value vs.\ call-by-name, small-step vs.\
big-step, or simply-typed vs.\ untyped variants of $λ$-calculus should
be handled smoothly.

Simple imperative languages like
IMP~\cite[Chapter~2]{DBLP:books/daglib/0070910} may also be organised
into transition monads, at least  in a trivial way since, in the absence
of first-class functions, reduction is generally defined on closed
programs. 

Languages whose transition rules involve some kind of evaluation
contexts, such as ML~\cite[Chapter~10]{ATAPL}, $λ$-calculi with
\texttt{let rec}~\cite{DBLP:journals/apal/AriolaB02}, or the
substitution calculus~\cite{DBLP:conf/rta/Accattoli19}, should also fit
into the framework, although with a bit more work.

A first class of languages which clearly cannot be organised naturally
as \transition monads is those whose transitions involve non-free
modules, as noticed in Remark~\ref{rk:pilts}.  Extending the register
$𝐑𝐞𝐠[𝐒𝐞𝐭^ℙ,𝐒𝐞𝐭^𝕊]_f$ of~§\ref{s:regfun} to cover such examples seems
at hand.

Another, more problematic class of languages is those with advanced
type systems, e.g., polymorphic or dependent types. Covering such
examples is one of our longer-term goals.

Another limitation of our approach is the weakness of the induced
induction principle.  As discussed in~§\ref{s:intro}, this is the
price to pay for its simplicity. What is missing, in comparison with
previous work like~\cite{AHLM19}, is a kind of Grothendieck
construction for signatures/registers. This works smoothly in most
examples.  However, in cases like the differential $λ$-calculus, this
would require extending the definition of unary multiterm substitution
and partial derivation (Lemma~\ref{lem:deriv}) to all models of the
syntax. And this appears to leave some design choices open, which
might be a reflection of the diversity of categorical semantics for
differential
$λ$-calculus~\cite{DBLP:journals/mscs/BluteCS06,DBLP:conf/tlca/Fiore07,DBLP:journals/mscs/Ehrhard18,Christine}.

In the longer term, we plan to refine our register in a way ensuring
that the generated transition system satisfies important properties
like congruence of useful behavioural equivalences, confluence, or type
soundness. In this direction, a result on congruence of applicative
bisimilarity for a simpler register has recently been obtained by
Borthelle et al.~\cite{BHL}.

\undef\abbrevbib
\bibliographystyle{alphaurl}
\bibliography{opmonades}

\end{document}